\documentclass[11pt,english,american,fleqn]{article}

\usepackage{amsmath,amsfonts,amssymb,amsthm,mathrsfs}

\usepackage{bbm}

\usepackage{setspace}

\usepackage{enumerate}
\usepackage{bm}
\usepackage{bbm}
\usepackage{verbatim}
\usepackage{xargs} 
\usepackage{authblk}
\usepackage[square,sort,comma,numbers]{natbib}
\usepackage{chngcntr}
\usepackage[font=small]{caption}
\usepackage[hypertexnames=false,hidelinks]{hyperref}

\theoremstyle{plain}\newtheorem{theorem}{Theorem}[section]
\theoremstyle{plain}\newtheorem{lemma}[theorem]{Lemma}
\theoremstyle{plain}\newtheorem{corollary}[theorem]{Corollary}
\theoremstyle{plain}
\theoremstyle{plain}\newtheorem{proposition}[theorem]{Proposition}
\theoremstyle{definition}
\theoremstyle{remark}\newtheorem{remark}{Remark}
\theoremstyle{definition}\newtheorem{def:and:lemma}[theorem]{Definition and Lemma}
\theoremstyle{plain}\newtheorem*{conjecture}{Conjecture}

\numberwithin{equation}{section}
\counterwithin{remark}{section}
\counterwithin{figure}{section}

\newcounter{remarks}
\newcommand{\D}{\textnormal{d}}
\newcommand{\lsp}{\big \langle }
\newcommand{\rsp}{\big \rangle }

\newcommand{\im}{\operatorname{Im}}
\newcommand{\re}{\operatorname{Re}}

\newcommand{\sno}{\vert\hspace{-1pt}\vert}

\newcommand{\I}{\big|}
\newcommand{\tc}{\textcolor{white}{.}}

\newcommand{\op}{\text{\tiny{$\rm op$}}}
\newcommand{\HS}{\text{\tiny{$\rm HS$}}}
\newcommand{\p}{\text{\tiny{$L^p$}}}
\newcommand{\q}{\text{\tiny{$L^q$}}}
\newcommand{\2}{\text{\tiny{$L^2$}}}
\newcommand{\1}{\text{\tiny{$L^1$}}}
\newcommand{\su}{\text{\tiny{$L^\infty$}}}
\newcommand{\Fock}{\text{\tiny{$\mathcal F$}}}
\newcommand{\h}{\text{\tiny{$\mathscr H$}}}
\newcommand{\PP}{P_\psi}
\newcommand{\QQ}{Q_\psi}

\usepackage[colorinlistoftodos,prependcaption,textsize=tiny]{todonotes}

\begin{document}

\bibliographystyle{alpha}

\title{\LARGE Optimal parabolic upper bound for the energy-momentum relation of a strongly coupled polaron}

\author{David Mitrouskas\thanks{Institute of Science and Technology Austria (ISTA), Am Campus 1, 3400 Klosterneuburg, Austria.\\ Email: \texttt{david.mitrouskas@ist.ac.at}}, Krzysztof My\' sliwy\thanks{Institute of Science and Technology Austria (ISTA), Am Campus 1, 3400 Klosterneuburg, Austria.\\ Email: \texttt{krzysztof.mysliwy@ist.ac.at}} \phantom{i}and Robert Seiringer\thanks{Institute of Science and Technology Austria (ISTA), Am Campus 1, 3400 Klosterneuburg, Austria.\\ Email: \texttt{robert.seiringer@ist.ac.at}}}

\date{March 23, 2022}

\maketitle

\frenchspacing

\begin{spacing}{1.15} 
 
\begin{abstract}
We consider the large polaron described by the Fr\"ohlich Hamiltonian and study its energy-momentum relation defined as the lowest possible energy as a function of the total momentum. Using a suitable family of trial states, we derive an optimal parabolic upper bound for the energy-momentum relation in the limit of strong coupling. The upper bound consists of a momentum independent term that agrees with the predicted two-term expansion for the ground state energy of the strongly coupled polaron at rest, and a term that is quadratic in the momentum with coefficient given by the inverse of twice the classical effective mass introduced by Landau and Pekar.
\end{abstract}

\allowdisplaybreaks

\tableofcontents

\section{Introduction}

\subsection{The Model}

The large polaron provides an idealized description for the motion of a slow band electron through a polarizable crystal. The analysis of the polaron is a classic problem in solid state physics that first appeared in 1933 when Landau introduced the idea of self-trapping of an electron in a polarizable environment \cite{Landau1933}. Since it provides a simple model for a particle interacting with a nonrelativistic quantum field, the polaron has been of interest also in field theory and mathematical physics. In particular the strong coupling theory of the polaron and Pekar's adiabatic approximation have been the source of interesting and challenging mathematical problems. 

Following H. Fr\"ohlich \cite{Froehlich1954} the Hamiltonian of the model acts on the Hilbert space
\begin{align}
\mathscr{H} \, =\,  L^2(\mathbb R^3,\D x) \otimes \mathcal F,
\end{align}
with $\mathcal F$ the bosonic Fock space over $L^2(\mathbb R^3)$, and is given by
\begin{align}\label{eq: Froehlich Hamiltonian}
H_\alpha \,  =\,  -\Delta_x + \alpha^{-2} \mathbb N + \alpha^{-1}\phi(h_x).
\end{align}
Here $x\in \mathbb R^3$ is the coordinate of the electron, $\mathbb N$ denotes the number operator on Fock space, and the field operator $\phi(h_x) =a^\dagger(h_x) + a(h_x)$ with coupling function
\begin{align}
\label{eq: def of h_x(y)}
h_x(y) \,  =\,  - \frac{1}{ 2\pi^2 \vert x - y\vert^2 }
\end{align}
accounts for the interaction between the electron and the quantum field. The creation and annihilation operators satisfy the usual canonical commutation relations
\begin{align}
\big[ a(f  ), a^\dagger( g ) \big] \, =\,  \lsp f | g  \rsp_\2, \quad  \big[ a( f ), a( g ) \big]\,  =\,  0 .
\end{align}
Since we set $\hbar = 1$ and the mass of the electron equal to $1/2$, the only free parameter is the coupling constant $\alpha>0$.

By rescaling all lengths by a factor $1/\alpha$, one can show that $\alpha^2 H_\alpha$ is unitarily equivalent to the Hamiltonian
\begin{align}
H_{\alpha}^{\rm {Polaron}} \,  =\, -\Delta_x + \mathbb N + \sqrt \alpha \phi(h_x),
\end{align} 
which is more common in the polaron literature and also explains why $\alpha \to \infty$ is called the strong coupling limit.

The Fr\"ohlich Hamiltonian defines a translation invariant model, i.e., it commutes with the total momentum operator,
\begin{align}
[H_\alpha, -i\nabla_x + P_f] \, = \, 0
\end{align}
where $P_f=\D \Gamma(-i\nabla)$ denotes the momentum operator of the phonons. This allows the definition of the energy-momentum relation $E_\alpha(P)$ as the lowest possible energy of $H_\alpha$ when restricted to states with total momentum $P\in \mathbb R^3$. To this end, it is convenient to switch to the Lee--Low--Pines representation
\begin{align}
H_\alpha (P) & \, =\,  (P_f-P)^2 + \alpha^{-2} \mathbb N + \alpha^{-1} \phi (h_0), 
\end{align}
where $H_\alpha(P)$ acts on the Fock space only \cite{LeeLowPines}. The Fr\"ohlich Hamiltonian $H_\alpha$ is unitarily equivalent to the fiber decomposition $\int^{\oplus }_{\mathbb R^3} H_\alpha(P) \D P$, which follows easily from transforming $H_\alpha$ with $e^{iP_fx}$ and diagonalizing the obtained operator in the electron coordinate. The energy-momentum relation is then defined as the ground state energy of the fiber Hamiltonian,
\begin{align}
E_\alpha(P) & \,  =\, \inf \sigma (H_\alpha(P)),
\end{align}
which by construction satisfies $E_\alpha( R P ) = E_\alpha (P)$ for all rotations $R\in \textnormal{SO}(3)$. It is known that $E_\alpha(0) \le E_\alpha(P)$ and hence $E_\alpha(0)= \inf \sigma(H_\alpha) $ (in fact it is expected that $E_\alpha(0) < E_\alpha(P)$ for all $P\neq 0$ \cite{DybalskyS2020}). Further properties, such as the domain of analyticity, existence of ground states and the value of the bottom of the continuous spectrum, were analyzed in \cite{Froehlich1974,Spohn1988,Moeller2006,Gerlach91,JonasPHD17}.

The aim of this work is to analyze the quantitative behavior of the energy-momentum relation for large coupling $\alpha \to \infty$. Our main result provides an upper bound for $E_\alpha (P)$. The upper bound consists of a momentum independent part coinciding with the optimal upper bound for the ground state energy of the strongly coupled polaron at rest, and a momentum dependent part. In more detail, the momentum independent part is given by the classical Pekar energy and the corresponding quantum fluctuations that are described by the energy of a system of harmonic oscillators with frequencies determined by the Hessian of the corresponding classical field functional. This part agrees with the expected asymptotic form of $E_\alpha(0)$, see \eqref{bogcor}. The momentum dependent part, on the other hand, describes the energy of a free particle with mass $M(\alpha) = \frac{2 \alpha^4}{3} \int|\nabla \varphi|^2$, where $\varphi$ denotes the self-consistent polarization field, which coincides with the classical polaron mass introduced by Landau and Pekar \cite{Landau1948}, see \eqref{eq: eff mass conjecture}. As will be explained in Section \ref{sec: motivation}, our result confirms the heuristic picture of the polaron (the electron and the accompanying classical field) as a free quasi-particle with largely enhanced mass. To our best knowledge, the upper bound we present in this work is the first rigorous result about the connection between the energy-momentum relation $E_\alpha(P)$ and the classical polaron mass $M(\alpha)$. 

Starting from the works in the 30's and 40's \cite{Landau1933,Landau1948,Froehlich1933} there has been a large number of publications in the physics literature that studied the ground state energy $E_\alpha(0)$ and the effective mass, that is, the inverse curvature of $E_\alpha(P)$ at $P=0$. For a comprehensive summary of the earlier results, we refer to \cite{Mitra1987}. More recent developments are reviewed in \cite{DevreeseA2010}. Mathematically rigorous results for the leading order asymptotics of $E_\alpha(0)$, for $\alpha$ large, were obtained by Lieb and Yamazaki \cite{Lieb1958} (with non-matching upper and lower bounds) and by Donsker and Varadhan \cite{Donsker1983} as well as Lieb and Thomas \cite{Lieb1997}. The effective mass has been studied in \cite{Spohn1987,DybalskyS2020,FeliciangeliRS21,Lieb2020,LiebSeiringer2014,Betz2022}. Further improvements have been obtained for confined polarons or polaron models with more regular interaction \cite{FrankS2021,FeliciangeliS21,Mysliwy2021}. For completeness, let us also mention recent progress in the understanding of the polaron path measure \cite{Mukherjee19,Betz2021} as well as the increased interest in the analysis of the Schr\"odinger time evolution of strongly coupled polarons \cite{Griesemer2017,LeopoldMRSS2020,LeopoldRSS2019,Mitrouskas21,FeliciangeliRS20,FrankG2017,
FrankS2014}.

\subsection{Pekar functionals}

The semiclassical theory of the polaron has been introduced by Pekar \cite{Pekar54}. It arises naturally in the context of strong coupling, based on the expectation that the electron and the phonons are adiabatically decoupled, similarly as the electrons are adiabatically decoupled from the heavy nuclei in the famous Born–Oppenheimer theory \cite{Born27,Born1954}. With this in mind, one can minimize the Fröhlich Hamiltonian over product states of the form 
\begin{align}
\Psi_{u,v} \,  =\, u \otimes e^{a^\dagger (\alpha v ) } \Omega  
\end{align}
where $u\in H^1(\mathbb R^3)$ is a normalized electron wave function, $\Omega = (1,0,0,\ldots)$ the Fock space vacuum and $ e^{a^\dagger (\alpha v) } \Omega$ the coherent state, up to normalization, that is associated with a classical field $\alpha v\in L^2(\mathbb R^3)$. A simple computation leads to the Pekar energy functional
\begin{align}
\mathcal G(u,v) \, =\,  \frac{\lsp \Psi_{u,v} | H_\alpha   \Psi_{u,v} \rsp_\h }{\lsp \Psi_{u,v} | \Psi_{u,v} \rsp_\h } \, =\,   \lsp  u | (-\Delta + V^{v} ) u \rsp_\2 + \sno v \sno^2_\2
\end{align}
with polarization potential
\begin{align}\label{eq: def of effective potential}
V^{v}(x) \,  =\, - 2 \re \lsp v | h_x \rsp_\2 \,  =\,  - \re \int \frac{ v (y)}{\pi ^2\vert x-y\vert^2 } \D y.
\end{align}
By completing the square, one can further remove the field variable and obtain the energy functional for the electron wave function,
\begin{align}\label{eq: electronic pekar functional}
\mathcal E(u) & \,  =\, \inf_{v \in L^2} \mathcal G(u,v) \,  =\, \int \vert u (x) \vert^2 \D x - \frac{1}{4\pi}\iint \frac{ \vert u(x)\vert^2 \vert u(y)\vert^2 }{ \vert x-y\vert } \D x \D y ,
\end{align}
which is known \cite{Lieb1977} to admit a unique rotational invariant minimizer $\psi >  0$ (the minimizing property is unique only up to translations and multiplications by a constant phase). Alternatively, one can minimize the Pekar energy functional w.r.t. the electron wave function first. This leads to the classical field functional
\begin{align}
\mathcal F(v) & \,  =\, \inf_{\sno u\sno_\2 = 1} \mathcal G(u,v) \,  =\, \inf \text{spec}\, (-\Delta + V^{v}) + \sno v\sno^2_\2 \label{eq: phonon pekar functional}
\end{align}
whose unique rotational invariant minimizer is readily shown to be
\begin{align}\label{eq: optimal phonon mode}
\varphi(z)  \,  =\, -  \lsp \psi \I h_\cdot (z)  \psi \rsp_\2 \,  =\,  \int \frac{\vert \psi (y)\vert ^2}{2\pi^2 \vert z-y\vert^2}\D y.
\end{align}
The corresponding classical ground state energy is called the Pekar energy
\begin{align}\label{eq: Pekar energy}
e^{\rm Pek} \,  =\, \mathcal E(\psi) \,  =\, \mathcal F(\varphi),\quad e^{\rm Pek}<0,
\end{align}
and by the variational principle it provides an upper bound for $\inf \sigma (H_\alpha)$. The validity of Pekar's ansatz was rigorously verified by Donsker and Varadhan \cite{Donsker1983} who proved that $ \lim_{\alpha \to \infty} \inf \sigma(H_\alpha) =  e^{\rm Pek}$ and subsequently by Lieb and Thomas \cite{Lieb1997} who added a quantitative bound for the error by showing that
\begin{align}\label{eq: Lieb Thomas bound}
\inf \sigma (H_\alpha^{\rm F}) \, \ge\,  e^{\rm Pek} + O(\alpha^{ - 1 / 5 }).
\end{align}
Given the potential $V^\varphi$ for the field $\varphi$, one can define the Schr\"odinger operator
\begin{align}
h^{\rm Pek} \,  =\, -\Delta + V^\varphi(x) - \lambda^{\rm pek}, \quad \lambda^{\rm Pek} \,  =\,  e^{\rm Pek} - \sno \varphi\sno_\2^2
\end{align}
with $\lambda^{\rm Pek} = \inf \sigma(-\Delta + V^\varphi(x))<0$ and $\psi$ the corresponding unique ground state. It follows from general arguments for Schr\"odinger operators that $h^{\rm Pek}$ has a finite spectral gap above zero, and thus the reduced resolvent
\begin{align}\label{eq: def of resolvent}
 R \,  =\,\QQ (h^{\rm Pek})^{-1} \QQ \quad \text{with} \quad \QQ \,  =\, 1-\PP,\quad \PP \,  =\, |\psi \rangle \langle \psi | ,
\end{align}
defines a bounded operator ($\PP$ denotes the orthogonal projection onto the state $\psi$).

The last object to be introduced in this section is the Hessian $H^{\rm Pek}$ of the energy functional $\mathcal F$ at its minimizer $\varphi$, defined by
\begin{align}\label{eq: definition of Hessian}
\lsp v  \I H^{\rm Pek}  v \rsp_\2 \, = \, \lim_{\varepsilon\to 0 } \frac{1}{\varepsilon^2 }\big( \mathcal F (\varphi + \varepsilon v )-\mathcal F (\varphi ) \big) \quad \forall v \in L^2 (\mathbb R^3).
\end{align}
In the following lemma we collect some important properties of $H^{\rm Pek}$.

\begin{lemma}\label{lem: Hessian} 
The operator $H^{\rm Pek}$ has integral kernel 
\begin{align}\label{eq: Hessian kernel}
H^{\rm Pek}(y,z) \,  =\,  \delta(y-z) -4 \re \lsp \psi \I h_\cdot(y) R  h_\cdot(z) \psi \rsp_\2
\end{align}
and satisfies the following properties.
\begin{itemize}
\item[(i)] $0\le H^{\rm Pek} \le  1$
\item[(ii)]  $\textnormal{Ker} H^{\rm Pek} = \textnormal{Span}\{\partial_i \varphi \, : \, i = 1,2,3 \}$
\item[(iii)] $H^{\rm Pek} \ge \tau >0$ when restricted to $(\textnormal{Ker} H^{\rm Pek})^\perp$ \item[(iv)] ${\rm Tr}_{L^2}(1- \sqrt{H^{\rm Pek}}) < \infty $.
\end{itemize}
\end{lemma}
The proof of the lemma, in particular Item (ii), is based on the analysis of the Hessian of the energy functional $\mathcal E$ \cite{Lenzmann09}. The details are given in Section \ref{Sec: Remaining Proofs}.

\subsection{Motivation and goal of this work\label{sec: motivation}}

In this work, we are interested in the behavior of the energy-momentum relation $E_{\alpha}(P)$ for large values of the coupling $\alpha$. In general, $E_{\alpha}(P)$ is expected to interpolate between two distinct regimes (see for instance \cite{Gerlach08,Gerlach03,Whitfield65,Spohn1988}): The \emph{quasi-particle regime} and the \emph{radiative regime}. The former corresponds to small momenta, and the expectation is that the system behaves effectively like a free particle with energy 
\begin{align}\label{eq: quasi-particle energy}
E^{\rm eff}_\alpha(P) = E_\alpha(0) + \frac{P^2}{2M^{\rm eff}(\alpha)}
\end{align}
where the effective mass is determined by the inverse curvature of $E_{\alpha}(P)$ at $P=0$ (which is known to be well-defined),
\begin{align}\label{eq: effective mass conjecture}
M^{\rm eff}(\alpha) \,  : =\, \frac{1}{2}  \lim_{P\to 0} \bigg( \frac{E_\alpha(P) - E_\alpha(0) }{P^2} \bigg)^{-1}.
\end{align}
It is easy to verify that $M^{\rm eff} (\alpha) \ge 1 / 2$ (the mass of the electron in our units), and one can further show that the inequality is strict if $\alpha >0$, so that the emerging quasi-particle is heavier than the bare electron. The heuristic idea is that the electron drags along a cloud of phonons when it moves through the crystal and thus appears to be heavier than it would be without the interaction. The radiative regime, on the other hand, describes a polaron at rest and an unbound/radiative phonon carrying the total momentum $P$. It is expected to be valid for large momenta and it is characterized by a flat energy-momentum relation that equals or approaches the bottom of the continuous spectrum \cite{Moeller2006},
\begin{align}\label{eq: cont bottom}
\inf\sigma_{\rm cont}(H_\alpha(P)) = E_\alpha(0) + \alpha^{-2}.
\end{align}
The two regimes cross at $|P| = P_{\rm c}(\alpha) : =\sqrt{2M^{\rm eff}(\alpha)} / \alpha$ which marks a characteristic momentum scale of the polaron. While the quasi-particle picture is expected to be accurate for $|P|\ll P_{\rm c}(\alpha)$, the radiative regime should hold for $|P|\gtrsim P_{\rm c}(\alpha)$ (see also Remark \ref{rem: radiative regime} below). Between the two regimes there is no concrete prediction for the behavior of $E_\alpha(P)$. A schematic plot is provided in Figure \ref{fig:my-label}.
\begin{figure}[t!]\label{Figure}
        \center{ \includegraphics[width=0.8\textwidth,height=4.75cm]
        {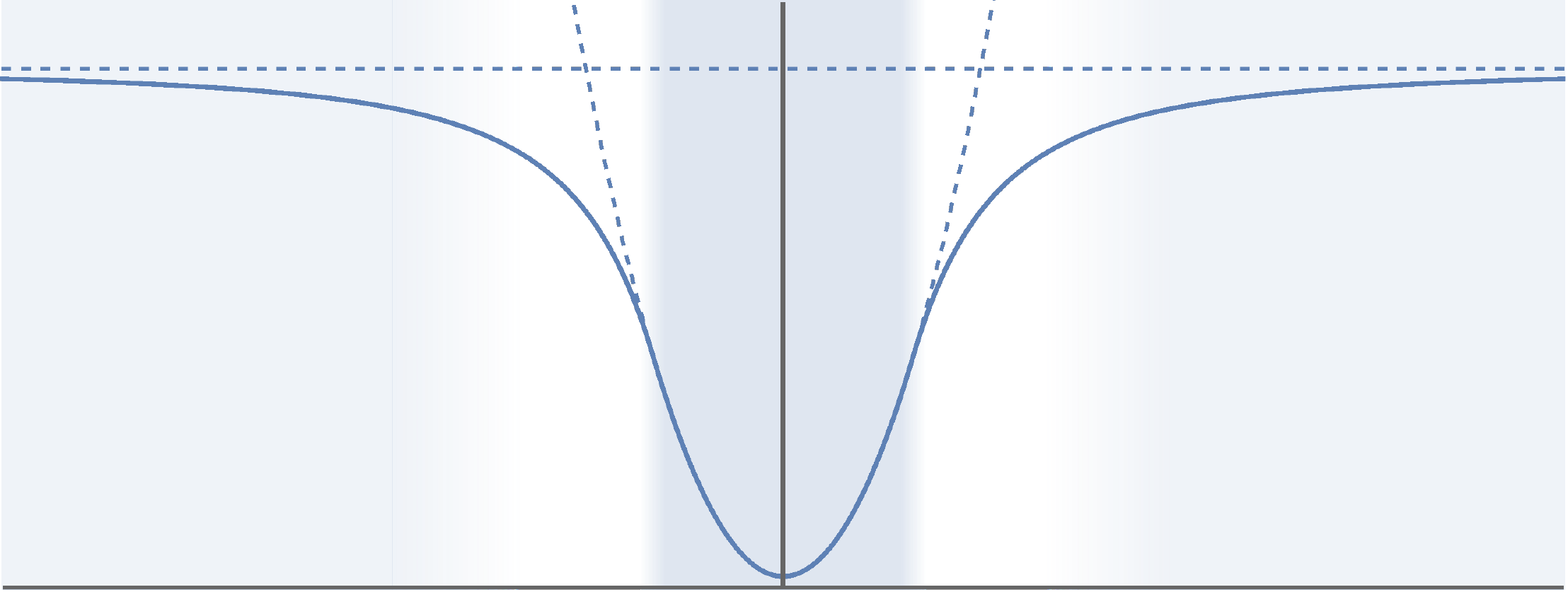} } 
        \caption{\label{fig:my-label}\small{The energy-momentum relation $E_\alpha(P)$ is expected to have two characteristic regimes: The parabolic quasi-particle regime for small momenta (dark area) and the radiative regime for large momenta (light area). For the transition between the two there is no concrete prediction. The dashed lines denote the quasi-particle energy \eqref{eq: quasi-particle energy} and the bottom of the continuous spectrum \eqref{eq: cont bottom}. Their intersection defines the momentum scale $P_c(\alpha)$ that is proportional to $\alpha$ for large coupling.\label{figure 1}}} 
\end{figure}

One aspect of this work is to show that the quasi-particle picture is mathematically rigorous, insofar as it provides a parabolic upper bound on $E_{\alpha}(P)$ that coincides with the expected form of the quasi-particle energy in the limit of large coupling. Since the quasi-particle energy \eqref{eq: quasi-particle energy} is determined by the values of $E_\alpha(0)$ and $M^{\rm eff}(\alpha)$, it is instructive to recall two long-standing open conjectures concerning their behavior for $\alpha \to \infty$. As explained in the previous section, the phonon field behaves classically for large coupling, and thus it is expected that $M^{\rm eff} (\alpha)$ should asymptotically tend to the expression that follows from the corresponding semiclassical counterpart of the problem. This semiclassical theory of the effective mass was introduced by Landau and Pekar in 1948 \cite{Landau1948}, and, based on this work (see also \cite{Spohn1987,FeliciangeliRS21}), it is conjectured that
\begin{align}\label{eq: eff mass conjecture}
\lim_{\alpha\to \infty} \frac{M^{\rm eff}(\alpha)}{\alpha^4} \,  =\, M^{\rm LP} \quad \text{with} \quad M^{\rm LP}\,  =\, \frac{2}{3}\sno \nabla \varphi \sno^2_\2.
\end{align}
Although this problem is many decades old, the best available rigorous result is that $M^{\rm eff} (\alpha )$ is divergent \cite{Lieb2020}, with a recent proof that it diverges at least as fast as $\alpha^{2/5}$ \cite{Betz2022}. Regarding the ground state energy $E_\alpha(0)$ the prediction from the physics literature (see e.g. \cite{Allcock65,Miyake1976,Tjablikow54,Gross76}) is that 
\begin{equation}\label{bogcor}
E_{\alpha}(0) \, =\,  e^{\rm Pek} + \frac{1}{2\alpha^2}{\rm Tr}_{L^2} (\sqrt{H^{\rm Pek}} - 1 ) + O(\alpha^{-2-\delta}) \quad \text{as} \quad \alpha\to \infty
\end{equation}
for some $\delta >0$ (in fact it is predicted that $\delta=2$ \cite{Gross76}). Compared to the semiclassical expansion this includes a subleading correction of order $\alpha^{-2}$, which we call the \emph{Bogoliubov energy}, and which arises from quantum fluctuations of the field around its classical value. For a nice heuristic derivation of this correction, we recommend the study of \cite{Miyake1976}. Now inserting \eqref{eq: eff mass conjecture} and \eqref{bogcor} into \eqref{eq: quasi-particle energy}, and based on the expectation that the quasi-particle regime is restricted to $|P| \ll \sqrt{2M^{\rm eff}(\alpha)} / \alpha \sim \alpha$, it is clear that the Bogoliubov energy needs to be taken into account in order to see the quasi-particle energy shift given by $P^2/( 2\alpha^4M^{\rm LP}) \ll \alpha^{-2}$. Mathematically, the validity of \eqref{bogcor} has been established only for confined polaron models \cite{FrankS2021,FeliciangeliS21}. The corresponding upper bound for the unconfined model is a corollary of our main result. 

As a summary of the above we arrive at the following claim.

\begin{conjecture}\label{conj: semiclassical expansion} Let $M^{\rm LP}$ be the Landau--Pekar mass defined in \eqref{eq: eff mass conjecture}. There exists a continuous function $f:[0,\infty)\rightarrow [0,\infty)$, satisfying $f(s) \to 1$ as $s\to \infty$ and
\begin{align}
 f(s) &\, =\, \frac{s}{2M^{\rm LP}}+O(s^2)\quad \text{as} \quad s \to 0,
\end{align} 
such that for all $P\in \mathbb R^3$ 
\begin{equation}\label{eq: conjecture}
  \lim_{\alpha\rightarrow\infty} \alpha^2\left(E_{\alpha}(\alpha P)-e^{\rm Pek}-\frac{1}{2\alpha^2} {\rm Tr}_{L^2}\big(\sqrt{H^{\rm Pek}}-1\big)\right)\, =\, f(P^2). 
\end{equation}
\end{conjecture}
Our main result, Theorem \ref{theorem: main estimate} below, provides an upper bound for $E_\alpha(\alpha P)$ that is compatible with the conjecture in the quasi-particle regime. To be more precise, our result implies that the left side of \eqref{eq: conjecture}, with the limit replaced by the $\limsup$, is bounded from above by $P^2 /(2M^{\rm LP})$ for all $P\in \mathbb R^3$. This shows that the corrections to the quasi-particle energy are always negative, a conclusion that is not entirely obvious a priori. 
\begin{remark}
An immediate consequence of the conjecture would be that
\begin{align}\label{pseudomass}
  \frac{1}{2}\lim_{P\rightarrow 0} \lim_{\alpha\to \infty} \alpha^2\bigg( \frac{E_\alpha(\alpha P) - E_\alpha(0) }{P^2} \bigg)^{-1} \, =\, M^{\rm LP}
\end{align}
which is to be compared with \eqref{eq: eff mass conjecture} where the limits are taken in reversed order. 
\end{remark}
\begin{remark} \label{rem: radiative regime}
Even though our analysis is focused on the quasi-particle regime, let us mention an interesting problem concerning the radiative regime. The question is whether $E_\alpha(P)$ enters the continuous part of the spectrum, i.e. whether the spectral gap closes at some finite momentum, or not. The answer may in fact depend on the dimension and possibly also on the value of $\alpha$. It is known that in two dimensions $E_\alpha(P)$ remains an isolated eigenvalue for all $P$, meaning that the curve approaches $\inf\sigma_{\rm cont}(H_\alpha(P))$ only in the limit $|P| \to \infty$ \cite{Spohn1988}. To our knowledge in three dimensions the question is not completely settled. While for small momenta it is known that $E_\alpha(P)$ corresponds to a simple eigenvalue \cite{Spohn1988}, there is indication from results obtained for weak coupling that $E_\alpha(P)$ agrees with the bottom of the continuous spectrum when $|P|$ is sufficiently large \cite{JonasPHD17}.
\end{remark}

\section{Main Result}

We are now ready to state the main result.

\begin{theorem} \label{theorem: main estimate} Let $E_\alpha(P) = \inf \sigma (H_\alpha(P))$,  $M^{\rm LP} = \frac{2}{3} \sno \nabla \varphi \sno^2_\2$ with $\varphi$ defined in \eqref{eq: optimal phonon mode} and choose $c>0$. For every $\varepsilon>0$ there exists a constant $C_{c,\varepsilon} > 0$ such that
\begin{align}\label{eq: main bound}
E_\alpha(P) \, \le \,  e^{\rm Pek} + \frac{{\rm Tr}_{L^2}(\sqrt { H^{\rm Pek} }  - 1 ) }{2 \alpha^2} + \frac{P^2}{2 \alpha^4 M^{\rm LP}} +C_{c,\varepsilon}\, \alpha^{-\frac{5}{2} + \varepsilon}
\end{align}
for all $|P|/\alpha \le c $ and all $\alpha$ large enough.
\end{theorem}

Since the operator $\sqrt { H^{\rm Pek} } - 1 $  is trace class, non-zero and non-positive (see Lemma \ref{lem: Hessian}), the second term on the right side is finite and lowers the energy. It corresponds to the predicted quantum corrections of the ground state energy of the Fr\"ohlich Hamiltonian \cite{Allcock65,Miyake1976,Tjablikow54,Gross76}. Since $E_\alpha(0) = \inf \sigma(H_\alpha)$, our theorem implies a two-term upper bound for the ground state energy of the Fr\"ohlich Hamiltonian that agrees with this prediction. For momenta in the range $\alpha^{-\frac{1}{4}+\frac{\varepsilon}{2}} \ll |P|/\alpha \le c$, the last term in \eqref{eq: main bound} is subleading for large $\alpha$ when compared to the momentum dependent term. In this region the upper bound describes a quadratic dispersion relation for a free quasi-particle with mass $\alpha^4 M^{\rm LP}$. The upper restriction on the range of $|P|$ is natural, since for $|P|/\alpha \ge \sqrt{2 M^{\rm LP}}$ the right side of \eqref{eq: main bound} would be larger than the value of the bottom of the continuous spectrum \eqref{eq: cont bottom}. The lower restriction $ |P|/\alpha \gg \alpha^{-\frac{1}{4}+\frac{\varepsilon}{2}} $, on the other hand, could in principle be improved by deriving a better error term in \eqref{eq: main bound}.

The derivation of a matching lower bound is, of course, more involved. To our knowledge the best known parabolic lower bound is still the one obtained by Lieb and Yamazaki \cite{Lieb1958} in 1958 stating that $E_\alpha(P) \ge c_1 e^{\rm Pek} + c_2 P^2 / ( 2 \alpha^4 M^{\rm LP} ) $ with $c_1 \approx 3.07$ and $c_2 \approx 0.11$. Even for $P=0$ it remains a challenging problem to improve the Lieb--Thomas bound \eqref{eq: Pekar energy} such that it includes the quantum corrections of order $\alpha^{-2}$. Progress in this direction has been achieved in \cite{FrankS2021,FeliciangeliS21} for simplified polaron models in which the electron and the quantum field are confined to suitable finite size regions.

In the next two sections we provide the definition of our trial state and formulate our main statement as a variational estimate. The remainder of the paper is devoted to the proof of the variational estimate. A sketch of the strategy of the proof is given in Section \ref{Sec: proof guide}.

\subsection{Bogoliubov Hamiltonian}

In this section we introduce and discuss a quadratic Hamiltonian defined on the Fock space. For its definition we set $\Pi_0$ and $\Pi_1$ to be the orthogonal projectors onto $\textnormal{Ker}H^{\rm Pek} = \text{Span}\{ \partial_i \varphi : i=1,2,3\} $ and $(\textnormal{Ker}H^{\rm Pek})^\perp$, that is
\begin{align}\label{eq: def of Pi_i}
\text{Ran}(\Pi_0) \, = \, \text{Ker}H^{\rm Pek}, \quad  
\text{Ran}(\Pi_1) \, = \, (\text{Ker}H^{\rm Pek})^\perp.
\end{align}
Even though we will not make explicit use of it, it is convenient to keep in mind that the decomposition $ L^2 (\mathbb R^3) \, =\, \text{Ran}(\Pi_0) \oplus \text{Ran}(\Pi_1) $
implies the factorization 
\begin{align}\label{eq: Fock space factorization}
\mathcal F = \mathcal F_0 \otimes \mathcal F_1 \quad \textnormal{with} \quad \mathcal F_0 = \mathcal F( \text{Ran}(\Pi_0)) \quad \textnormal{and} \quad \mathcal F_1 = \mathcal F( \text{Ran}(\Pi_1)).
\end{align}

For technical reasons, which are explained in Section \ref{sec: transf prop UK}, we introduce the Bogoliubov Hamiltonian $\mathbb H_K $ with a momentum cutoff $K\in (0,\infty]$. Setting $\mathbb N_1 = \D \Gamma(\Pi_1)$ (the number operator on $\mathcal F_1$) we define
\begin{align}\label{eq: Bogoliubov Hamiltonian maintext}
\mathbb H_K  \,  =\, \mathbb N_1  - \lsp \psi \I  \phi( h_{K,\cdot}^1 ) R \phi( h_{K,\cdot}^1 )  \psi \rsp_\2,
\end{align}
where the new coupling function
\begin{align}\label{def: cut off coupling function}
h_{K,x}^1(y) \,  =\, \int \D z\,  \Pi_1(y,z) h_{K,x}(z) \quad \text{with}\quad h_{K,x}(y) \,  =\, \frac{1}{(2\pi)^3} \int_{|k|\le K} \frac{e^{ik(x-y)}}{|k|} \D k 
\end{align}
results from the coupling function $h_{x}$ by removing all momenta larger than $K$ and then projecting to $\text{Ran}(\Pi_1)$. The second term in \eqref{eq: Bogoliubov Hamiltonian maintext} defines the quadratic operator given by
\begin{align}
&  \lsp \psi \I \phi( h^1_{K,\cdot} ) R \phi( h^1_{K,\cdot} )  \psi \rsp_\2  \notag \\
& \quad \quad    = \iint \D y  \D z \,  \lsp \psi \I (h^1_{K,\cdot}) (y) R ( h^1_{K,\cdot})(z)  \psi \rsp_\2 (a_y^\dagger + a_y) (a_z^\dagger + a_z) .
\end{align}
By definition $\mathbb H_K$ acts non-trivially only on the tensor component $\mathcal F_1$. Below we will show that $\mathbb H_K$ is bounded from below and diagonalizable by a unitary Bogoliubov transformation. For the precise statement, we need some further preparations.

For $K\in (0,\infty]$ we introduce $H^{\rm Pek}_K$ as the operator on $L^2(\mathbb R^3)$ defined by 
\begin{subequations}
\begin{align}\label{eq: def of H:Pek:K:1}
H^{\rm Pek}_K \restriction  \textnormal{Ran} (\Pi_1)  & \, =\,  \Pi_1 - 4 T_K \\
H^{\rm Pek}_K \restriction  \textnormal{Ran} (\Pi_0) & \, =\,  0 \label{eq: def of H:Pek:K:0}
\end{align}
\end{subequations}
where $T_K$ is defined by the integral kernel
\begin{align}\label{eq: Hessian kernel with cutoff}
T_K(y,z)&  \, = \,  \re \lsp \psi \I  h^1_{K,\cdot}(y) R  h^1_{K,\cdot}(z)  \psi \rsp_\2.
\end{align}
By definition $H^{\rm Pek}_\infty = H^{\rm Pek}$, see \eqref{eq: Hessian kernel}. Moreover we set $\Theta_K = (H^{\rm Pek}_K)^{1/4}$ and
\begin{subequations}
\begin{align}
A_K \restriction \text{Ran}(\Pi_1)  & \, =\, \frac{\Theta_K^{-1} + \Theta_K }{2} && \hspace{-2cm} B_K \restriction \text{Ran}(\Pi_1) \, = \, \frac{\Theta_K^{-1} - \Theta_K }{2} \\
A_K \restriction \text{Ran}(\Pi_0)  & \, =\,  \Pi_0 && \hspace{-2cm} B_K \restriction \text{Ran}(\Pi_0)  \, =\,  0  . \label{eq: A and B on Pi0}
\end{align}
\end{subequations}
The next lemma, whose proof can be found in Section \ref{Sec: Remaining Proofs}, implies some useful properties of these operators, among others, that there is a constant $C>0$ such that
\begin{align}\label{eq: Shale--Stinespring}
\sup_{K\ge K_0} \big( \sno A_K \sno_{\op} + \sno B_K \sno_{\HS} \big) \le C
\end{align}
for some $K_0$ large enough.

\begin{lemma}\label{lem: regularized Hessian} For $K_0$ large enough there exist constants $\beta \in (0,1)$ and $C>0$ such that
\begin{itemize}
\item[{(i)}] $0 \le H^{\rm Pek}_K \le 1$ and $ (H^{\rm Pek}_K - \beta) \restriction  \textnormal{Ran} (\Pi_1)  \ge 0$
\item[{(ii)}] $(B_K)^2 \le C( 1- H_K^{\rm Pek})$
\item[(iii)] $ {\rm{Tr}}_{L^2}(1- H_K^{\rm Pek}) \le C  $
\end{itemize}
for all $K\in (K_0,\infty]$. Moreover for all $K\in (K_0, \infty)$
\begin{itemize}
\item[(iv)] ${\rm{Tr}}_{L^2}((-i\nabla)( 1- H_K^{\rm Pek} )(-i\nabla))  \le C K $.
\end{itemize} 
\end{lemma}

Up to normal ordering the Hamiltonian $\mathbb H_K$ corresponds to the second quantization of $H^{\rm Pek}_K$. From the properties of the latter we can deduce that $\mathbb H_K$ is diagonizable by a unitary Bogoliubov transformation. To this end we introduce the transformation
\begin{align} \label{eq: def of U}
\mathbb U^{\textcolor{white}{.}}_K a(f) \mathbb U^\dagger_K \,  & = \,  a( A_K f ) + a^\dagger (  B_{K} \overline{ f} ) \quad \text{for all}\ f\in L^2(\mathbb R^3).
\end{align}
That this transformation defines a unitary operator $\mathbb U_K$ for all $K\in (K_0,\infty]$ is a consequence of \eqref{eq: Shale--Stinespring}. This is known as the Shale-Stinespring condition and we refer to \cite{JPS2007} for more details. Also note that $\mathbb U_K$ does not mix the two components in $\mathcal F = \mathcal F_0 \otimes \mathcal F_1$.
\begin{lemma}\label{prop: diagonalization of HBog} 
For $K\in (K_0,\infty]$ with $K_0$ large enough and $\mathbb U_K$ the unitary operator defined by \eqref{eq: def of U}, we have
\begin{align}\label{eq: diagonalization of HBog} 
\mathbb U^{\textcolor{white}.}_K \mathbb H^{\textcolor{white}.}_K \mathbb U^\dagger_K & \, = \,  \D \Gamma(\sqrt{H^{\rm Pek}_K})  + \frac{1}{2} \textnormal{Tr}_{L^2 }( \sqrt{H^{\rm Pek}_K} - \Pi_1  )
\end{align}
with $H^{\rm Pek}_K$ defined by \eqref{eq: def of H:Pek:K:1} and \eqref{eq: def of H:Pek:K:0}.
\end{lemma}
The proof is obtained by an explicit computation and postponed to Section \ref{Sec: Remaining Proofs}. From this lemma, we can infer that the ground state energy of $\mathbb H_K$ is given by
\begin{align}\label{rem: Bog ground state energy} 
\inf \sigma( \mathbb H_K) \,  = \,  \frac{1}{2}\textnormal{Tr}_{ L^2  } \big(  \sqrt{H^{\rm Pek}_K} - \Pi_1 \big) \, = \,  \frac{1}{2}\textnormal{Tr}_{ L^2  } \big(  \sqrt{H^{\rm Pek}_K} - 1 \big) + \frac{3}{2},
\end{align}
where we also used $\Pi_1 = 1-\Pi_0$ and $\text{Tr}_{L^2}(\Pi_0)= 3$. Moreover, since $H^{\rm Pek}_K \le \Pi_1$ we have $\inf \sigma( \mathbb H_K) < 0$ and from Item (iv) of Lemma \ref{lem: regularized Hessian} we find that $\inf \sigma( \mathbb H_K) > - \infty $ uniformly in $K\to \infty$.

For the ground state of $\mathbb H^{\textcolor{white}.}_K$ we shall use the notation
\begin{align}
\Upsilon_K \,  = \, \mathbb U_K^\dagger \Omega ,
\end{align}
where it is important to keep in mind that the state $\Upsilon_K $ has excitations only in $\mathcal F_1$ (i.e., no zero-mode excitations) since $\mathbb U_K^\dagger$ acts as the identity on $\mathcal F_0$, see \eqref{eq: A and B on Pi0}.

\subsection{Trial state \label{Sec: trial state} and variational estimate} 

As starting point for the definition of our trial state consider the Fock space wave function obtained from the fiber decomposition of the classical Pekar product state $\Psi_{\psi,\varphi}$, that is
\begin{align}\label{eq:fiber:Pekar}
\Psi^{\rm Pek}_\alpha(P) & \,  =\, \int \D x\, e^{i ( P_f - P) x } \psi(x) e^{a^\dagger( \alpha \varphi) } \Omega .
\end{align}
Testing the energy of $H_\alpha(P)$ with $\Psi^{\rm Pek}_\alpha(P)$, one would in fact obtain that $E_\alpha(P)$ is bounded from above by
\begin{align}\label{eq:upper:bound:fiber:Pekar}
e^{\rm Pek} - \frac{3}{2\alpha^2} + \frac{P^2}{\alpha^4 M^{\rm LP}} + o(\alpha^{-2}).
\end{align}
For $E_\alpha(0)$ this provides already a better bound compared to the semiclassical approximation for $\inf \sigma(H_\alpha)$. The improvement comes from taking into account the translational symmetry and can be interpreted as the missing zero-point energy of three quantum oscillators (that turned into translational degrees of freedom). As a side remark, we find it somewhat surprising that fiber decompositions of this form have been employed very rarely in the polaron literature, exceptions being \cite{Hoehler1961} and \cite{Nagy1989}. We think they could be of interest also for other translation-invariant polaron type models.

To obtain the desired bound for $E_\alpha(P)$, we need to add several modifications to the integrand in \eqref{eq:fiber:Pekar}. On the one hand, we have to replace the classical field $\varphi$ by a suitably shifted $\varphi_P$ in order to get the correct momentum dependent term (note that \eqref{eq:upper:bound:fiber:Pekar} is missing a factor $\frac{1}{2}$ in the quadratic term). The missing part of the rest energy (compare with \eqref{rem: Bog ground state energy}), on the other hand, is caused by two types of correlations that need to be added to the Pekar product state. First, we include correlations between the electron and the phonons. This is done in the spirit of first-order adiabatic perturbation theory. Second, we rotate the vacuum by the unitary transformation \eqref{eq: def of U} that diagonalizes the Bogoliubov Hamiltonian \eqref{eq: Bogoliubov Hamiltonian maintext}. As discussed, the latter describes the quantum fluctuations of the phonons around the classical field. For technical reasons, briefly explained in Section \ref{Sec: proof guide}, we also need to introduce suitable momentum and space cutoffs in the trial state.

Explicitly, we consider the family of Fock space wave functions $\Psi_{K,\alpha}(P) \in \mathcal F$, depending on the coupling $\alpha$, the total momentum $P\in \mathbb R^3$ and the cutoff $K \in (K_0,\infty)$, given by
\begin{align} \label{eq: def of trial state}
\Psi_{K,\alpha}(P)  & \,  =\, \int \D x \,  e^{ i ( P_f - P ) x }\, e^{a^\dagger(\alpha \varphi_P) - a(\alpha \varphi_P)} \big( G^{0}_{K,x} - \alpha^{-1} G^{1}_{K,x}\big) 
\end{align}
where  
\begin{align}\label{eq: def of varphi_P}
\varphi_P \, =\,  \varphi + i \xi_P  \quad \text{with} \quad \xi_P \,  =\, \frac{1}{\alpha^2 M^{\rm LP}} (P  \nabla ) \varphi, \quad  M^{\rm LP} = \frac{2}{3} \sno \nabla \varphi \sno^2_\2,
\end{align}
and
\begin{align}
G^{0}_{K,x} & \,  =\,  \psi (x)  \Upsilon_K  , \quad G^{1}_{K,x} = u_\alpha(x) (R \phi( h^1_{K,\cdot} ) \psi )(x)  \Upsilon_K \quad \text{and}\quad \Upsilon_K = \mathbb U_K^\dagger \Omega.
\end{align}
Here $u_\alpha :\mathbb R^3 \to [0,1]$ is a radial function, satisfying
\begin{align}\label{eq: properties of ualpha}
u_\alpha (x) \,  =\,
\begin{cases}
  1 & \forall \ |x| \le \alpha  \\
  0 & \forall\ |x| \ge  2\alpha
\end{cases} \quad \quad\text{and}\quad \quad \sno \nabla u_\alpha \sno_\su + \sno \Delta u_\alpha \sno_\su  \, \le \,  C \alpha^{-1}
\end{align}
for some $C>0$. 
For completeness, we recall that $\psi > 0$ and $\varphi$ are the unique rotational invariant minimizers of the Pekar functionals \eqref{eq: electronic pekar functional} and \eqref {eq: phonon pekar functional}.

\begin{remark} \label{remark: HS isomorphism}
Writing $G_{K,x}^{i}$ we think of these states as elements in $L^2(\mathbb R^3, \mathcal F)$ and of
\begin{align}
(R \phi( h^1_{K,\cdot} ) \psi )(x) \,  =\, \iint \D z \D y \, R(x,y) h_{K,y}^1 (z)  \psi(y)\,  \big( a^\dagger_z + a_z \big)
\end{align}
as an $x$-dependent Fock space operator. Via the isomorphism $L^2(\mathbb R^3, \mathcal F) \simeq \mathscr H$, we can view $G_{K,x}^{i}$ also as a wave function in $ \mathscr H$. In this case we shall write
\begin{align}\label{eq: alternative definition of G}
G^{0}_{K } \,  =\,  \psi \otimes  \Upsilon_K , \quad G_{K}^{1} \,  =\,  u_\alpha R \phi(h^1_{K,\cdot}) \psi \otimes \Upsilon_K .
\end{align}
\end{remark}

For the introduced trial states, we prove the following variational estimate, where $\mathbb H_\infty$ denotes the Bogoliubov Hamiltonian \eqref{eq: Bogoliubov Hamiltonian maintext} for $K=\infty$.

\begin{proposition}\label{theorem: main estimate 2} 
Let $\Psi_{K,\alpha}(P)\in \mathcal F$ as in \eqref{eq: def of trial state}, choose $c,\tilde c>0$ and set $r(K,\alpha) = K^{-1/2}\alpha^{-2} + \sqrt K \alpha^{-3}$. For every $\varepsilon>0$ there exists a constant $C_{\varepsilon} > 0$ \textnormal{(}we omit the dependence on $c$ and $\tilde c$\textnormal{)} such that
\begin{align}\label{eq: main estimate 2}
\frac{\lsp \Psi_{K,\alpha}(P) | H_\alpha(P) \Psi_{K,\alpha}(P) \rsp_{\Fock} }{\lsp \Psi_{K,\alpha}(P)  | \Psi_{K,\alpha}(P) \rsp_{\Fock}} \, \le \,  e^{\rm Pek} + \frac{\inf \sigma (\mathbb H_\infty) - \frac{3}{2}}{\alpha^2} + \frac{P^2}{2 \alpha^4 M^{\rm LP}} + C_\varepsilon \alpha^{\varepsilon} r(K,\alpha)
\end{align}
for all $|P|/\alpha  \le c   $ and all $K$, $\alpha$ large enough with $ K /\alpha \le \tilde c $.
\end{proposition}

With \eqref{rem: Bog ground state energy} and $H_\infty^{\rm Pek} = H^{\rm Pek}$ we can rewrite the term of order $\alpha^{-2}$ as 
\begin{align}\label{eq: Bog plus -3/2}
\inf \sigma (\mathbb H_\infty)  -\frac{3}{2}\, =  \, \frac{1}{2}\textnormal{Tr}_{L^2} \big( \sqrt {H^{\rm Pek}} - 1 \big).
\end{align}
Choosing $K$ now proportional to $\alpha$ optimizes the asymptotics of the error in \eqref{eq: main estimate 2} and proves Theorem \ref{theorem: main estimate}.

\section{Proof of Proposition \ref{theorem: main estimate 2}}

We recall the definition of the field operators
\begin{align}
\phi(f) = a^\dagger(f) + a(f) , \quad \pi(f) = \phi(i f )
\end{align}
and the Weyl operator
\begin{align}\label{eq: BCH for Weyl}
W(f) \,  =\, e^{-i\pi(f) } \,  =\, e^{a^\dagger(f) - a(f)} \,  =\, e^{a^\dagger(f)} e^{-a(f)} e^{-\frac{1}{2} \sno f\sno^2_\2 }.
\end{align}
The Weyl operator is unitary and satisfies
\begin{align}\label{eq: Weyl identities}
W^\dagger(f) \,  =\,  W(-f) , \quad W(f) W(g) \,  =\, W(g)W(f) e^{2 i \im \langle g | f \rangle_\2 }  \,  =\, W(f + g ) e^{ i \im \langle g | f \rangle_\2 }.
\end{align}

\subsection{The total energy}

The proof of Proposition \ref{theorem: main estimate 2} starts with a convenient formula for the energy evaluated in the trial state. For the precise statement, we introduce the $y$-dependent function in $L^2(\mathbb R^3)$,
\begin{align}
w_{P,y} \,  =\, (1-e^{-y\nabla})\varphi_P,
\end{align}
and the $y$-dependent Fock space operator
\begin{align}\label{eq: definition AP}
A_{P,y} \,  =\, i P_f y + i g_{P}(y)  , \quad g_{P}(y) \,  =\,  - \frac{2}{M^{\rm LP}} \int_0^1 \D s\, \langle \varphi | e^{-sy  \nabla} (y \nabla)^3 (P \nabla) \varphi \rangle_\2 .
\end{align}
Since $g_{P}(y)$ is real-valued we have $(A_{P,y})^\dagger = - A_{P,y}$. We further consider the Weyl-transformed Fr\"ohlich Hamiltonian,
\begin{align}
\widetilde H_{\alpha,P} \,  =\, W(\alpha \varphi_P)^\dagger\big( H_\alpha - e^{\rm Pek} \big) W(\alpha \varphi_P) \, & = \, h^{\rm Pek} + \alpha^{-2} \mathbb N + \alpha^{-1} \phi(h_{x} + \varphi_P ),
\end{align} 
where we recall $h^{\rm Pek} = -\Delta + V^\varphi - \lambda^{\rm Pek}$, and denote the shift operator acting on $L^2(\mathbb R^3)$ by $T_y = e^{y\nabla}$ with $y\in \mathbb R^3$.

\begin{lemma} \label{lem: energy identity} For $\Psi_{K,\alpha}(P)$ defined in \eqref{eq: def of trial state} we have
\begin{align}
\lsp \Psi_{K,\alpha}(P)| H_\alpha(P) \Psi_{K,\alpha}(P)\rsp_\Fock & \,  =\, \Big( e^{\rm Pek} + \frac{P^2}{2 \alpha^4 M^{\rm LP}} \Big)\, \mathcal N  + \mathcal E +  \mathcal G + \mathcal K
\end{align}
where $\mathcal N = \sno \Psi_{K,\alpha}(P)\sno^2_\Fock$ and
\begin{subequations}
\begin{align}
\mathcal E & \,  =\, \int \D y\,  \lsp G_{K}^{0}| \widetilde H_{\alpha,P} T_y  e^{A_{P,y}} W(\alpha w_{P,y}) G_{K}^{ 0} \rsp_\h \label{eq: E}\\
\mathcal G & \,  =\, - \frac{2}{\alpha }  \int \D y\, \re \lsp G_{K}^{ 0}| \widetilde H_{\alpha,P} T_y e^{A_{P,y}} W(\alpha w_{P,y})G_{K}^{ 1} \rsp_\h \label{eq: G}\\
\mathcal K & \,  =\, \frac{1}{\alpha^2} \int \D y\,  \lsp G_{K}^{1}| \widetilde H_{\alpha,P} T_y e^{A_{P,y}} W(\alpha w_{P,y}) G_{K}^{1} \rsp_\h .\label{eq: K}
\end{align}
\end{subequations}
\end{lemma}

For the proof we recall that the Weyl operator shifts the creation and annihilation operators by complex numbers,
\begin{align}
W(g)^\dagger a^\dagger (f) W(g) \,  =\, a^\dagger(f) + \langle g| f\rangle_\2 , \quad W(g)^\dagger a (f) W(g) \,  =\, a(f) + \overline{\langle g| f\rangle_\2},
\end{align}
and, as a simple consequence,
\begin{subequations}
\begin{align}
W(g)^\dagger \phi(f) W(g) & \,  =\,  \phi(f) +  2 \re \lsp f | g\rsp_\2  \label{eq: W phi W},\\[0.5mm]
W(g)^\dagger \mathbb N W (g) & \,  =\, \mathbb N +  \phi(g) +  \sno g \sno^2_\2, \label{eq: W N W}\\[0.5mm]
W(g)^\dagger P_f W(g) & \,  =\, P_f - a^\dagger( i \nabla  g) - a( i \nabla  g) -  \lsp g | i \nabla   g \rsp_\2 .
\end{align}
\end{subequations}
Moreover we need the following identity.

\begin{lemma} 
\label{lem: W shift identity} Let $\varphi_P = \varphi + i \xi_P$ with $\xi_P$ defined by \eqref{eq: def of varphi_P}. Then
\begin{align}
W^\dagger(\alpha \varphi_P) e^{ i ( P_f - P)  y } W(\alpha \varphi_P ) & \,  =\, e^{A_{P,y}} W(\alpha w_{P,y}) .
\end{align}
\end{lemma}
\begin{proof}[Proof of Lemma \ref{lem: W shift identity}]
We first observe that
\begin{align}
e^{- i P_f y } a^\dagger(f) e^{i P_f y } & \,  =\,  a^\dagger( e^{-  y \nabla } f ) \label{eq: tranformatoin of Weyl operator with P_f}
\end{align}
which follows from $\frac{d}{ds} e^{- i s P_f y} a^\dagger( e^{(s-1)y\nabla} f ) e^{isP_fy} = 0$. 
In combination with \eqref{eq: Weyl identities} this leads to
\begin{align}
W^\dagger(\alpha \varphi_P) e^{ i P_f y } W(\alpha \varphi_P ) & 
\,  =\, e^{iP_f y} W( \alpha ( 1 - e^{- y \nabla}) \varphi_P) \exp\big( i\alpha^2  \im \langle \varphi_P | e^{-y \nabla } \varphi_P \rangle_\2 \big).
\end{align}
Recalling $\varphi_P = \varphi + i \frac{1}{\alpha^2 M^{\rm LP}} (P \nabla) \varphi$, we compute
\begin{align}
\alpha^2 \im & \langle \varphi_P | e^{-y \nabla }  \varphi_P \rangle_\2 
 \,  =\, \frac{2}{ M^{\rm LP}}  \langle \varphi | e^{-y \nabla } (P\nabla) \varphi \rangle_\2 \notag\\  & =\, - \frac{2}{ M^{\rm LP}}  \langle \varphi | (y\nabla) (P\nabla) \varphi  \rangle_\2 - \frac{2}{M^{\rm LP}} \int_0^1 \D s\, \langle \varphi | e^{-sy\nabla} (y\nabla)^3 (P\nabla) \varphi \rangle_\2
\end{align}
where we inserted $e^{-y\nabla} = 1 - (y\nabla) + \frac{1}{2}( y \nabla)^2 - \int_0^1 \D s\, e^{-s y\nabla } (y\nabla)^3$ and used that, due to rotational invariance of $\varphi$, $  \langle \varphi | (P\nabla) \varphi \rangle_\2 = 0 =  \langle \varphi | (y \nabla)^2 (P\nabla) \varphi \rangle_\2 $. Also because of rotational invariance,
\begin{align}
 \langle \varphi | (y\nabla) (P  \nabla) \varphi \rangle_\2 \,  =\, - \frac{(P y )}{3} \sno \nabla  \varphi \sno_\2^2 \,  =\, - \frac{ (Py) } {2}  M^{\rm LP}  , 
\end{align}
and thus, $ \alpha^2 \im \langle \varphi_P | e^{-y \nabla } \varphi_P \rangle_\2  =  P  y+  g_{P}(y)$.
\end{proof}

\begin{proof}[Proof of Lemma \ref{lem: energy identity}] The norm squared is given by
\begin{align}
\mathcal N  & = \Big\|   \int \D x \, e^{i ( P_f - P )x} W(\alpha \varphi _P)  \big( G^{0}_{K,x} - \alpha^{-1} G^{1}_{K,x}\big)  \Big\|^2_\Fock \notag\\
& = \sum_{i\in \{0,1\}} \alpha^{-2i} \iint \D y\D x\,  \lsp G_{K,x}^{i}| W(\alpha \varphi_P )^\dagger e^{i( P_f - P )(y-x)} W(\alpha \varphi _P )  G_{K,y}^{i} \rsp_\Fock \notag \\
& \quad  \quad - 2 \alpha^{-1} \re  \iint \D y\D x\, \lsp G_{K,x}^{0}| W(\alpha \varphi_P )^\dagger e^{i(P_f-P)(y-x)}  W(\alpha \varphi _P)  G_{K,y}^{1} \rsp_\Fock\label{eq: energy derivation 2}.
\end{align}
Shifting $y\to y+x$ and writing the $x$-integration as an inner product in the electron coordinate, cf. Remark \ref{remark: HS isomorphism}, we can proceed for $i,j\in\{0,1\}$ with
\begin{align}
& \iint \D y\D x\, \lsp G_{K,x}^{i}| W(\alpha \varphi_P )^\dagger e^{i(P_f-P)(y-x)}  W(\alpha \varphi _P)  G_{K,y}^{j} \rsp_\Fock \notag\\
& \quad =  \iint \D y\D x\, \lsp G_{K,x}^{i}| W(\alpha \varphi_P )^\dagger e^{i(P_f-P)y}  W(\alpha \varphi _P)  G_{K,y+x}^{j} \rsp_\Fock \notag\\
& \quad = \int \D y\, \lsp G_{K}^{i}| W(\alpha \varphi_P )^\dagger e^{i(P_f-P)y}  W(\alpha \varphi _P) T_y G_{K}^{j} \rsp_\h  \notag\\
& \quad =  \int \D y\, \lsp G_{K}^{i}| e^{A_{P,y}} W(\alpha w_{P,y}) T_y G_{K}^{j} \rsp_\h, \label{eq: energy derivation 2b}
\end{align}
where we applied Lemma \ref{lem: W shift identity} in the last step. Similarly for the energy
\begin{align}
& \lsp \Psi_{K,\alpha}(P)| H_\alpha(P) \Psi_{K,\alpha}(P)\rsp_\Fock \notag \\[1mm]
& = \sum_{i\in \{0,1\}} \alpha^{-2i} \iint \D y\D x\,  \lsp G_{K,x}^{i}| W(\alpha \varphi_P )^\dagger e^{-i(P_f-P)x} H_\alpha (P) e^{i( P_f - P )y} W(\alpha \varphi _P )  G_{K,y}^{i} \rsp_\Fock \notag \\
&  - 2 \alpha^{-1}\re  \iint \D y\D x \lsp G_{K,x}^{0}| W(\alpha \varphi_P )^\dagger e^{-i(P_f-P)x} H_\alpha(P) e^{i( P_f - P )y} W(\alpha \varphi _P)  G_{K,y}^{1} \rsp_\Fock\label{eq: energy derivation 3}
\end{align}
where we also used self-adjointness of $H_\alpha(P) $. Next we invoke
\begin{align}
e^{-i(P_f-P)x} H_\alpha(P) = \big( -\Delta_x  + \alpha^{-2}\mathbb N + \alpha^{-1}\phi(h_x) \big)  e^{-i(P_f-P)x} 
\end{align}
to proceed for $i,j\in \{0,1\}$ with
\begin{align}
&  \iint \D y\D x\,  \lsp G_{K,x}^{i}| W(\alpha \varphi_P )^\dagger \big( -\Delta_x  + \alpha^{-2}\mathbb N + \alpha^{-1}\phi(h_x)\big) e^{i( P_f - P )(y-x)} W(\alpha \varphi _P )  G_{K,y}^{j} \rsp_\Fock \notag \\
&  \hspace{0mm} =\,  \int \D y \, \lsp G_{K}^{i}| W(\alpha \varphi_P)^\dagger H_\alpha e^{i ( P_f - P ) y } W(\alpha \varphi_P) T_y   G_{K }^{j} \rsp_\h \notag\\
& \hspace{0mm} =\,  \int \D y \, \lsp G_{K}^{i}| W(\alpha \varphi_P)^\dagger H_\alpha  W(\alpha \varphi_P)  e^{A_{P,y}} W(\alpha w_{P,y})  T_y   G_{K }^{j} \rsp_\h \label{eq: energy derivation} .
\end{align}
Using \eqref{eq: W phi W}, \eqref{eq: W N W} and $- 2\re  \langle \varphi_P | h_{x} \rangle_\2 = - 2\re  \langle \varphi | h_{x} \rangle_\2 = V^\varphi(x)$ we have
\begin{align}\label{eq: Weyl transform of H}
 W(\alpha \varphi_P)^\dagger  H_\alpha   W(\alpha \varphi_P) & = -\Delta_x + V^\varphi(x)  + \alpha^{-2} \mathbb N + \alpha^{-1} \phi(h_x + \varphi_P) + \sno \varphi_P\sno^2_\2   \notag\\[1mm]
& = \widetilde H_{\alpha, P} + e^{\rm Pek} + \sno \varphi_P\sno_\2^2 - \sno \varphi \sno_\2^2
\end{align}
where we added and subtracted $e^{\rm Pek}= \lambda^{\rm Pek} + \sno \varphi\sno_\2^2$. It remains to compute
\begin{align}\label{eq: norm of varphi_p}
\sno \varphi_P \sno^2_\2 - \sno \varphi \sno^2_\2 \,  =\,  \frac{1}{\alpha^4 ( M^{\rm LP} )^2} \sno (P\nabla) \varphi \sno^2_\2  & \,  =\,  \frac{P^2}{2 \alpha^4 M^{\rm LP} } 
\end{align}
since $\sno (P\nabla) \varphi \sno^2_\2  =  \frac{P^2}{3} \sno \nabla \varphi \sno^2_\2 = \frac{P^2}{2} M^{\rm LP}$ because of rotational invariance of $\varphi$. With \eqref{eq: Weyl transform of H} inserted into \eqref{eq: energy derivation}, the stated formula for the energy follows from \eqref{eq: energy derivation 2} and \eqref{eq: energy derivation 3}.
\end{proof}

\subsection{A short guide to the proof \label{Sec: proof guide}}

\subsubsection{Heuristic picture}
Given Lemma \ref{lem: energy identity}, the remaining task is to show that $(\mathcal E+\mathcal G + \mathcal K)/\mathcal N$ coincides, up to small errors, with the energy contribution of order $\alpha^{-2}$ in \eqref{eq: main estimate 2}. Although our proof is somewhat technical, the main idea is a simple one, and we explain the corresponding heuristics here in order to facilitate the reading. The main point is that the integrals appearing in the terms given in Lemma \ref{lem: energy identity} turn out to be, as $\alpha\rightarrow\infty$ and $|P|/ \alpha \leq c$, sharply localized around zero at the length scale of order $ \alpha^{-1}$. In this regime, as formally $w_{P,y}(z)\approx y\nabla\varphi(z) $ for $y$ small, the Weyl operator $W(\alpha w_{P,y})$ effectively acts non-trivially only on the $\mathcal{F}_0$ part of the Fock space (at this point it is convenient to recall the factorization \eqref{eq: Fock space factorization}). Moreover, we shall show that $e^{A_{P,y}}$ can be effectively replaced by the identity operator and it suffices to consider $T_y \approx 1+ y\nabla$. Since our trial state coincides with the vacuum on $\mathcal{F}_0$, we thus expect for $|y|$ small that
\begin{align}\label{eq: expected}
T_ye^{A_{P,y}}W(\alpha w_{P,y}) G^i_K \approx e^{-\lambda \alpha^2 y^2 } \left(1 +  y\nabla \right) e^{a^{\dagger}(\alpha y\nabla \varphi)} G^i_K , \quad i=0,1 
\end{align}
with $\lambda=\frac{1}{6}\sno \nabla \varphi \sno_\2^2$. (Since $T_y$ acts on the electron coordinate, it commutes with $e^{A_{P,y}}$ and $W(\alpha w_{P,y})$). Taking this approximation for granted, and considering only the term with $i=j=0$ in \eqref{eq: energy derivation 2b}, would lead to
\begin{align}
  \mathcal{N} \, \approx \, \int \D y\, \lsp G_K^0 | T_ye^{A_{P,y}}W(\alpha w_{P,y}) G_K^0 \rsp_\h  \, = \, \int \D y\, e^{-\lambda \alpha^2 y^2 }  \ + \ \text{Errors} .
\end{align}
With the above replacement, and keeping only the terms of order $\alpha^{-2}$ (relative to the factor from the norm), the energy terms are found to be given by 
\begin{subequations}
\begin{align}
\mathcal E & \, =\, \frac{1}{\alpha^2}  \lsp \psi \otimes \Upsilon_K | \mathbb N_1  \psi \otimes \Upsilon_K \rsp_\h  \int \D y\, e^{-\lambda \alpha^2 y^2 }  \ + \ \text{Errors}  \label{eq: heuristic E1} \\
& \quad  + \,  \frac{1}{\alpha}\int \D y\, e^{-\lambda \alpha^2 y^2 } \lsp \psi \otimes \Upsilon_K | \big( \phi(h_\cdot +\varphi) (1 + (y\nabla)a^{\dagger}(\alpha y  \nabla \varphi ) ) \big) \psi \otimes \Upsilon_K \rsp_\h   \label{eq: heuristic E2} \\[1mm] 
\mathcal G& \, = \, -  \frac{2}{\alpha^2} \re \lsp \psi \otimes \Upsilon_K | \phi(h_\cdot^1) u_\alpha R \phi(h_{K,\cdot}^1) \psi \otimes \Upsilon_K \rsp_\h  \int \D y \, e^{-\lambda \alpha^2 y^2 }   \ + \ \text{Errors} \label{eq: heuristic G} \\
 \mathcal K & \, = \, \frac{1}{\alpha^2} \lsp \psi \otimes \Upsilon_K | \phi(h_\cdot^1) R u_\alpha h^{\rm Pek} u_\alpha R \phi(h^1_{K,\cdot})\psi \otimes \Upsilon_K \rsp_\h \int \D y\, e^{-\lambda \alpha^2 y^2 }   \ + \ \text{Errors} \label{eq: heuristic K}.
\end{align}
\end{subequations}
From here the Bogoliubov energy is obtained by setting $u_\alpha = 1$ and $K=\infty$ in the leading-order terms, and using $R h^{\rm Pek} R = R$, since this would imply (omitting the errors)
\begin{align}
\eqref{eq: heuristic E1} + \eqref{eq: heuristic G} + \eqref{eq: heuristic K} & = \lsp \psi \otimes \Upsilon_\infty | ( \mathbb N_1  - \phi(h_\cdot^1) R \phi(h^1_{\infty,\cdot}) ) \psi \otimes \Upsilon_\infty \rsp_\h \frac{1}{\alpha^2} \int \D y \, e^{-\lambda \alpha^2 y^2 } \notag \\
&  =\frac{\inf \sigma (\mathbb H_\infty)}{\alpha^2} \int \D y \, e^{-\lambda \alpha^2 y^2 }  .\label{eq: Bogoliubov contribution}
\end{align}
The remaining $-\frac{3}{2\alpha^2}$ term stems from the part of the interaction involving the zero modes. In \eqref{eq: heuristic E2}, the term not involving $y\nabla$ vanishes due to $\langle \psi|h_{\cdot}\psi\rangle_\2 =-\varphi$. Moreover,  $\langle \psi|  h_{\cdot} \nabla \psi\rangle_\2 =-\frac{1}{2} \nabla \varphi$ using $\nabla h_{\cdot}=-(\nabla h)_{\cdot}$ via integration by parts (in the sense of distributions). Thus, since $[a^\dagger (y \nabla \varphi),\mathbb U_\infty^\dagger]=0$,
\begin{align}
  \eqref{eq: heuristic E2}  & \, = \, \int \D y \, e^{-\lambda \alpha^2 y^2 } \lsp \Omega| \phi(\langle \psi| h_\cdot y\nabla \psi\rangle )a^{\dagger} ( y  \nabla \varphi ) \Omega \rsp_\Fock  \notag \\
  & \, = \, -\frac{1}{2}\int \D y\,  e^{-\lambda \alpha^2 y^2 } \sno y\nabla \varphi \sno^2_\2=-\frac{3}{2\alpha^2} \int \D y\, e^{-\lambda \alpha^2 y^2 }.\label{eq: -3/2 contribution}
\end{align}
Equations \eqref{eq: Bogoliubov contribution} and \eqref{eq: -3/2 contribution} now add up to the desired energy of order $\alpha^{-2}$, see \eqref{eq: Bog plus -3/2}. Note that for estimating the error induced by replacing $e^{A_{P,y}}$ by unity we require the momentum cutoff $K$ in the definition of the trial state, see Lemma \ref{lem: bounds for P_f}.
 
The main issue in \eqref{eq: expected} is that while for small enough $y$ one can use the first-order approximation $W(\alpha w_{P,y}) \approx W(\alpha y\nabla\varphi)$, for $y$ large, on the other hand, the higher-order terms in $w_{P,y}$ begin to play an important part, ultimately killing the Gaussian factor. Writing
\begin{align}\label{eq: unwanted}
& \lsp G_K^i | \widetilde H_\alpha(P) e^{A_{P,y}}W(\alpha w_{P,y})T_y G^j_K\rsp_\h  \\
& \hspace{0.25cm} = e^{- \frac{\alpha^2}{2} \sno w_{P,y}\sno^2_\2 } \lsp G_K^i | \widetilde H_\alpha(P) e^{A_{P,y}} e^{a^{\dagger} (\alpha w_{P,y} )} e^{-a (\alpha w_{P,y} )} T_y G_K ^j \rsp_\h , \quad i,j=0,1, \notag
\end{align} 
we notice that, since
\begin{align}
\sno w_{P,y}\sno_\2^2 = 2 \int \D k\, |\hat \varphi_P(k)|^2 (1 - \cos(ky)) \to 2 \sno \varphi_P \sno_\2^2\quad \text{for}\quad |y| \to \infty,
\end{align}
the prefactor should lead to a $y$-independent, exponentially small constant.  
In order to make use of this exponential decay in $\alpha$, however, we need to ensure that
\begin{align}
\I \lsp G_K^i | \widetilde H_\alpha(P) e^{A_{P,y}} e^{a^{\dagger} (\alpha w_{P,y} )} e^{-a (\alpha w_{P,y} )} T_y G_K ^j \rsp_\h \I \le C \alpha^n g(y)
\end{align}
is polynomially bounded in $\alpha$ with some integrable function $g(y)$, which heuristically can be expected to be true since the average number of particles in the state $\widetilde H_\alpha(P) G_K^i$ is of order one w.r.t. $\alpha$. To obtain the required integrability in $y$ is also the reason for introducing the cutoff function $u_{\alpha}$ in the definition of $G^1_K$.

\subsubsection{Outline of the proof}

Although the replacement \eqref{eq: expected} illustrates the main idea behind extracting the leading order terms, in our proof we do not directly perform this replacement and estimate the resulting error. Instead, when taking inner products, we commute the exponential operators $e^{a^\dagger(\alpha w_P)}$ and $e^{-a(\alpha w_P)}$ in $W(\alpha w_{P,y})$ to the left resp. to the right until they hit the vacuum state in $G_K^i$. This involves the Bogoliubov transformation, cf. Lemma \ref{lem: U W U transformation}, and gives rise to a slight modification of $w_{P,y}$, which we denote by $\widetilde{w}_{P,y}$. These manipulations naturally lead to a multiplicative factor $\exp(-\frac{\alpha^2}{2}\sno \widetilde{w}_{P,y} \sno^2_\2 )$ which, as we shall see, indeed behaves like the Gaussian function in \eqref{eq: expected} for $|y|$ small and tends to a constant exponentially small in $\alpha$ as $|y| \rightarrow\infty$. In Lemma \ref{lem: Gaussian lemma} we prove the large $\alpha$ asymptotics of integrals of the type $\int g(y)\exp(-\frac{\alpha^2}{2}\sno \widetilde{w}_{P,y}\sno_\2^2 ) \D y$ for a suitable class of functions $g$. The major part of the proof, apart from extracting the leading order terms, is to establish that the resulting error terms in the integrands are, in fact, functions in this class. This is, for the most part, achieved by use of elementary estimates combined with the commutator method by Lieb and Yamazaki \cite{Lieb1958} in the form stated in Lemma \ref{lem: LY CM}. As already mentioned, for certain terms this makes the introduction of the space cutoff $u_{\alpha}$ and the momentum cutoff $K$ necessary, while for others, it is enough to use the well-known regularity properties of $\psi$, the relevant consequences of which are summarized in Lemma \ref{lemma: props_peks}.

In the next two sections, we state the remaining necessary lemmas. The main proof is then carried out in Sections \ref{sec: norm}--\ref{Sec: concluding the proof}.

Throughout the remainder of the proof we will abbreviate constants by the letter $C$ and write $C_\tau$ whenever we want to specify that it depends on a parameter $\tau$. As usual, the value of a constant may change from one line to the next.

\subsection{The Gaussian lemma}

We recall that $w_{P,y} = (1-e^{-y\nabla})\varphi_P$ and $\Theta_K =  (H^{\rm Pek}_K)^{1/4} $ and set
\begin{subequations}
\begin{align}
w_{P,y}^{0} & \, =\,  \Pi_0 w_{P,y} \in \text{Ker}H^{\rm Pek}\\[1mm]
w_{P,y}^{ 1 } & \, =\,  \Pi_1 w_{P,y} \in ( \text{Ker}H^{\rm Pek})^\perp  \\[1mm]
\widetilde w_{P,y}&  \, =\,  w_{P,y}^0  + \Theta_K \re( w_{P,y}^{1} )  + i \Theta_K^{-1} \im (w_{P,y}^{1} ) \label{eq: def of tilde w}.
\end{align}
\end{subequations}
\begin{remark}\label{rem: symmetries_of_w}
Note that $(y,z)\mapsto \re(w_{P,y})(z)$ is even as a function on $\mathbb{R}^6$, while $\im (w_{P,y})(z)$ is odd on the same space. Since $\Pi_0$ and $\Theta_K$ have real-valued kernels that are even as functions on $\mathbb{R}^6$, they preserve the parity properties just mentioned. That $\Pi_0$ has the desired properties follows directly from its explicit form. To see this for $\Theta_K$, it is enough to check this for $H_K^{\rm Pek}$, which can be easily done using the fact that the resolvent $R$ commutes with the reflection operator, which, on the other hand, follows from the invariance of $h^{\rm Pek}$ and $\Pi_i$ under parity. Thus $(y,z) \mapsto \re(w^i_{P,y})(z)$ is even as a function on $\mathbb{R}^6$ for $i=0,1$ while the corresponding imaginary parts are odd on the same space. These facts will be of relevance below where they lead to the vanishing of several integrals.
\end{remark} 

The following lemma is proven in Section \ref{Sec: Remaining Proofs}.

\begin{lemma} \label{lem: bound for w1 and w0} Let $\lambda = \frac{1}{6}\sno \nabla \varphi \sno^2_\2$ and $K_0>0$ large enough. For every $c>0$ there exists a constant $C>0$ such that
\begin{subequations}
\begin{align}
 \sno w_{P,y}^1 \sno^2_\2  +  \sno \widetilde w_{P,y}^1 \sno^2_\2  & \, \le\,  C \big(\alpha^{-2} y^2 + y^4  \big) \label{eq: bound for w perp a} \\[1mm]
\big| \sno w_{P,y}^0 \sno^2_\2 - 2 \lambda y^2 \big|  & \, \le \, C \big( \alpha^{-2} y^2  + y^4 + y^6  \big) \label{eq: bound for w 0 a} \\[1mm]
\big| \sno \widetilde w_{P,y} \sno^2_\2 - 2 \lambda y^2 \big|  & \,  \le\, C \big(\alpha^{-2} y^2 + y^4 + y^6 \big) \label{eq: bound for tilde w}
\end{align}
\end{subequations}
for all $y\in \mathbb R^3$, $|P|/\alpha \le c$, $K \in (K_0,\infty] $ and $\alpha>0$.
\end{lemma}

For $0\le \delta <1$ and $\eta >0$ we introduce the weight function
\begin{align}\label{eq: definition of F}
n_{\delta , \eta  }(y) & \, =\,  \exp\bigg( - \frac{ \eta \alpha^{2(1-\delta)}  \sno \widetilde w_{P,y} \sno^2_\2}{2} \bigg) 
\end{align}
where, for ease of notation, the dependence on $\alpha$, $P$ and $K$ is omitted. Using the arguments laid down in Remark \ref{rem: symmetries_of_w}, it is easy to see that $n_{\delta , \eta }(y)$ is even as a function of $y$. Moreover in the limit of large $\alpha$ the dominant part of the weight function when integrated against suitably decaying functions comes from the term in the exponent that is quadratic in $y$, cf. \eqref{eq: bound for tilde w}. This is a crucial ingredient in our proofs and the content of the next lemma.

\begin{lemma}\label{lem: Gaussian lemma}
Let $\eta_0 > 0$, $c>0$, $\lambda = \frac{1}{6}\sno \nabla \varphi \sno^2_\2$ and $n_{\delta,\eta }(y)$ defined in \eqref{eq: definition of F}. For every $n\in \mathbb N_0$ there exist constants $ d, C_n>0$ such that
\begin{align}\label{eq: Gaussian estimate}
\int |y|^n  g(y)  \, \Big| n_{\delta,\eta }(y) - e^{ - \eta \lambda \alpha^{2(1-\delta)} y^2} \Big| \D y \, \le\,  C_n  \frac{ \sno g \sno_\su  } {\alpha^{(4+n)(1-\delta) + \delta }}  + e^{- d   \alpha^{-2\delta+1} }\sno |\cdot |^n g \sno_\1
\end{align}
for all non-negative functions $g\in L^\infty(\mathbb R^3) \cap L^1(\mathbb R^3)$, $\eta\ge \eta_0$, $\delta \in [0,1) $, $|P|/\alpha \le c $ and all $K,\alpha$ large enough.
\end{lemma}

At first reading, one should think of $n=0$, $\delta = 0$, $\eta = 1$ and $g$ a suitable $\alpha$-independent non-negative function. In this case the integral involving the Gaussian is of order $\alpha^{-3}$ whereas the term on the right hand side is of order $\alpha^{-4}$ and thus contributing a subleading error. The proof of the lemma is given in Section \ref{Sec: Remaining Proofs}. As a direct consequence that will be useful to estimate error terms, we find

\begin{corollary}\label{cor: Gaussian for errors} Given the same assumptions as in Lemma \ref{lem: Gaussian lemma}, for every $n \in \mathbb N_0$ there exist constants $d, C_n >0$ such that
\begin{align}
 \int  |y|^n g (y) n_{\delta,\eta } (y) \D y \, \le\, C_n \frac{ \sno g \sno_\su }{  \alpha^{(3+n)(1-\delta)} } + e^{- d \alpha^{-2\delta+1} } \sno |\cdot|^n g \sno_\1
\end{align}
for all non-negative functions $g\in L^\infty(\mathbb R^3) \cap L^1(\mathbb R^3)$, $\eta \ge \eta_0 $, $\delta \in [0,1) $, $|P|/\alpha \le c$ and all $K,\alpha$ large enough.
\end{corollary}

\begin{proof}[Proof of Corollary \ref{cor: Gaussian for errors}] Since
\begin{align}
\int \D y\, |y|^n  e^{ - \eta \lambda \alpha^{2(1-\delta)} y^2} \, =\,  ( \eta \lambda \alpha^{2(1-\delta)} )^{- \frac{3+n}{2}} \int \D y\, |y|^n e^{-y^2} \, =\,  C_n \alpha^{-(3+n)(1-\delta)},
\end{align}
the statement follows immediately from Lemma \ref{lem: Gaussian lemma}.
\end{proof}

\subsection{Further preliminaries}

\subsubsection{Estimates involving the Pekar minimizers}

\begin{lemma}\label{lemma: props_peks}
Let $\psi>0$ be the rotational invariant unique minimizer of the Pekar functional \eqref{eq: electronic pekar functional}, and let
\begin{align}\label{eq: definition of H}
  H(x)\, :=\, \lsp \psi|T_x\psi\rsp_\2 \, =\, (\psi\ast \psi)(x).
\end{align} 
We have that $\psi$, $|\nabla \psi|$ and $H$ are $L^p(\mathbb{R}^3,(1+|x|^n)\D x)$ functions for all $1\leq p \leq \infty $ and all $n\geq 0$. Moreover, there exists a constant $C>0$ such that for all $x$ we have 
\begin{align}\label{eq: quadratic_bd_on_H}
  |H(x)-1|\, \leq\, Cx^2.
\end{align}
\end{lemma}
\begin{proof}
As follows from \cite{Lieb1977}, $\psi(x)$ is monotone decreasing in $|x|$; moreover, it is smooth and bounded and vanishes exponentially at infinity, i.e. there exists a constant $C0$ such that $\psi(x)\leq Ce^{-|x|/C}$ for all $|x|$ large enough (for the precise asymptotics see \cite{Moroz2017}). This clearly implies the statement for $\psi$. It further implies that all the derivatives of $\psi$ are bounded. Hence, in order to show the desired result for $| \nabla\psi|$, it suffices to show that $\int \D x |x|^n | \nabla\psi (x)|$ is finite for all $n\geq 0$. Since $\psi$ is radial, i.e. there is a function $\psi^{\rm rad}:[0,\infty)\to (0,\infty)$ such that $\psi(x) = \psi^{\rm rad}(|x|)$, and monotone decreasing, we have 
\begin{align}
  \int \D x \, |x|^n | \nabla\psi (x)| & \, =\, -4\pi \int_0^{\infty} \frac{\D \psi^{\rm rad}(r)}{\D r}r^{n+2} \D r \, =\,  (n+2) \int \frac{|\psi(x)|}{|x|}|x|^n \D x \notag \\
   &\, \leq\, 4\pi \left(R_0^{n+2}\sno \psi \sno_\su + \frac{n+2}{R_0} \sno |\cdot|^n\psi\sno_\1 \right)
\end{align} 
for all $R_0>0$. Clearly $H$ is bounded, and hence, by $|x+y|^n\leq 2^{1-n}\left(|x|^n+|y|^n\right)$, we can easily bound
\begin{align}
  \int |x|^n H(x)\D x\, \leq\, 2^{2-n}\sno \psi\sno_\1 \sno |\cdot|^n \psi\sno_\1
\end{align} from which the statement follows also for $H$. To show \eqref{eq: quadratic_bd_on_H}, use the Fourier representation 
\begin{align}
  H(x)\, =\, \int |\widehat{\psi}(k)|^2 \cos(kx)\D k,
\end{align} 
together with $H(x)\leq 1$, $\cos(kx)\geq 1-\frac{(kx)^2}{2}$ and $\nabla \psi \in L^2$.
\end{proof}

The next lemma contains suitable bounds for the potential $V^\varphi$ and the resolvent $R$ introduced in \eqref{eq: def of effective potential}, \eqref{eq: optimal phonon mode} and \eqref{eq: def of resolvent}.

\begin{lemma}\label{lem: bound for R} There is a constant $C>0$ such that
\begin{align}\label{eq: bound potential well}
(V^\varphi)^2 \, \le\, C (1-\Delta), \quad  \pm V^\varphi \, \le\, \frac{1}{2} (-\Delta) + C \quad \text{and} \quad \sno \nabla R^{1/2} \sno_{\op} \le C .
\end{align}
\end{lemma}
\begin{proof}
For the proof of the first two inequalities, we refer to \cite[Lemma III.2]{LeopoldRSS2019}. The bound for the resolvent is obtained through
\begin{align}
0\, \le\,   R^\frac{1}{2} (-\Delta ) R^\frac{1}{2} & \, \le\,  R^\frac{1}{2} h^{\rm Pek}R^\frac{1}{2} - R^\frac{1}{2} (V^\varphi -\lambda^{\rm Pek} ) R^\frac{1}{2} \, \le\,   C R +  \frac{1}{2}R^\frac{1}{2} (-\Delta ) R^\frac{1}{2},
\end{align}
where we made use of the second inequality in \eqref{eq: bound potential well}.
\end{proof}

\subsubsection{The commutator method}

In the course of the proof we are frequently faced with bounding field operators like $\phi(h_x)$. From the standard estimates for creation and annihilation operators, we would obtain
\begin{align}\label{eq: standard estimates for a and a*}
\sno a(f) \Psi \sno_\h \, \le\, \sno f \sno_\2  \sno \mathbb N^{1/2} \Psi \sno_\h, \ \sno a^\dagger(f) \Psi \sno_\h \, \le\, \sno f \sno_\2 \sno (\mathbb N+1)^{1/2} \Psi \sno_\h, \ \Psi \in  \mathscr H,
\end{align}
which is not sufficient since $h_{0}(y)$ is not square-integrable. With the aid of the commutator method introduced by Lieb and Yamazaki \cite{Lieb1958} one obtains suitable upper bounds by using in addition some regularity in the electron variable of the wave function $\Psi$. For our purpose, the version summarized in the following lemma will be sufficient.

\begin{lemma}\label{lem: LY CM} 
Let $h_{K,\cdot}$ for $K\in (1,\infty]$ as defined in \eqref{def: cut off coupling function}, let $A$ denote a bounded operator in $L^2(\mathbb R^3)$ \textnormal{(}acting on the field variable\textnormal{)} and $a^\bullet \in \{ a , a^\dagger \}$. Further let $X,Y$ be bounded symmetric operators in $L^2(\mathbb R^3)$ \textnormal{(}acting on the electron variable\textnormal{)} that satisfy $D_0 := \sno X \sno_{\op} \sno Y \sno_{\op} + \sno \nabla X \sno_{\op}\sno Y \sno_{\op} + \sno X \sno_{\op} \sno \nabla Y \sno_{\op}< \infty $. There exists a constant $C>0$ such that
\begin{subequations}
\begin{align}
\sno X a^\bullet(A h_{K,\cdot + y }) Y \Psi \sno_\h & \, \le\, C D_0 \sno (\mathbb N+1)^{1/2} \Psi \sno_\h \\ \sno X a^\bullet ( A h_{\Lambda ,\cdot+y} - A h_{K,\cdot+y} )  Y \Psi \sno_\h & \, \le\, \frac{C D_0}{\sqrt {K}} \sno (\mathbb N+1)^{1/2} \Psi \sno_\h
\end{align}
\end{subequations}
for all $y\in \mathbb R^3$, $\Psi\in \mathscr H$ and $\Lambda > K >1$.
\end{lemma}
\begin{remark} \label{eq: remark Ah} Note that $A h_{K,\cdot+y}= T_y (Ah_{K,\cdot}) $ and in case that $A$ has an integral kernel,
\begin{align}
(A h_{K,x})(z) = \int \D u\, A(z,u)h_{K,x}(u).
\end{align}  
\end{remark}
\begin{proof}[Proof of Lemma \ref{lem: LY CM}] To obtain the first inequality, write $h_{K,\cdot} = (h_{K,\cdot}-h_{1,\cdot}) + h_{1,\cdot}$ and apply the second inequality (with $\Lambda$ and $K$ interchanged) to the term in parenthesis. The bound for the term involving $h_{1,\cdot}$ follows from \eqref{eq: standard estimates for a and a*}, as
\begin{align}\label{eq: bound for h1 in LY lem}
& \sno a^\bullet(A h_{1,\cdot + y }) Y \Psi \sno^2_\h  =  \int \D x\, \sno a^\bullet(A h_{1 , x + y }) (Y \Psi)(x) \sno^2_{\mathcal F} \\
&  \le  \int \D x \sno A h_{1 , x + y } \sno^2_\2\, \sno (\mathbb N+1)^{1/2} (Y\Psi)(x)\sno^2_{\mathcal F} \le \sno A\sno_{\op}^2  \sno h_{1 , 0} \sno_\2^2   \sno Y \sno_{\op}^2 \sno (\mathbb N+1)^{1/2} \Psi \sno^2_\h .\notag
\end{align}
To verify the second inequality, write the difference as a commutator
\begin{align}\label{eq: commutator identity}
h_{\Lambda ,x}(z) - h_{K ,x}(z) = [-i\nabla_x , j_{K,\Lambda,x}(z)] , \quad  j_{K,\Lambda,x}(z) = \frac{1}{(2\pi )^3} \int\limits_{K \le |k| \le \Lambda} \D k \frac{k e^{ik(x-z)} }{|k|^3}  
\end{align}
and use that $\nabla$ and $A$ commute (they act on different variables). Then similarly as in \eqref{eq: bound for h1 in LY lem} we obtain
\begin{align}
& \sno X a^\bullet([\nabla  , A j_{K,\Lambda,\cdot+y} ] ) Y \Psi \sno_\h \, \le \, \sno X\nabla  a^\bullet( A j_{K,\Lambda,\cdot+y} ) Y \Psi \sno_\h  + \sno X   a^\bullet( A j_{K,\Lambda,\cdot + y } ) \nabla Y \Psi \sno_\h  \notag \notag\\[1mm]
& \quad \quad \quad \quad  \le\,  \sno X\nabla\sno_{\op} \sno a^\bullet( A j_{K,\Lambda,\cdot+y} ) Y \Psi \sno_\h + \sno X \sno_{\op} \sno a^\bullet( A j_{K,\Lambda,\cdot+y } ) \nabla Y \Psi \sno_\h \notag \\[1mm]
& \quad \quad \quad \quad  \le \, \sno A\sno_{\op} \big( \sno X\nabla \sno_{\op} \sno Y \sno_{\op}  + \sno X\sno_{\op} \sno \nabla Y \sno_{\op} \big)\, \sno  j_{K,\Lambda,0} \sno_\2 \,  \sno (\mathbb N+1)^{1/2}  \Psi \sno_\h .
\end{align}
The desired bound now follows from $\sup_{\Lambda> K} \sno j_{K,\Lambda,0} \sno_\2^2 \le C/K$.
\end{proof}
A simple but useful corollary is given by
\begin{corollary}\label{cor: X h Y bounds}
Under the same conditions as in Lemma \ref{lem: LY CM}, with the additional assumption that $Y$ is a rank-one operator, there exists a constant $C>0$ such that
\begin{subequations}
\begin{align}
\int \D z\, \sno X ( A h_{K,\cdot+y})(z) Y  \sno^2_{\op} & \, \le \, C D_0^2 \label{eq: X h Y bound}\\
\int \D z\, \sno X  \big( (A h_{K,\cdot+y})(z) -  (A h_{\Lambda,\cdot + y })(z) \big) Y  \sno_{\op}^2 & \,  \le \, \frac{C D_0^2}{\Lambda}  \label{eq: X h Y difference bound}
\end{align}
for all $y\in \mathbb R^3$ and $\Lambda > K >1$.
\end{subequations}
\end{corollary}
\begin{proof} Since $Y$ has rank one, we can use
\begin{align}
\int \D z\, \sno X ( A h_{K,\cdot+y})(z) w \sno^2_\2 & \, = \, \sno X a^\dagger(Ah_{K,\cdot+y}) w\otimes \Omega\sno^2_\h ,
\end{align}
for any $w\in L^2(\mathbb R^3)$, and similarly for \eqref{eq: X h Y difference bound}, and apply Lemma \ref{lem: LY CM}.
\end{proof}

\subsubsection{Transformation properties of $\mathbb U_K$ \label{sec: transf prop UK}}

The next lemma collects relations for the Bogoliubov transformation $\mathbb U_K$ defined in \eqref{eq: def of U}. Its proof follows directly from this definition and the fact that $\Theta_K= (H^{\rm Pek}_K)^{1/4}$ is real-valued. We omit the details.

\begin{lemma} \label{lem: U W U transformation}
Let $f \in L^2(\mathbb R^3)$, $f^0 = \Pi_0 f$, $f^1 = \Pi_1 f$ with $\Pi_i$ defined in \eqref{eq: def of Pi_i} and set
\begin{subequations}
\begin{align}
 \underline f \, & =\,   f^0 + \Theta^{-1}_{K} \re (f^1) + i \Theta_K \im(f^1) \label{eq: underline notation} \\[1mm]
 \widetilde {f} \, & =\,  f^0  + \Theta^{\tc}_K \re(f^1) + i \Theta_K^{-1} \im (f^1 ) .
\end{align}
\end{subequations}
The unitary operator $\mathbb U_K$ defined in \eqref{eq: def of U} satisfies the relations
\begin{subequations}
\begin{align}
\mathbb U_K a(f) \mathbb U_K^\dagger & \, =\,  a( f^0 ) + a (A_K f^1 ) + a^\dagger (B_K \overline{  f^1 } ) \label{eq: U again} \\[1.5mm]
\mathbb U ^\dagger_K a(f) \mathbb U^{\tc}_K & \, =\,  a( f^0 ) + a (A_K f^1 ) - a^\dagger (B_K \overline{  f^1 } ) \label{eq: U reverse}\\[1.5mm]
\label{eq: transformation of phi}
\mathbb U^{\tc}_K \phi(f) \mathbb U^\dagger_K & \, =\,  \phi(\underline f) , \quad 
 \mathbb U^{\tc}_K \pi(f) \mathbb U^\dagger_K  \, =\,  \pi(\widetilde{f}) \\[1.5mm]
\mathbb U_K W(f) \mathbb U_K^\dagger & \, = \,  W( \widetilde f).\label{eq: transformation of Weyl}
\end{align}
\end{subequations}
\end{lemma}
The following statements provide helpful bounds involving the number operator when transformed with the Bogoliubov transformation.

\begin{lemma}\label{lem: bounds for the number operator} There exists a constant $b>0$ such that 
\begin{align}
\mathbb U^{\tc}_K (\mathbb N +1 )^n \mathbb U^{\dagger}_K \, \le \, b^n  n^n  (\mathbb N+1)^n ,\quad \, \mathbb U_K^\dagger (\mathbb N +1 )^n \mathbb U^{\tc}_K \, \le \, b^n  n^n  (\mathbb N+1)^n
\end{align}
for all $n\in \mathbb N $ and $K\in (K_0,\infty]$ with $K_0$ large enough.
\end{lemma}
\begin{proof} With $b$ replaced by $b_K = 2 \sno B_K \sno_{\HS}^2 + \sno A_K\sno_{\rm {op}}^2 + 1 $, both estimates follow from \cite[Lemma 4.4]{Bossmann2019} together with \eqref{eq: U again} and \eqref{eq: U reverse}. That $b_K\le b $ for some $K$-independent $b>0$ is inferred from Lemma \ref{lem: regularized Hessian}.
\end{proof}
In the next two statements we denote by $\mathbbm{1}(\mathbb N > c)$ (resp. $\mathbbm{1}(\mathbb N \le c)$) the orthogonal projector in $\mathcal F$ to all states with phonon number larger than (resp. less or equal to) $c$.

\begin{corollary}\label{Cor: Upsilon estimates} Let $ \Upsilon_K  = \mathbb U_K^\dagger \Omega$ and $\Upsilon_K^{>} :=  \mathbbm{1}(\mathbb N > \alpha^\delta )  \Upsilon_K $ for $\delta > 0$. There exist constants $b,C_{\delta,n} > 0$ such that
\begin{subequations}
\begin{align}
\lsp \Upsilon_K | (\mathbb N+1)^n \Upsilon_K \rsp_\Fock & \, \le\,  b^n n^n \label{eq: bound for Y<} \\[1mm]
\lsp \Upsilon_K^> | (\mathbb N+1)^n \Upsilon_K^> \rsp_\Fock  &\,  \le\, C_{\delta,n}\, \alpha^{- 20 }.\label{eq: exp bound for tail}
\end{align}
for all $n\in \mathbb N_0$ and all $K\in (K_0,\infty]$ with $K_0$ large enough.
\end{subequations}
\end{corollary} 
\begin{proof} The first bound follows directly from Lemma \ref{lem: bounds for the number operator}. The second one is obtained from
\begin{align}
\lsp \Upsilon_K^> | (\mathbb N+1)^n \Upsilon_K^> \rsp_\Fock \, & \le \, \sno \mathbb N^{m} (\mathbb N+1)^{n} \Upsilon_K^{ >} \sno_\Fock  \,  \sno \mathbb N^{-m}\Upsilon_K^{ >} \sno_\Fock \notag\\[1mm]
& \,  \le\, \sno (\mathbb N + 1) ^{n+m} \Upsilon_K  \sno_\Fock \, \alpha^{- m \delta }   \, \le \, (2 (n+m) b)^{n+m} \alpha^{- m \delta }
\end{align}
with $m \ge 20 / \delta$.
\end{proof}
\begin{lemma}\label{lem: exp N bound} For $\delta>0$ and $\kappa = 1/ (16 e b \alpha^\delta)$ with $b > 0$ the constant from Lemma \ref{lem: bounds for the number operator}, the operator inequality
\begin{align}
\mathbbm{1}(\mathbb N \le 2 \alpha^\delta) \mathbb U_K^\dagger \exp( 2\kappa \mathbb N ) \mathbb U_K^{{\color{white}{\dagger}}} \mathbbm{1}(\mathbb N \le 2 \alpha^\delta) \, \le \, 2 
\end{align}
holds for all $K,\alpha $ large enough.
\end{lemma}

\begin{proof} We write out the Taylor series for the exponential and invoke Lemma \ref{lem: bounds for the number operator},
\begin{align}
 \mathbbm{1}(\mathbb N \le 2 \alpha^\delta) \mathbb U_K^\dagger e^{2\kappa \mathbb N } \mathbb U_K^{{\color{white}{\dagger}}} \mathbbm{1}(\mathbb N \le 2 \alpha^\delta) & \, =\,  \sum_{n=0}^\infty  \frac{(2 \kappa )^n}{n!}\mathbbm{1}(\mathbb N \le 2 \alpha^\delta) \mathbb U_K^\dagger (\mathbb N+1)^n  \mathbb U_K^{{\color{white}{\dagger}}} \mathbbm{1}(\mathbb N \le 2 \alpha^\delta) \notag\\
& \,  \le \,  \sum_{n=0}^\infty  \frac{(2\kappa b n )^n}{n!} \mathbbm{1}( \mathbb N \le 2 \alpha^\delta) \mathbb (\mathbb N+1)^n  \mathbbm{1}(\mathbb N \le 2 \alpha^\delta)  \notag\\
& \, \le\,   \sum_{n=0}^\infty  \frac{(8 \alpha^\delta \kappa b n )^n}{n!}
\label{eq: bound for error in norm}
\end{align}
where we used $1\le 2\alpha^\delta $ in the last step. The stated bound now follows from the elementary inequality $n! \ge (\frac{n}{e})^n$.
\end{proof}

The reason for introducing the momentum cutoff in $\mathbb H_K$ is to obtain a finite upper bound for the norm of the state $P_f \Upsilon_K $. This is the content of the next lemma, whose proof is given in Section \ref{Sec: Remaining Proofs}.

\begin{lemma}\label{lem: bounds for P_f} Let $P_f = \int \D k \, k\, a_k^\dagger a_k$ and $K_0$ large enough. There is a $C>0$ such that 
\begin{align}
\lsp \Omega | \mathbb U_K^{\tc} (P_f)^2 \mathbb U_K^\dagger  \Omega \rsp_\Fock \, \le\, C K
\end{align}
for all $K \in (K_0, \infty)$.
\end{lemma}

\subsection{The norm \label{sec: norm}}

In this section we provide the computation of the norm $\mathcal N = \sno \Psi_{K,\alpha}(P)\sno^2_\Fock$.

\begin{proposition} \label{prop: norm bound} 
Let $\lambda = \frac{1}{6} \sno \nabla \varphi \sno^2_\2$ and $c>0$. For every $\varepsilon >0$ there exist a constant $C_\varepsilon > 0$ \textnormal{(}we omit the dependence on $c$\textnormal{)} such that
\begin{align}
\bigg| \, \mathcal N - \bigg(\frac{\pi}{\lambda  \alpha^2}\bigg)^{3/2}  \bigg| & \, \le\,   C_\varepsilon  \sqrt K \alpha^{-4 + \varepsilon} 
\end{align}
for all $|P|/\alpha \le c$ and all $K,\alpha$ large enough.
\end{proposition}

\begin{proof} It follows from \eqref{eq: energy derivation 2} and \eqref{eq: energy derivation 2b} that $\mathcal N  =  \mathcal N_{0} + \mathcal N_1  + \mathcal N_2 $ with
\begin{subequations}
\begin{align}
\mathcal N_0 & \, =\,   \int \D y\, \lsp G_{K }^{0} \I  T_y e^{ A_{P,y}} W(\alpha w_{P,y})  G_{K}^{0} \rsp_\h \label{eq: norm line 1}\\
\mathcal N_{1} & \, =\,  - \frac{2}{\alpha}  \int \D y\,  \re \lsp G_{K }^{0} \I  T_y e^{ A_{P,y}  } W(\alpha w_{P,y}) G_{K}^{1} \rsp_\h  \label{eq: norm line 2}\\
\mathcal N_{2}  & \, =\,   \frac{1}{\alpha^2}  \int \D y\, \lsp G_{K}^{1} \I T_y  e^{ A_{P,y}  } W(\alpha w_{P,y}) G_{K }^{1} \rsp_\h \label{eq: norm line 3}.
\end{align}
\end{subequations}
\noindent \underline{Term $\mathcal N_0$.} This part contains the leading order contribution $(\frac{\pi}{\lambda  \alpha^2})^{3/2}$. With $H$ defined in \eqref{eq: definition of H}, let us write
\begin{align}
 \mathcal N_0 & \, = \, \int \D y\,  H(y) \lsp   \Upsilon_K \I   W(\alpha w_{P,y})    \Upsilon_K \rsp_\Fock \notag  \\
&\quad \, +\,  \int \D y\, H(y)  \lsp \Upsilon_K \I (e^{A_{P,y}}-1) W(\alpha w_{P,y})  \Upsilon_K \rsp_\Fock = \mathcal N_{01}  + \mathcal N_{02}.
\end{align}
In the first term we use $\Upsilon_K  = \mathbb U_K^\dagger  \Omega  $ and apply \eqref{eq: transformation of Weyl} to transform the Weyl operator with the Bogoliubov transformation. This gives
\begin{align}\label{eq: identity for tilde w transformation}
\mathbb U_K W( \alpha w_{P,y} ) \mathbb U_K^\dagger & \, = \,  W( \alpha \widetilde w_{P,y})
\end{align}
with $\widetilde w_{P,y}$ defined in \eqref{eq: def of tilde w}. From \eqref{eq: BCH for Weyl} and \eqref{eq: definition of F}, we thus obtain
\begin{align}
\mathcal N_{01} \, = \, \int \D y\,  H(y)  \lsp \Omega \I   W(\alpha \widetilde w_{P,y})  \Omega \rsp_\Fock  & \, = \, \int \D y\, H(y) n_{0,1}(y).
\end{align}
Since $\sno H\sno_\1 + \sno H\sno_\su \le C $, cf. Lemma \ref{lemma: props_peks}, we can apply Lemma \ref{lem: Gaussian lemma} in order to replace the weight function $n_{0,1}(y)$ by the Gaussian $e^{-\lambda \alpha^2 y^2}$. More precisely,
\begin{align}
\bigg | \int \D y\, H(y) n_{0,1}(y) - \int \D y\, H(y) e^{- \lambda \alpha^2 y^2 } \bigg| \, \le\,  C\alpha^{-4}
\end{align}
for all $|P| /\alpha \le c$ and all $K,\alpha$ large enough. Then we use $| H(y)-1| \le C y^2$ in order to obtain
\begin{align}
\bigg|\, \mathcal N_{01}- \bigg( \frac{\pi}{\lambda \alpha^2 }\bigg)^{3/2}\bigg|\, \le \, C\alpha^{-4}.
\end{align}

To treat $\mathcal N_{02}$ it is convenient to decompose the state $ \Upsilon_K $ into a part with bounded particle number and a remainder. To this end, we choose a small $\delta>0$ and write
\begin{align} \label{eq: decomposition of Upsilon}
\Upsilon_K  \, = \,  \Upsilon_K^{<}  +  \Upsilon_K^{>} \, = \,  \mathbbm{1}(\mathbb N \le \alpha^\delta )  \Upsilon_K  +  \mathbbm{1}(\mathbb N > \alpha^\delta ) \Upsilon_K .
\end{align}
Inserting this into $\mathcal N_{02}$ and using unitarity of $e^{A_{P,y}}$ and $\sno H\sno_\1\le C$, we can estimate 
\begin{align}
| \mathcal N_{02}|\, & \le\, \int \D y\, H(y)   \I \lsp  \Upsilon_K^{ <} \I (e^{A_{P,y}} -1 ) W(\alpha w_{P,y}) \Upsilon_K \rsp_\Fock \I + C \sno \Upsilon_K^>\sno_\Fock.
\end{align}
By Corollary \ref{Cor: Upsilon estimates} for $n=0$, $\sno \Upsilon_K^>\sno \le C_\delta \, \alpha^{-10}$. In the remaining expression we use \eqref{eq: identity for tilde w transformation},
\begin{align}\label{eq: inserting identity}
\lsp  \Upsilon_K^{ <} \I (e^{A_{P,y}} -1 ) W(\alpha w_{P,y})  \Upsilon_K \rsp_\Fock \, = \,  \lsp  \Upsilon_K^< \I (e^{A_{P,y}}-1) \mathbb U_K^\dagger W(\alpha \widetilde w_{P,y})  \Omega \rsp_\Fock ,
\end{align}
and insert the identity
\begin{align} \label{eq: number op identity}
\mathbbm{1} \, = \, e^{\kappa \mathbb N } e^{- \kappa \mathbb N } \quad \text{with }\quad \kappa \, = \,  \frac{1}{16 e b \alpha^\delta}
\end{align}
on the left of the Weyl operator (where $b>0$ is the constant from Lemma \ref{lem: bounds for the number operator}). After applying the Cauchy--Schwarz inequality, this leads to
\begin{align} \label{eq: bound for the norm error}
& \I \lsp  \Upsilon_K^{ <} \I (e^{A_{P,y}} -1 ) W(\alpha w_{P,y}) \Upsilon_K \rsp_\Fock \I  \notag\\[1mm]
& \hspace{2.5cm} \le \, \sno e^{\kappa \mathbb N} \mathbb U_K (e^{-A_{P,y}}-1) \Upsilon_K^< \sno_\Fock  \sno e^{-\kappa \mathbb N } W(\alpha \widetilde w_{P,y})  \Omega \sno_\Fock.
\end{align}
In the second factor we then employ
\begin{align}
\sno e^{- \kappa \mathbb N }  W (\alpha \widetilde w_{P,y})  \Omega  \sno_\Fock \, = \, 
e^{-\frac{\alpha^2}{2} \sno \widetilde w_{P,y}\sno^2_\2} \sno e^{- \kappa \mathbb N }  e^{ a^\dagger(\alpha \widetilde w_{P,y})} e^{\kappa \mathbb N } \Omega \sno_\Fock
\end{align}
and use $e^{-\kappa \mathbb N}a^\dagger(f)e^{\kappa \mathbb N} = a^\dagger(e^{-\kappa} f) $ to write
\begin{align}\label{eq: comm expN and a dagger}
e^{- \kappa \mathbb N }  e^{a^\dagger(\alpha \widetilde w_{P,y})} e^{\kappa \mathbb N } \Omega   \, = \, e^{a^\dagger(e^{-\kappa } \alpha \widetilde w_{P,y})}   \Omega   \, = \, e^{\frac{\alpha^2 e^{-2 \kappa}  }{2} \sno \widetilde w_{P,y}\sno^2_\2}  W (e^{-\kappa } \alpha \widetilde w_{P,y})  \Omega  .
\end{align} 
Combining the previous two lines we obtain
\begin{align}\label{eq: bound for epsilon N}
\sno e^{- \kappa \mathbb N } W (\alpha \widetilde w^{\tc}_{P,y}) \Omega \sno_\Fock \, = \, \exp\Big( -\frac{\alpha^2 }{2} (1-e^{-2\kappa} )  \sno \widetilde w_{P,y}\sno^2_\2 \Big)  \le n_{\delta, \eta}(y)
\end{align}
for some $\alpha$-independent $\eta > 0 $ and $\alpha$ large enough. To estimate the first factor in \eqref{eq: bound for the norm error}, we apply Lemma \ref{lem: exp N bound} (note that $(e^{A_{P,y}} - 1) \Upsilon_K^<  \in \text{Ran}(\mathbbm{1} (\mathbb N\le 2\alpha ^\delta ))$)
\begin{align}
\sno e^{\kappa \mathbb N} \mathbb U_K (e^{ - A_{P,y}}-1)  \Upsilon_K^< \sno_\Fock \le \sqrt 2 \sno (e^{ - A_{P,y}}-1) \Upsilon_K\sno_\Fock.
\end{align}
On the right side we use the functional calculus for self-adjoint operators
\begin{align}
\sno (e^{- A_{P,y} }-1) \Upsilon_K\sno_\Fock  \le \sno A_{P,y} \Upsilon_K\sno_\Fock  \le  \sno (y P_f) \Upsilon_K \sno_\Fock +  |g_{P}(y)| \le  C \big( \sqrt K |y|  +  \alpha|y|^3 \big),
\end{align}
where in the last step we applied Lemma \ref{lem: bounds for P_f} and used
\begin{align}\label{eq: bound for g_P}
|g_{P}(y)|\, \le  \, C \alpha |y|^3,
\end{align}
which is inferred from \eqref{eq: definition AP} using $\sno \Delta \varphi\sno_\2< \infty$. Returning to \eqref{eq: bound for the norm error} we have shown that 
\begin{align}
|\mathcal N_{02} | \, \le\, C \int \D y \, H(y)  \big(   \sqrt{K} |y| +  \alpha |y|^3 \big) n_{\delta , \eta} (y) + C_\delta\, \alpha^{-10} ,
\end{align}
and hence we are in a position to apply Corollary \ref{cor: Gaussian for errors}. This implies for all $K,\alpha$ large
\begin{align}\label{eq: example error bound}
|\mathcal N_{02} | \, \le \,  C\big( \sqrt K \alpha^{-4 (1-\delta) } + \alpha^{-6(1-\delta)+1} \big)  + C_\delta\, \alpha^{-10} \le C_\delta\, \sqrt K  \alpha^{-4(1-\delta) }.
\end{align}
\noindent \underline{Term $\mathcal  N_1$}. We start by inserting \eqref{eq: alternative definition of G} for $G_K^1$ in expression \eqref{eq: norm line 2}. Since the Weyl operator commutes with $u_\alpha$, $R$ and $\PP = |\psi \rangle \langle \psi | $, we can apply \eqref{eq: W phi W} to obtain 
\begin{align}\label{eq: identity for WG1}
W(\alpha  w_{P,y} ) G_K^1\, =\,  u_\alpha R \big( \phi(h_{K,\cdot}^1) + 2 \alpha \langle h_{K,\cdot} | \re ( w_{P,y}^1 ) \rangle_\2 \big) \PP W(\alpha  w_{P,y}) G_K^0,
\end{align}
where we used that $h_{K,x}$ is real-valued. Note that $\langle h_{K,\cdot} | \re( w_{P,y}^1 ) \rangle_\2 $ is a $y$-dependent multiplication operator in the electron variable. With $( T_y e^{A_{P,y}})^\dagger =  T_{-y} e^{-A_{P,y}}$ and \eqref{eq: decomposition of Upsilon}, we can thus write
\begin{align}
\mathcal N_1   & \, = \, - \frac{2}{\alpha} \int \D y \, \re \lsp R_{1,y} \psi \otimes \big( \Upsilon_K^< + \Upsilon_K^> \big) \I W(\alpha  w_{P,y})  G_K^0 \rsp_\h = \mathcal N_{1}^< + \mathcal N_{1}^>,\label{eq: N11 plus N12}
\end{align}
where we introduced the operator $R_{1,y} = R_{1,y}^1 + R_{1,y}^2$ with
\begin{subequations}
\begin{align}\label{eq: R_1y^1}
R_{1,y}^1 & \, = \, \PP \phi(h_{K,\cdot}^1) R u_\alpha T_{-y} \PP e^{-A_{P,y}} ,\\[1mm]
R_{1,y}^2 & \, = \, 2\alpha  \PP \lsp h_{K,\cdot} | \re (w_{P,y}^1) \rsp_\2  R u_\alpha T_{-y} \PP e^{-A_{P,y}} .\label{eq: R_1y^2}
\end{align}
\end{subequations}
Using Lemma \ref{lem: LY CM} in combination with $\sno \nabla P_\psi \sno_{\op} + \sno \nabla R^{1/2}\sno_{\op} <\infty $, see Lemmas \ref{lemma: props_peks} and \ref{lem: bound for R}, we can bound the first operator, for any $\Psi \in \mathscr H$, by
\begin{align}
\sno R_{1,y}^1 \Psi \sno_\h &  \le  C  \sno (\mathbb N+1)^{1/2} u_\alpha T_{-y}  \PP e^{-A_{P,y}}  \Psi \sno_\h \le  C  \sno u_\alpha T_{-y} \PP \sno_{\op} \sno (\mathbb N+1)^{1/2} \Psi \sno_\h.
\end{align}
To estimate the second operator, we write out the inner product, use Cauchy--Schwarz twice, apply Corollary \ref{cor: X h Y bounds} (with $A=1$, $X=R$ and $Y=P_\psi$) and use \eqref{eq: bound for w perp a},
\begin{align}\label{eq: R12 bound}
 \sno R_{1,y}^2 \Psi \sno^2_\h \, & =  \, 4\alpha^2 \sno \int \D z\, \re (w_{P,y}^1(z)) \PP h_{K,\cdot}(z) R  u_\alpha T_{-y} \PP  e^{-A_{P,y}} \Psi  \sno^2_\h \notag\\
& \le \, 4 \alpha^2 \int\D u\, | w_{P,y}^1(u)|^2 \int \D z\, \sno \PP h_{K,\cdot}(z) R \sno_{\op}^2 \, \sno   u_\alpha T_{-y}P_\psi  e^{-A_{P,y}} \Psi  \sno^2_\h \notag\\[1mm]
& \le \, C \alpha^2 \sno w_{P,y}^1 \sno^2_\2  \sno u_\alpha T_{-y} \PP  e^{-A_{P,y}} \Psi  \sno^2_\h \notag \\[2.5mm]
& \le \,  C \alpha^2 (y^4 + \alpha^{-4}) \sno u_\alpha T_{-y} \PP \sno_{\op}^2 \sno \Psi \sno^2_\h .
\end{align}
Combining the above estimates we arrive at
\begin{align}
\sno R_{1,y} \Psi \sno_\h  & \, \le\, C  \sno u_\alpha T_{-y} \PP  \sno_{\op} (1+ \alpha y^2) \sno   (\mathbb N+1)^{1/2} \Psi \sno_\h
  \label{eq: bound for R_1}.
\end{align}
Since $\psi(x)$ decays exponentially for large $|x|$, the function $f_{\alpha}(y) := \sno u_\alpha T_{-y} \PP \sno_{\op}   $ satisfies
\begin{align}\label{eq: bound for f alpha}
\sno  |\cdot |^n f_{\alpha} \sno_\1 \, \le \,  \int \D y\, |y|^n \bigg( \int \D x\, \psi(x+y)^2 u_\alpha(x)^2 \bigg)^{1/2} \, \le\,  C_n \alpha^{3+n}\ \ \text{for all}\ n \in \mathbb N_0.
\end{align}
With this at hand we can estimate the part containing the tail. Invoking Corollary \ref{Cor: Upsilon estimates}
\begin{align}
| \mathcal N_{1}^> | \, \le \, \frac{C}{\alpha}  \sno (\mathbb N+1)^{1/2} \Upsilon_K^> \sno_\Fock   \int \D y\,  f_{\alpha} (y) (1+\alpha y^2) \, \le \, C_\delta\, \alpha^{-5} .
\end{align}
To estimate the first term in \eqref{eq: N11 plus N12}, we proceed similarly as in the bound for $\mathcal N_{02}$. We insert the identity \eqref{eq: number op identity}, apply Cauchy--Schwarz and employ \eqref{eq: bound for epsilon N}. This leads to
\begin{align}
\vert \mathcal N_{1}^< \vert &\,  \le\,  \frac{2}{\alpha}  \int \D y \,
\sno e^{\kappa \mathbb N} \mathbb U^{\tc}_K ( e^{- A_{P,y}}  R_{1,y} \psi \otimes \Upsilon_K^< )\sno_\Fock \, \sno e^{- \kappa \mathbb N } W(\alpha \widetilde w_{P,y}) \Omega  \sno_{\Fock} \nonumber \\
& \, \le\,  \frac{2}{\alpha}  \int \D y \,  \sno e^{\kappa \mathbb N} \mathbb U^{\tc}_K ( e^{ - A_{P,y}} R_{1,y}  \psi \otimes  \Upsilon_K^< ) \sno_\Fock \, n_{\delta,\eta}(y) .
\end{align}
In the remaining norm we use the fact that $R_{1,y}$ changes the number of phonons at most by one, and thus we can apply Lemma \ref{lem: exp N bound} and \eqref{eq: bound for R_1}, together with \eqref{eq: bound for Y<}, to get
\begin{align}
\sno e^{\kappa \mathbb N} \mathbb U^{\tc}_K  (e^{-A_{P,y}}  R_{1,y} \psi \otimes \Upsilon_K^< ) \sno_\Fock \, \le \, \sqrt 2 \sno  R_{1,y} \psi \otimes \Upsilon_K^< \sno_\Fock \,\le \, C f_{\alpha} (y) \big( 1+ \alpha y^2\big) .
\end{align}
With Corollary \ref{cor: Gaussian for errors}, \eqref{eq: bound for f alpha} and $\sno f_\alpha  \sno_\su \leq 1 $, this leads to
\begin{align}
|\mathcal N_{1}^<  | & \, \le \, \frac{C}{\alpha} \int \D y\,  f_{\alpha} (y) \big( 1 + \alpha y^2 \big) \, n_{\delta,\eta}(y)
\,  \le\,  C \alpha^{-1-3(1- \delta)}.
\end{align}
\noindent \underline{Term $\mathcal  N_{2}$}. The strategy for estimating this term is similar to the one for $\mathcal N_1$. Proceeding as described before \eqref{eq: N11 plus N12}, one obtains
\begin{align}
\mathcal  N_2 & \, = \, \frac{1}{\alpha^2} \int \D y \, \lsp  R_{2,y} \psi \otimes \big( \Upsilon_K^< + \Upsilon_K^> \big) \I   W(\alpha w_{ P ,y} )  G_K^0  \rsp_\h \,  = \,  \mathcal  N_{2}^< + \mathcal  N_{2}^> \label{eq: Norm term N2}
\end{align}
with $R_{2,y} = R_{2,y}^1 + R_{2,y}^2$ and
\begin{subequations}
\begin{align}
 R_{2,y}^1 &\,  =\,  \PP \phi(h^1_{K,\cdot}) R e^{-A_{P,y}} u_\alpha T_{-y} u_\alpha R \phi(h^1_{K,\cdot}) \PP , \\
 R_{2,y}^2 & \, = \, 2 \alpha \PP \langle h_{K,\cdot } | \re (w^1_{P,y})\rangle_\2 R e^{-A_{P,y}} u_\alpha T_{-y} u_\alpha R \phi(h^1_{K,\cdot}) \PP .
\end{align}
\end{subequations}
It follows in close analogy as for $R_{1,y}$ in \eqref{eq: R_1y^1}--\eqref{eq: R_1y^2} that given any $\Psi \in \mathscr H$,
\begin{align}
\sno  R_{2,y}  \Psi \sno_\h & \, \le \, C \sno u_\alpha T_{-y} u_\alpha \sno_{\op} (1+\alpha y^2)  \sno (\mathbb N+1)  \Psi \sno_\h,
\end{align}
and since $\sno u_\alpha T_{-y} u_\alpha \sno_{\op} \le \mathbbm{1} (|y|\le 4\alpha)$, we can use Corollary \ref{Cor: Upsilon estimates} to estimate
\begin{align}
|\mathcal  N_{2}^> | \, \le \, \frac{C}{\alpha^2} \sno (\mathbb N+1) \Upsilon_K^> \sno_{\Fock} \int \D y \, \mathbbm{1}(|y|\le 4 \alpha) ( 1 + \alpha y^2 ) \,  \le\,  C_\delta\, \alpha^{-6 }.
\end{align}
To bound the first term in \eqref{eq: Norm term N2} we proceed similarly as for $\mathcal N_{01}$,
\begin{align}
|\mathcal  N_{2}^< | & \le  \alpha^{-2} \int \D y\, \sno e^{\kappa \mathbb N } \mathbb U_K (R_{2,y} \psi \otimes \Upsilon_K^< ) \sno_\Fock \, \sno e^{-\kappa \mathbb N } W(\alpha \widetilde w_{P,y}) \Omega \sno_\Fock \notag \\
& \le \frac{\sqrt 2}{\alpha^2} \int \D y\, \sno  R_{2,y} \psi \otimes \Upsilon_K^<\sno_\h \, n_{ \delta, \eta }(y) \le  \frac{C}{\alpha^2} \int \D y\, \mathbbm{1} (|y|\le 4\alpha) (1 + \alpha y^2  ) \, n_{ \delta, \eta }(y).
\end{align}
The last integral is estimated again via Corollary \ref{cor: Gaussian for errors}, and thus $|\mathcal  N_{2}^<  | \le C \alpha^{-5+3\delta}$. 

Collecting all relevant estimates and choosing $\delta>0$ small enough completes the proof of the proposition.
\end{proof}

\subsection{Energy contribution $\mathcal E$}

In this section we prove the following estimate for the energy contribution $\mathcal E$ defined in \eqref{eq: E}.

\begin{proposition} \label{prop: bound for E} Let $\mathbb N_1 = \D\Gamma(\Pi_1)$ and choose $c>0$. For every $\varepsilon>0$ there is a constant $C_\varepsilon>0$ \textnormal{(}we omit the dependence on $c$\textnormal{)} such that
\begin{align}
\bigg|\, \mathcal E - \frac{1}{\alpha^2}  \bigg( \lsp \Upsilon_K | \mathbb N_1 \Upsilon_K \rsp_\Fock - \frac{3}{2} \bigg)\, \mathcal N\,  \bigg| \,  \le \,  C_\varepsilon \sqrt K \alpha^{-6 + \varepsilon} 
\end{align}
for all $|P|/ \alpha  \le c $ and $K,\alpha$ large enough.
\end{proposition}

\begin{proof}
Since $G_K^0 = \psi \otimes \Upsilon_K$, $ h^{\rm Pek} \psi = 0$ and $\mathbb N\Upsilon_K = \mathbb N_1 \Upsilon_K  $, one has
\begin{align}
\mathcal E & \,  =\,   \int \D y\,  \lsp G_{K}^{0}| \big(\alpha^{-2} \mathbb N_1 + \alpha^{-1} \phi(h_\cdot + \varphi_P) \big) T_y e^{A_{P,y}} W(\alpha w_{P,y})  | G_{K}^{ 0} \rsp_\h  \, =\,  \mathcal E_1 + \mathcal E_2,
\end{align}
where both terms provide contributions to the energy of order $\alpha^{-2}$.\medskip

\noindent \underline{Term $\mathcal E_1$}. Recall that $H(y) = \langle \psi | T_y \psi \rangle_\2 $ and use this to write
\begin{align}
\mathcal E_1 & \, =\,  \frac{1}{\alpha^2} \int \D y\, H(y)  \lsp \Upsilon_K | \mathbb N_1  W(\alpha w_{P,y})  \Upsilon_K \rsp_\Fock  \notag \\
& \quad + \frac{1}{\alpha^{2}} \int \D y\, H(y)   \lsp \Upsilon_K | \mathbb N_1 (e^{A_{P,y}} -1 ) W(\alpha w_{P,y})  \Upsilon_K \rsp_\Fock \,  =\,  \mathcal E_{11} + \mathcal E_{12}. \label{eq: E11 + E12}
\end{align}
With \eqref{eq: identity for tilde w transformation}, \eqref{eq: Weyl identities} and \eqref{eq: definition of F} it follows that
\begin{align}
W(\alpha w_{P,y}) \Upsilon_K \, =\, \mathbb U_K^\dagger W(\alpha \widetilde w_{P,y}) \Omega \,  =\,  n_{0,1}(y) \mathbb U_K^\dagger\, e^{a^\dagger( \alpha w_{P,y}^0) } e^{a^\dagger( \alpha \widetilde w_{P,y}^1 ) } \Omega ,
\end{align}
and since $e^{a^\dagger( \alpha w_{P,y}^0) }$ commutes with $\mathbb U_K \mathbb N_1 \mathbb U_K^\dagger$ and $e^{a( \alpha w_{P,y}^0) } \Upsilon_K = \Upsilon_K $ (we use $\mathbb U_K a^\dagger(f^0) \mathbb U_K^\dagger = a^\dagger (f^0)$ for $f^0 \in \text{Ran}(\Pi_0)$), this leads to
\begin{align}
\mathcal E_{11} & \, =\,  \frac{1}{\alpha^2} \int \D y\, H(y) n_{0,1}(y) \lsp \Omega | \mathbb U_K \mathbb N_1  \mathbb U_K^\dagger e^{a^\dagger( \alpha \widetilde w_{P,y}^1 ) } \Omega \rsp_\Fock.
\end{align}
Because $\mathbb U_K \mathbb N_1  \mathbb U_K^\dagger $ is quadratic in creation and annihilation operators, we can expand the exponential in the inner product and use that only the zeroth and second order terms give a non-vanishing contribution,
\begin{align}
\mathcal E_{11} \, & = \, \frac{1}{\alpha^2} \int \D y\, H(y) n_{0,1}(y) \lsp \Upsilon_K | \mathbb N_1  \Upsilon_K \rsp_\Fock \notag\\
& \ + \, \frac{1}{2\alpha^2} \int \D y\, H(y) n_{0,1}(y)  \lsp \Upsilon_K | \mathbb N_1  \mathbb U_K^\dagger  a^\dagger( \alpha \widetilde w_{P,y}^1 ) a^\dagger( \alpha \widetilde w_{P,y}^1 )  \Omega \rsp_\Fock = \mathcal E_{111} + \mathcal E_{112}.
\end{align}
Next we add and subtract the Gaussian to separate the leading-order term,
\begin{align}
\mathcal E_{111} & \, =\,  \frac{1}{\alpha^2} \int \D y\, H(y) \, e^{- \lambda \alpha^2 y^2} \lsp \Upsilon_K | \mathbb N_1  \Upsilon_K \rsp_\Fock \notag\\
& \quad +  \frac{1}{\alpha^2} \int \D y\, H(y) \big( n_{0,1}(y) - e^{- \lambda \alpha^2 y^2}  \big) \lsp \Upsilon_K | \mathbb N_1  \Upsilon_K \rsp_\Fock \,  =\,  \mathcal E_{111}^{\rm lo} + \mathcal E_{111}^{\rm err} . \label{eq: add subtract gaussian for E111}
\end{align}
In $\mathcal E_{111}^{\rm lo}$ we use $| H(y) - 1 | \le C y^2$ and Corollary \ref{Cor: Upsilon estimates} to replace $H(y)$ by unity at the cost of an error of order $\alpha^{-7}$. In the term where $H(y)$ is replaced by unity, we perform the Gaussian integral and use Proposition \ref{prop: norm bound} and again Corollary \ref{Cor: Upsilon estimates}. This leads to
\begin{align}\label{eq: bound for E111 lo}
\Big|\, \mathcal E_{111}^{ \rm lo} -  \mathcal N \frac{1}{\alpha^2}  \lsp \Upsilon_K | \mathbb N_1  \Upsilon_K \rsp_\Fock \,  \Big| \, \le \, C_\varepsilon  \sqrt K \alpha^{-6+\varepsilon} .   
\end{align}
The error in \eqref{eq: add subtract gaussian for E111} is bounded with the help of Lemma \ref{lem: Gaussian lemma},
\begin{align}\label{eq: bound for E111 err}
| \mathcal E_{111}^{\rm err} |\,  \le\,  \frac{C}{\alpha^2} \int \D y\, H(y) | n_{0,1}(y) - e^{-\lambda \alpha^2 y^2 }| \,  \le \,  C \alpha^{-6} .
\end{align}
In  $\mathcal E_{112}$ we use the Cauchy--Schwarz inequality, Corollary \ref{Cor: Upsilon estimates} and Lemma \ref{lem: bound for w1 and w0}, to obtain
\begin{align}
& \I  \lsp \Upsilon_K | \mathbb N_1  \mathbb U_K^\dagger  a^\dagger( \alpha \widetilde w_{P,y}^1 ) a^\dagger( \alpha \widetilde w_{P,y}^1 )  \Omega \rsp_\Fock \I  \notag\\[1mm]
& \hspace{0.5cm} \, \le \, \sno \mathbb N_1 \Upsilon_K \sno_\Fock \, \sno a^\dagger( \alpha \widetilde w_{P,y}^1 ) a^\dagger( \alpha \widetilde w_{P,y}^1 )  \Omega \sno_\Fock  \, \le \, 2 \alpha^2 \sno  \widetilde w_{P,y}^1\sno^2_\2 \, \le  \, C \alpha^2 (y^4 + \alpha^{-4}).
\end{align}
With $\sno |\cdot |^n H \sno_\1 \le C_n$ we can now apply Corollary \ref{cor: Gaussian for errors} to obtain
\begin{align}\label{eq: bound for E112}
| \mathcal E_{112} | \, \le \,  C \int \D y\, H(y) ( y^4 + \alpha^{ - 4}) n_{0,1}(y) \, \le\,  C \alpha^{-7}.
\end{align}

In order to bound $\mathcal E_{12}$ in \eqref{eq: E11 + E12}, we decompose $\Upsilon_K = \Upsilon_K^< + \Upsilon_K^>$ for some $\delta>0$, see \eqref{eq: decomposition of Upsilon}, and then follows similar steps as described below \eqref{eq: inserting identity}. This way we can estimate
\begin{equation}\label{A-Schwarzing}
| \mathcal E_{12} |   \le   \frac{1}{\alpha^2} \int \D y\, H(y) \sno e^{\kappa \mathbb N } \mathbb U_K (e^{ - A_{P,y}} - 1)  \mathbb N_1 \Upsilon_K^< \sno_\Fock \, n_{\delta,\eta}(y) 
  + \, \frac{2}{\alpha^2} \sno \mathbb N_1 \Upsilon_K^>\sno_\Fock \int \D y\, H(y) .
\end{equation}
While the second term is bounded via \eqref{eq: exp bound for tail} by $C_\delta\, \alpha^{-12}$, in the first term we apply Lemma \ref{lem: exp N bound} and use the functional calculus for self-adjoint operators,
\begin{align}
 \sno e^{\kappa \mathbb N } \mathbb U_K (e^{ - A _{P,y} } - 1)  \mathbb N_1 \Upsilon_K^< \sno_\Fock & \, \le \, \sqrt 2 \sno  ( e^{ - A_{P,y}} - 1 ) \mathbb N_1 \Upsilon_K^< \sno_\Fock \notag\\[1mm]
 &\,  \le \, \sqrt 2 \sno  (P_fy + g_{P}(y) )  \mathbb N_1 \Upsilon_K^<  \sno_\Fock.
\end{align}
Since $P_f$ changes the number of phonons in $\mathcal F_1$ at most by one, we can proceed by
\begin{align}
\hspace{-1mm}\sno  (P_fy + g_{P}(y) )  \mathbb N_1 \Upsilon_K^<  \sno_\Fock \le  (\alpha^{\delta}+1) \sno (P_f y  + g_{P}(y)) \Upsilon_K \sno_\Fock \le   C  \alpha^\delta (\sqrt K |y| + \alpha |y|^3),
\end{align}
where we used $1\le \alpha^\delta$, Lemma \ref{lem: bounds for P_f} and \eqref{eq: bound for g_P} in the second step. We conclude via Corollary \ref{cor: Gaussian for errors} that
\begin{align}\label{eq: bound for E12}
| \mathcal E_{12} | \, \le \, \frac{C}{\alpha^2} \int \D y\, H(y) (\sqrt K |y| + \alpha |y|^3 ) n_{\delta,\eta}(y) + C_\delta\, \alpha^{-12}
\, \le \, C_\delta \, \sqrt K \alpha^{- 6 + 4\delta} .
\end{align}

\noindent \underline{Term $\mathcal E_2$}. 
Here we start with
\begin{align}
\mathcal E_2 & \, = \, \alpha^{-1} \int \D y\,  \lsp \Upsilon_K | L_{1,y} W(\alpha w_{P,y})   \Upsilon_K \rsp_\Fock \notag \\
& \quad + \alpha^{-1} \int \D y\, \lsp \Upsilon_K | L_{1,y}  ( e^{A_{P,y}}  -1 ) W(\alpha w_{P,y})   \Upsilon_K \rsp_\Fock \, = \, \mathcal E_{21} + \mathcal E_{22} ,
\end{align}
where
\begin{align}
 L_{1,y}  \, =\,  \lsp \psi| \phi(h_\cdot + \varphi_P) T_y \psi \rsp_\2 = \phi(l_y)+\pi(j_y)
\end{align}
with 
\begin{align}
  l_y \,  = \,  H(y)\varphi+\lsp \psi|h_{\cdot}T_y \psi\rsp_\2, \quad j_ y \,  = \,  H(y)\xi_P,
\end{align} 
and $\xi_P$ defined in \eqref{eq: def of varphi_P}. We record the following properties of $l_y$ and its derivative. The proof of the lemma is postponed until the end of the present section.   
\begin{lemma}\label{lemma: props_lx}
For $k=0,1$ and for all $n\in \mathbb{N}_0$,
\begin{align}\label{eq: props_lx}
  \sup_y \|\nabla^k l_y \|_\2 \, <\, \infty, \quad \int |y |^n\|\nabla^k l_y \|_\2 \, \D y \, <\, \infty.
\end{align}  
\end{lemma}
Note that, by Lemma \ref{lemma: props_peks}, $j_y$ clearly has these properties as well. We proceed by writing $\mathcal E_{21} =  \mathcal{E}_{21}^0 +  \mathcal{E}_{21}^P$ with
\begin{subequations}
\begin{align}
  \mathcal{E}_{21}^0 & \, =\, \alpha^{-1}\int \D y\, \lsp \Upsilon_K|\phi(l_y)W(\alpha w_{P,y}) \Upsilon_K\rsp_\Fock \\
  \mathcal{E}^P_{21} &  \, =\, \alpha^{-1}\int \D y \, \lsp \Upsilon_K |\pi(j_y)W(\alpha w_{P,y}) \Upsilon_K\rsp_\Fock ,
  \end{align} 
\end{subequations}
and estimate the two parts separately. Using the canonical commutation relations and \eqref{eq: transformation of phi}, we evaluate
\begin{align}
   \mathcal{E}_{21}^0 & \, =\,  \int \lsp \underline{l_y}|\widetilde{w}_{P,y}\rsp_\2 n_{0,1}(y) \D y \notag\\
   & \, =\, \int \Big( \lsp l_y^0|w_{P,y}^0\rsp_\2 + \lsp l_y^1 |\mathrm{Re}(w_{P,y}^1)\rsp_\2 + i\lsp l_y^1|\Theta^{-2}_K\mathrm{Im}(w_{P,y}^1)\rsp_\2 \Big) n_{0,1}(y) \D y 
\end{align}
where we used that $l_y$ is real-valued.  Note that $l_{-y}(-z)=l_y(z)$. As discussed in Remark \ref{rem: symmetries_of_w}, $n_{0,1}(y)$ is even, and using the arguments therein one can conclude that $ \Theta^{-2}_K\mathrm{Im}(w_{P,y}^1)$ and $\mathrm{Im}(w_{P,y}^0)$ are odd functions on $\mathbb{R}^6$ since $(y,z)\mapsto \mathrm{Im}(w_{P,y})(z)$ is odd on this space, and hence
\begin{align}
\int  \lsp l_y^0 |\mathrm{Im}(w_{P,y}^0) \rsp_\2 n_{0,1}(y)  \D y \, =\, \int \lsp l_y^1 |\Theta^{-2}_K\mathrm{Im}(w_{P,y}^1)\rsp_\2 n_{0,1}(y)\D y \, =\, 0.
\end{align}
Thus, with  $\mathrm{Re}(w_{P,y})=w_{0,y}$, and with
\begin{align}
  v(y)\, :=\, \lsp l_y|w_{0,y}\rsp_\2
\end{align} 
we finally have
\begin{align}\label{eq: E_21^0}
    \mathcal{E}_{21}^0\, =\,  \int \lsp l_y^0 + l_y^1| \re(w_{P,y}^0)+ \re (w_{P,y}^1)  \rsp_\2 n_{0,1}(y)\D y\,  =\,  \int v(y)n_{0,1}(y)\D y. 
\end{align} 
Note that $v\in L^1\cap L^{\infty}$ since $y\mapsto \sno l_y \sno_\2$ is, while $\sno w_{0,y}\sno_\2$ is uniformly bounded in $y$. Because of
 $\varphi(z)  = -\langle \psi |h_{\cdot}(z)\psi\rangle_\2$ 
 and
 $\nabla_z h(x-z)\, =\, -\nabla_x h(x-z)$ 
we have by integration by parts  
\begin{align}
  \nabla \varphi\, =\, -2\lsp\nabla \psi|h_{\cdot}\psi\rsp_\2 .
\end{align}
Thus 
\begin{align}
  l_y\, =\, -\frac{1}{2}y\nabla\varphi +\varphi(H(y)-1)+\lsp \psi|h_{\cdot}(T_y \psi -\psi-y\nabla\psi)\rsp_\2 .
\end{align}
Since $\psi$ is a smooth function with uniformly bounded derivatives, there exists a $C>0$ such that  for all $y$
\begin{align}\label{eq: expansion of Txpsi}
  \sno T_y \psi-\psi-y\nabla \psi \sno_\su \, \leq\,  Cy^2.
\end{align} 
Moreover, for $k=0,1$ and every $z\in \mathbb R^3$, 
\begin{align}\label{eq: h(z)psi maps}
x\mapsto (h_{\cdot}(z)\nabla^k\psi)(x)\in L^1(\mathbb{R}^3,\D x)\quad \text{and}\quad z\mapsto \sno h_{\cdot}(z)\nabla^k \psi \sno_\1 \in L^2(\mathbb R^3, \D z).
\end{align}
The first statement follows easily from Lemma \ref{lemma: props_peks}; to show the second one, use
\begin{align}\label{eq: integral identity}
\int \D z\, \frac{1}{|u-z|^2|v-z|^2} \, =\,  \frac{1}{\pi^3 |u-v|}
\end{align}
and apply the Hardy--Littlewood--Sobolev inequality. This, together with \eqref{eq: quadratic_bd_on_H},  shows that there exists a function $f$ in $L^2(\mathbb R^3, \D z)$ such that
\begin{align}\label{eq: bd on l_x}
|l_y(z)+\frac{1}{2}y\nabla\varphi(z)|\, \leq\,  f(z)y^2.
\end{align} 
Now let 
\begin{align}
  b_y(z)\, :=\, w_{0,y}(z)-y\nabla\varphi(z)\, =\, \int_0^1 \D s \int_0^s\D t\, (y\nabla)^2 \varphi(z-ty)
\end{align} 
and note that $\sno b_y \sno^2_\2 \leq \frac{1}{4} y^4\sno \Delta \varphi\sno^2_\2$ which is finite since $\Delta \varphi\in L^2$. This equation, together with \eqref{eq: bd on l_x}, implies
\begin{align}\label{eq: bd on v_x}
  \bigg| v(y)+\frac{1}{2} \sno y\nabla \varphi \sno_\2^2 \bigg| \, \leq\,  C(|y|^3 + |y|^4).
\end{align}
From this, and from $v\in L^1\cap L^{\infty}$ it is also easy to deduce that $|\cdot|^{-2}v \in L^1\cap L^{\infty}$.
We can thus write
\begin{align}
  \int \D y \, v(y) n_{0,1}(y)\, =\, \int \D y\, v(y) e^{-\alpha^2\lambda y^2}+\int \D y \, |y|^{-2}v(y) y^2 \big( n_{0,1}(y)-e^{-\alpha^2\lambda y^2} \big)
\end{align}
and use Lemma \ref{lem: Gaussian lemma} for $g=|\cdot|^{-2}|v|$ to bound
\begin{align}\label{eq: bound on the error E_21^0}
  \left|\int \D y\, |y|^{-2}v(y) y^2 (n_{0,1}(y)-e^{-\alpha^2\lambda y^2})\right|\, \leq\,  C\alpha^{-6}.
\end{align}
Using \eqref{eq: bd on v_x}, the definition of $\lambda=\frac{1}{6}\sno \nabla \varphi \sno_\2^2$ as well as $\int y^2 e^{-y^2}\D y= \frac{3}{2}\pi^{3/2}$, we further have that
\begin{align}
  \bigg| \int  \D y \, v(y) e^{-\alpha^2\lambda y^2} + \frac{3}{2\alpha^2}\left(\frac{\pi}{\lambda \alpha^2}\right)^{3/2} \bigg| \, \le\,  C\alpha^{-6}
\end{align}
which finally gives the estimate 
\begin{align}\label{eq: bound for E_21^0}
 \bigg| \mathcal{E}_{21}^0 + \bigg( \frac{3}{2\alpha^2} \bigg)  \mathcal N \, \bigg| \, \le\,  C_\varepsilon \sqrt K \alpha^{- 6 + \varepsilon } 
\end{align}
using Proposition \ref{prop: norm bound}.

In a similar fashion as for $\mathcal{E}_{21}^0$, we obtain  
\begin{align}
  \mathcal{E}^P_{21} &  \, =\, \frac{1}{\alpha^{2}M^{\rm LP}}\int \lsp  i P\nabla \varphi |w^0_{P,y}\rsp_\2  H(y)n_{0,1}(y) \D y.
\end{align}
Explicit computation, using $\Pi_0 = \frac{3}{\sno \nabla \varphi \sno^2_\2}\sum_{i=1}^3 |\partial_i \varphi \rangle \langle \partial_i \varphi |$ and $\langle \varphi | \nabla \varphi\rangle_\2 = 0$, gives
\begin{align}\label{eq: w0 explicit}
  \frac{1}{3}w^0_{P,y}(z) \, =\, -\frac{(\varphi\ast \nabla \varphi)(y)}{\sno \nabla \varphi\sno_\2^2}\nabla \varphi(z)+\frac{iP}{\alpha^2 M^{\rm LP}}\left(\sno \nabla \varphi\sno^2_\2 -(\nabla \varphi\ast \nabla \varphi)(y)\right)\frac{\nabla \varphi(z)}{\sno \nabla \varphi\sno^2_\2}.
\end{align} 
Note that the real part of the above is odd as a function of $y$ and hence
 \begin{align}
 \int \lsp \nabla \varphi|\mathrm{Re}(w_{P,y}^0)\rsp_\2  n_{0,1}(y)H(y)\D y\, =\, 0,
 \end{align}
and, taking rotational invariance of $\varphi$ into account, we arrive at
 \begin{align}
  \mathcal{E}_{21}^P\, =\, \frac{P^2}{\alpha^4 (M^{\mathrm{LP}})^2}\int  \left(\|\nabla \varphi\|_2^2-\left(\nabla \varphi\ast\nabla \varphi\right)(y)\right)n_{0,1}(y)H(y)\D y.
\end{align} 
Further note that $|\sno \nabla \varphi\sno_\2^2 - \left(\nabla \varphi\ast\nabla \varphi\right)(y)|\leq Cy^2$ and thus, by Lemma \ref{lemma: props_peks} and Corollary \ref{cor: Gaussian for errors}, one obtains
\begin{align}\label{eq: bound for E_21^P}
|\mathcal{E}_{21}^P|\, \leq\,  C\frac{P^2}{\alpha^9} \le \frac{C}{\alpha^7}.
\end{align}
This completes the analysis of $\mathcal E_{21}$.

In order to estimate the term $\mathcal{E}_{22}$, we proceed as before by splitting $\Upsilon_K=\Upsilon_K^{<}+\Upsilon_{K}^{>}$. Using \eqref{eq: standard estimates for a and a*} we can estimate
\begin{align}\label{eq: bound for E22 >}
&\bigg| \alpha^{-1}\int \D y\, \lsp \Upsilon^{>}_K |(\phi(l_y)+\pi(j_y))(e^{A_{P,y}}-1)W(\alpha {w}_{P,y}) \Upsilon_K\rsp_\Fock \bigg|  \notag \\ 
& \quad \quad \, \leq \, C  \alpha^{-1}\int \D y \, (\sno l_y \sno_\2 + \sno j_y\sno_\2)\|(\mathbb{N}+1)^{1/2}\Upsilon_K^{>}\|_{\Fock} \,   \leq \,  C_\delta \, \alpha^{-11}
\end{align} 
where we used Corollary \ref{Cor: Upsilon estimates} and Lemmas \ref{lemma: props_peks} and \ref{lemma: props_lx}. The term involving $\Upsilon_K^{<}$, we split again into two contributions, 
\begin{subequations}
\begin{align}
  \mathcal{E}_{22}^0  & \, =\,   \alpha^{-1}\int \D y\, \lsp  \Upsilon^{<}_K | \phi(l_y)(e^{A_{P,y}}-1)W(\alpha {w}_{P,y}) \Upsilon_K\rsp_\Fock \\
\mathcal{E}^P_{22} & \,  =\,  \alpha^{-1}\int \D y\, \lsp \Upsilon^{<}_K |\pi(j_y) (e^{A_{P,y}}-1)W(\alpha {w}_{P,y})\Upsilon_K\rsp_\Fock .
\end{align}
\end{subequations}
To bound the first one we proceed as in \eqref{A-Schwarzing}, i.e. use Lemma \ref{lem: exp N bound} and the fact that $\phi(l_y)$ changes the number of phonons at most by one. This leads to
\begin{align}
|\mathcal E_{22}^0|  & \, \leq\,  \alpha^{-1}\int \D y\, \|e^{\kappa \mathbb{N}}\mathbb{U}_{K}(e^{-A_{P,y}}-1)\phi(l_y)\Upsilon_K^{<}\|_\Fock\, n_{\delta,\eta}(y) \notag \\
  &\, \leq \, \sqrt 2 \alpha^{-1}\int \D y\, \|(e^{-A_{P,y}}-1)\phi(l_y)\Upsilon_K^{<}\|_\Fock \, n_{\delta,\eta}(y)\label{eq: bound E22 contd}.
\end{align} 
Furthermore, we have 
\begin{align}
  \sno (e^{-A_{P,y}}-1)\phi(l_y)\Upsilon^{<}_K \sno_\Fock & \,  \leq \,  \sno A_{P,y}\phi(l_y)\Upsilon_K^{<}\sno_\Fock 
   \,  \leq \, \sno \phi(l_y)A_{P,y}\Upsilon_K^{<}\sno_\Fock + \sno [A_{P,y},\phi(l_y)]\Upsilon_K^{<}\sno_\Fock \notag \\[0.5mm]
   & \, \leq \, C \alpha^{\delta/2}\left(\sno l_y \sno_\2 \sno A_{P,y}\Upsilon_K \sno_\Fock + \sno y\nabla l_y \sno_\2\right)\label{eq: E_2^P two terms}
\end{align} 
where we used $[iP_f y,\phi(f)]=\pi(y\nabla f)$ and $\Upsilon_K^< = \mathbbm{1}(\mathbb N \le \alpha^\delta) \Upsilon_K$. Note that in order to estimate the remaining expression, it is not sufficient to directly apply Corollary \ref{cor: Gaussian for errors}. To obtain a better bound, we first replace $n_{\delta,\eta}(y)$ by $e^{-\eta\lambda  \alpha^{2(1-\delta)} y^2}$ and then, for the part containing the Gaussian, we use that $\sno l_y\sno_\2$ and $\sno \nabla l_y \sno_\2$ provide additional factors of $|y|$, as is shown below. More precisely, with Lemma \ref{lemma: props_lx} and the aid of Lemmas \ref{lem: Gaussian lemma} and \ref{lem: bounds for P_f}, we bound
\begin{align}
& \alpha^{\frac{\delta}{2}-1}\int \D y\,  \sno l_y\sno_\2 \sno A_{P,y} \Upsilon_K \sno_\Fock \,  n_{\delta,\eta}(y) \, \leq \, C\alpha^{\frac{\delta}{2}-1} \int \D y \, \sno l_y\sno_\2 (\sqrt K |y| + \alpha |y|^3 ) n_{\delta,\eta}(y) \notag \\
& \quad\quad\quad \, \leq\,  C\alpha^{\frac{\delta}{2}-1} \int \D y\,  \sno l_y \sno_\2 (\sqrt K |y| + \alpha |y|^3 ) e^{-\eta\lambda  \alpha^{2(1-\delta)} y^2}+ C\sqrt{K} \alpha^{-6 + \frac{9\delta}{2}}.
\end{align}
Next we use that by Equation \eqref{eq: bd on l_x} there exists an $L^2$ function $f$ such that \begin{align}\label{eq: l_x}
  |l_y(z)|\, \leq \, \frac{1}{2}|y\nabla\varphi(z)|+f(z) y^2.
\end{align}
Hence, by integration
 \begin{align}\label{eq: E_22^P second term}
  \alpha^{\delta/2-1}\int \D y \, \sno l_y \sno_\2 \left(\sqrt K |y| + \alpha |y|^3 \right)e^{-\lambda \eta \alpha^{2(1-\delta)}y^{2}}\, \leq \, C \sqrt{K}\alpha^{-6 + 11/2\delta}.
\end{align} 
With regard to the second term in \eqref{eq: E_2^P two terms},
\begin{align}\label{eq: E_22^P two terms}
  \alpha^{\delta/2-1}\int \D y\, |y| \sno \nabla l_y \sno_\2 n_{\delta,\eta}(y)
\end{align}
we proceed in a similar way, using that 
\begin{align}\label{eq: bd on nablalx}
\sno \nabla l_y \sno_\2\, \leq \, C(|y|+ y ^2).
\end{align} 
In fact, since $\nabla \varphi(z) =-\langle \psi|h_{\cdot}(z) \nabla \psi\rangle_\2  -\langle \nabla \psi |h_{\cdot}(z) \psi\rangle_\2$, we have the identity
\begin{align}
 \nabla l_y(z) & \, =\,  H(y)\nabla \varphi(z) +\lsp \nabla\psi|h_\cdot(z) T_y\psi\rsp_\2 + \lsp\psi|h_\cdot(z) \nabla T_y\psi\rsp_\2  \\
  &\, =\,  (H(y)-1)\nabla \varphi (z) + \lsp \nabla \psi|h_{\cdot}(z)( T_y-1 ) \psi\rsp_\2 + \lsp \psi |h_{\cdot}(z) (T_y -1 ) \nabla \psi \rsp_\2 .\notag
\end{align}
Again using that $\psi$ has  bounded derivatives, we have  
\begin{align}
\sno (T_y-1)\psi\sno_\su + 
\sno (T_y-1)\nabla \psi\sno_\su \le C |y|,
\end{align}
and the desired inequality now follows from $| H(y)-1 | \le Cy^2$ and \eqref{eq: h(z)psi maps}. Given \eqref{eq: props_lx}, we can use Lemma \ref{lem: Gaussian lemma} to replace $n_{\delta,\eta}(y)$ in \eqref{eq: E_22^P two terms} with $e^{-\lambda\eta\alpha^{2(1-\delta}) y^2}$ at the energy penalty $C \alpha^{-6 + 9\delta /2}$, and then use \eqref{eq: bd on nablalx} to bound the remaining integral involving the Gaussian factor, which yields an error of the same order. Altogether, this gives the estimate 
\begin{align}\label{eq: bound for E_22^0}
  |\mathcal{E}_{22}^0|\, \leq\,  C \sqrt{K}\alpha^{-6+\frac{11}{2}\delta}.
\end{align}

For the term $\mathcal{E}^P_{22}$ we proceed in exactly the same way as in \eqref{eq: bound E22 contd}: 
\begin{align}\label{eq: bound for E_22^P}
 | \mathcal{E}^P_{22} | &\, \leq\,  \sqrt{2} \alpha^{-1} \int \D y\,  \|\left(e^{- A_{P,y}}-1\right)\pi (j_y)\Upsilon_K^{<}\|_\Fock \, n_{\delta,\eta}(y) \notag \\
 & \, \leq\,  C\alpha^{\delta/2-1}\int \D y\, \|j_y\|_\2 \|A_{P,y}\Upsilon_K\|_\Fock\, n_{\delta,\eta}(y) +C\alpha^{\delta/2-1}\int \D y\, \|y\nabla j_y\|_\2 n_{\delta,\eta}(y) \notag \\
 &\, \leq\,  C\alpha^{\delta/2-1} \frac{|P|}{\alpha^2}\int \D y\,  H(y)\, (\sqrt{K}|y|+\alpha|y|^3) n_{\delta,\eta}(y) \notag \\
 & \quad + C\alpha^{\delta/2-1} \frac{|P|}{\alpha^2} \int \D y\, |y|H(y) \,  n_{\delta,\eta}(y) \notag\\
 &\, \leq\,   C\alpha^{-6 + \frac{9}{2}\delta}  \sqrt{K} 
\end{align} 
where the last estimate follows from Corollary \ref{cor: Gaussian for errors} and the assumption $|P|\le c \alpha$.

Combining the relevant estimates, that is \eqref{eq: bound for E111 lo}, \eqref{eq: bound for E111 err}, \eqref{eq: bound for E112} and \eqref{eq: bound for E12} for $\mathcal E_1$ as well as \eqref{eq: bound for E_21^0}, \eqref{eq: bound for E_21^P}, \eqref{eq: bound for E22 >}, \eqref{eq: bound for E_22^0} and \eqref{eq: bound for E_22^P} for $\mathcal E_2$, we arrive at the statement of Proposition \ref{prop: bound for E}, thus providing an appropriate bound for $\mathcal{E}$.
\end{proof}

\begin{proof}[Proof of Lemma \ref{lemma: props_lx}] Since $H$ has the desired properties, we need to show them for 
\begin{align}
  l^{(1)}_y \, =\,  \lsp \psi|h_{\cdot}T_y \psi\rsp_\2 .
\end{align}
To this end we introduce
\begin{align}
  \mathcal{S}\, =\, \lbrace f\in L^p(\mathbb{R}^3,(1+|y|^n)\D y) \quad \forall 1\leq p\leq \infty, \quad \forall n\geq 0\rbrace
\end{align} 
and start with the following observation: Suppose $f_1,f_2,f_3$ and $f_4$ are functions in $\mathcal{S}$. Then 
\begin{align}\label{eq: HLS lemma}
  S(y)\, :=\, \iint \D u \D v \frac{f_1(u)f_2(v)f_3(u+y)f_4(v+y)}{|u-v|} \in \mathcal{S}.
\end{align}
In fact, $|S(y)| \leq C \|f_4\|_\su \|f_3\|_\su \|f_1\|_\p \|f_2\|_\q$ for  all $1<p<3/2, q=3p/(5p-3)$ by the Hardy--Littlewood--Sobolev inequality. Since $\int \D y |y|^n f_3(u+y)\leq 2^{n-1}\left(|u|^n \| f_3 \|_\1 +\||\cdot|^nf_3\|_\1\right)$, we have also
\begin{align}
  \int \D y |y|^n S(y)\,  \leq\,  C \|f_4\|_\su \left(\||\cdot|^nf_1\|_\p\|f_2\|_\q \|f_3\|_\1 +\|f_1\|_\p \|f_2\|_\q \||\cdot|^n f_3\|_\1\right)
\end{align} from which \eqref{eq: HLS lemma} follows. Moreover,
\begin{align}\label{eq: square root lemma}
f\in \mathcal{S }\implies \sqrt{|f|}\in \mathcal{S}. 
\end{align} 
Indeed, we have for all $n\geq 0$, 
\begin{align}
  \int |y|^n \sqrt{|f|}\D y \, \leq\,  \sqrt{\|f\|_\su }\int_{|y|\leq 1}|y|^n\D y+\frac{1}{2}\int |y|^{n+m}|f| \D y+\frac{1}{2}\int_{|y|>1}|y|^{n-m}\D y \, <\,  \infty
\end{align}
since $m$ can be chosen arbitrarily large by assumption. Thus, it suffices to prove the desired statement for the functions $\|\nabla^k l_y^{(1)}\|_\2^2$.
For $k=0$, we use \eqref{eq: integral identity} to compute
\begin{align}
  \|l^{(1)}_y\|_\2^2\, =\, \frac{1}{4\pi}\iint \D u \D v \frac{\psi(u)\psi(v)\psi(y+u)\psi(v+y)}{|u-v|}.
\end{align}
The statement now follows easily from \eqref{eq: HLS lemma} and Lemma \ref{lemma: props_peks}. Arguing again via \eqref{eq: square root lemma}, for $k=1$ it suffices to show the statement for 
\begin{align}
  \|\nabla l^{(1)}_y\|_\2^2 &\, =\, \|\langle  \nabla \psi|h_{\cdot} T_y\psi\rangle_\2 +\lsp\psi|h_{\cdot}\nabla T_y\psi\rangle_\2 \|_\2^2 \notag \\
  & \, \leq\,  2 \|\langle \nabla \psi|h_{\cdot} T_y\psi\rangle_\2 \|_\2^2 + 2 \|\langle \psi|h_{\cdot}\nabla T_y \psi\rangle_\2 \|_\2^2 
\end{align} 
(the first equality follows from $\nabla_z h_x(z) = - \nabla_x h_x(z)$ and integration by parts). Using \eqref{eq: integral identity}, we find 
\begin{subequations}
\begin{align}
\|\langle \nabla \psi|h_{\cdot} T_y\psi\rangle_\2 \|_\2^2 & \, \leq\,  C \iint \D u \D v \frac{|\nabla\psi(u)||\nabla\psi(v)|\psi(v+y)\psi(u+y)}{|u-v|},\\
\|\langle \psi|h_{\cdot}\nabla T_y\psi\rangle_\2 \|_\2^2 & \, \leq \, C  \iint \D u \D v \frac{|\nabla \psi(u+y)||\nabla\psi(v+y)|\psi(v)\psi(u)}{|u-v|}.
\end{align}
\end{subequations}
We arrive at the desired conclusion by Lemma \ref{lemma: props_peks} and \eqref{eq: HLS lemma}.
\end{proof}

\subsection{Energy contribution $\mathcal G$}

This energy contribution, defined in \eqref{eq: G}, is evaluated by the following proposition.

\begin{proposition} \label{prop: bound for G} Let $\mathbb H_K$ as in \eqref{eq: Bogoliubov Hamiltonian maintext}, $\mathbb N_1= \D \Gamma(\Pi_1)$ and choose $c>0$. For every $\varepsilon> 0$ there exists a constant $C_\varepsilon >0$ \textnormal{(}we omit the dependence on $c$\textnormal{)} such that
\begin{align}
\bigg| \, \mathcal G  - \mathcal N \frac{2}{\alpha^{2}} \lsp \Upsilon_K| ( \mathbb H_K - \mathbb N_1 )  \Upsilon_K \rsp_\Fock \,  \bigg| \, \le \, C_\varepsilon\, \alpha^\varepsilon \big(  \sqrt K \alpha^{-6} + K^{-1/2} \alpha^{-5} \big) 
\end{align}
for all $| P |/\alpha \le c$ and all $K,\alpha$ large enough.
\end{proposition}
\begin{proof} Using $h^{\rm Pek} G_K^0 = 0$ and $\mathbb NG_K^0 = \mathbb N_1 G_K^0$ we can decompose $\mathcal G$ into two terms
\begin{align}
\mathcal G & \, =\,  - \frac{2}{\alpha} \int \D y\, \re \lsp G_{K}^{ 0 } | (\alpha^{-2} \mathbb N_1 + \alpha^{-1} \phi (h_\cdot + \varphi_P) ) T_y e^{A_{P,y}} W(\alpha w_{P,y})  G_{K}^{ 1} \rsp_\h \notag\\
&\,  =\,  \mathcal G_1 + \mathcal G_2,
\end{align}
where the first term will contribute to the error while the second one provides an energy contribution of order $\alpha^{-2}$. We proceed for each one separately.\medskip

\noindent \underline{Term $\mathcal G_1$}. With the aid of \eqref{eq: decomposition of Upsilon}, \eqref{eq: identity for WG1} and $(T_y e^{A_{P,y}} )^\dagger = T_{-y} e^{-A_{P,y}}$, one finds
\begin{align}
\mathcal G_1 & \, =\,  -  \frac{2}{ \alpha^{3} }  \int \D y\, \re \lsp  R_{3,y} \psi \otimes \big( \Upsilon_K^< + \Upsilon_K^>\big) |  W(\alpha w_{P,y})  G_K^0 \rsp_\h  \, =\,  \mathcal G_{1}^< + \mathcal G_{1}^>
\end{align}
where we introduced the operator $R_{3,y} = R_{3,y}^1 + R_{3,y}^2$ with 
\begin{subequations}
\begin{align}
R_{3,y}^1 & \, =\,  \PP \phi(h_{K,\cdot}^1)  R u_\alpha T_{-y} \PP e^{-A_{P,y}} \mathbb N_1 \\[1mm]
 R_{3,y}^2 & \, =\,  2\alpha  \PP  \lsp h_{K,\cdot} | \re (w_{P,y}^1) \rsp_{\2} R u_\alpha T_{-y} \PP  e^{-A_{P,y}} \mathbb N_1 .
\end{align}
\end{subequations}
Proceeding similarly as for $R_{1,y}^1$ and $R_{2,y}^2$ in \eqref{eq: R_1y^1}--\eqref{eq: R_1y^2}, one further verifies
\begin{align}
\sno  R_{3,y} \Psi \sno_\h \, \le \, C \sno u_\alpha T_{-y}  \PP\sno_{\op} \big( 1 + \alpha  y^2\big) \sno  (\mathbb N+1)^{3/2} \Psi \sno_\h.
\end{align}
Recalling the definition $f_\alpha(y) =  \sno u_\alpha T_{-y}  \PP  \sno_{\op}$ and \eqref{eq: bound for f alpha}, we can use Corollary \ref{Cor: Upsilon estimates} to find
\begin{align}\label{eq: bound for G_1^>}
|\mathcal G_{1}^> | \, \le\,  \frac{C}{\alpha^3} \sno (\mathbb N+1)^{3/2} \Upsilon_K^>\sno_{\Fock} \int \D y \, f_\alpha(y) (1+ \alpha y^2) \, \le\, C_\delta\, \alpha^{-7 }  .
\end{align}
In the first term we proceed with \eqref{eq: bound for epsilon N} and Lemma \ref{lem: exp N bound} to obtain
\begin{align}
|\mathcal G_{1}^< | & \, \le \, \frac{2}{ \alpha^{3} } \int \D y\, \sno e^{\kappa \mathbb N} \mathbb U_K  ( R_{3,y}  \psi \otimes \Upsilon_K^< ) \sno_\h \, \sno e^{-\kappa \mathbb N} W(\alpha \widetilde w_{P,y} ) \Omega \sno_{\Fock} \notag\\
& \, \le\,  \frac{2\sqrt 2}{ \alpha^{3} } \int \D y\, \sno  R_{3,y} \psi\otimes  \Upsilon_K \sno_{\h} \, n_{\delta,\eta}(y) \, \le\,  \frac{C}{ \alpha^{3} } \int \D y\, f_\alpha(y) (1+\alpha y^2) \, n_{\delta,\eta}(y) ,
\end{align}
which brings us again into a position to apply Corollary \ref{cor: Gaussian for errors}. Hence 
\begin{align}\label{eq: bound for G_1^y}
|\mathcal G_{1}^< | \le C \alpha^{- 6 +3 \delta }.
\end{align}
\noindent \underline{Term $\mathcal G_2$}. Here we have
\begin{align}
\mathcal G_{2} &  =   -  \frac{2}{\alpha^2} \int \D y\,  \re \lsp G_K^0 | \phi (h_{\cdot} + \varphi_P) T_y W(\alpha w_{P,y})  G_K^1 \rsp_\h \notag \\
& \quad  - \, \frac{2}{\alpha^2} \int \D y \, \re \lsp G_K^0 | \phi (h_\cdot + \varphi_P ) T_y ( e^{ A_{P,y}} -1 )  W(\alpha w_{P,y})  G_K^1 \rsp_\h = \mathcal G_{21} + \mathcal G_{22} \label{eq: def of G_22}. 
\end{align}
To separate the leading order contribution in $\mathcal G_{21}$ we insert $1 = \mathbb U_K^\dagger \mathbb U_K$ next to $G_K^0$ and bring $\mathbb U_K^\dagger$ to the right side of the inner product. With $\mathbb U_K\Upsilon_K = \Omega$, \eqref{eq: transformation of phi} and \eqref{eq: identity for tilde w transformation} this gives
\begin{align}
\mathcal G_{21} & \, =\,  - \frac{2}{\alpha^2}  \int \D y \, \re \lsp \psi \otimes \Omega | a(\underline { h_\cdot + \varphi_P } )T_y  W(\alpha \widetilde w_{P,y})  u_\alpha R  a^\dagger(\underline{h_{K,\cdot}^1} )  \psi \otimes \Omega \rsp_\h,
\end{align}
where $\underline{\, \cdot\, }$ is defined in \eqref{eq: underline notation}. Next we write $ W(\alpha \widetilde w_{P,y}) = n_{0,1}(y) e^{a^\dagger(\alpha \widetilde w_{P,y})} e^{-a(\alpha \widetilde w_{P,y})}$ and move the first exponential to the left side and the second exponential to the right side until they act both on the Fock space vacuum. Using $e^{-a(f)} a^\dagger(g) e^{a(f)} = a^\dagger(g) - \langle f |g \rangle $ we find this way
\begin{subequations}
\begin{align}
\!\!\!\!\!\mathcal G_{21} &  =   - \frac{2}{\alpha^2}  \int \D y \, n_{0,1}(y) \, \re \lsp \psi \otimes \Omega | a(\underline { h_\cdot + \varphi_P} ) T_y u_\alpha R  a^\dagger ( \underline{h_{K,\cdot}^1} ) \psi \otimes \Omega \rsp_\h \label{eq: aux formula G21}\\
& \quad + 2  \int \D y \, n_{0,1}(y)\, \re \lsp \psi \otimes \Omega | \langle \underline { h_\cdot + \varphi_P } | \widetilde w_{P,y} \rangle_\2  T_y  u_\alpha  R  \langle \widetilde w_{P,y} | \underline{h_{K,\cdot}^1} \rangle_\2 \psi \otimes \Omega \rsp_\h .
\end{align}
\end{subequations}
In the first line we write $h_\cdot + \varphi_P = h_{\cdot}^0 + h_{\cdot}^1 + \varphi + i \xi_P$, with $h_\cdot^i = ( \Pi_ih)_{\cdot}$, and use that 
\begin{align}
\lsp \psi \otimes \Omega | a(\underline { h_\cdot^0 + i \xi_P } ) T_y u_\alpha R  a^\dagger ( \underline{h_{K,\cdot}^1} )  \psi \otimes \Omega \rsp_\h   = 0
\end{align}
since $h_{x}^0 + i \xi_P \in \text{Ran}(\Pi_0)$ whereas $h_{K,x}^1 \in \text{Ran}(\Pi_1)$. Finally we can replace $a$ and $a^\dagger$ by $\phi$, and then transform back with $\mathbb U_K$, using \eqref{eq: transformation of phi}, in order to obtain
\begin{align}
\eqref{eq: aux formula G21} \,  =\,  - \frac{2}{\alpha^2}  \int \D y \, n_{0,1}(y) \, \re \lsp \psi \otimes \Upsilon_K | \phi( h_\cdot^1 + \varphi ) T_y u_\alpha R \phi  ( h_{K,\cdot}^1 ) \psi \otimes \Upsilon_K \rsp_\h.
\end{align}
To summarize, we have shown that
\begin{align}
\mathcal G_{21} &  =  - \frac{2}{\alpha^2}  \int \D y\,  \re \lsp  G_K^0 | L_{2,y} G_K^0\rsp_\h   n_{0,1}(y) +  \int \D y \, \ell_2(y) n_{0,1}(y)   =  \mathcal G_{211} + \mathcal G_{212} 
\end{align}
with
\begin{subequations}
\begin{align}\label{eq: def of L2 and l2}
L_{2,y} & \, =\, \PP \phi ( h_\cdot^1 +  \varphi ) T_y   u_\alpha R  \phi( h_{K,\cdot}^1 ) \PP \\[1.5mm]
\ell_2(y) & \, =\,  2 \re \lsp \psi  | \langle \underline{ h_\cdot + \varphi_P } | \widetilde w_{P,y} \rangle_\2  T_y   u_\alpha R  \langle \widetilde w_{P,y}^1 | \underline{ h_{K,\cdot}^1 } \rangle_\2  \psi \rsp_\2 .\label{eq: def of l2}
\end{align}
\end{subequations}
In the first term we add and subtract the Gaussian,
\begin{align}
 \mathcal G_{211} &  =  - \frac{2}{\alpha^2} \int \D y \, \re \lsp G_K^0  |  L_{2,y} G_K^0 \rsp_\h \, e^{-\lambda \alpha^2 y^2 } \notag \\
 & \quad  - \frac{2}{\alpha^2} \int \D y \, \re \lsp G_K^0  |  L_{2,y}  G_K^0 \rsp_\h \big(  n_{0,1}(y) - e^{-\lambda \alpha^2 y^2 }  \big) \, =\,   \mathcal G_{211}^{\rm lo} + \mathcal G_{211}^{\rm err},
\end{align}
and proceed with $\mathcal G_{211}^{\rm lo}$ by inserting $h_{\cdot}^1 = h_{K,\cdot}^1 +( h_\cdot^1 - h_{K,\cdot}^1)$, $T_y = 1 + ( T_y-1 ) $ and $u_\alpha = 1 + ( u_\alpha -1 ) $,
\begin{align}
\mathcal G_{211}^{\rm lo} &  =    - \frac{2}{\alpha^2} \re \lsp G_K^0 | \phi( h_{K,\cdot}^1 + \varphi )   R  \phi( h_{K,\cdot}^1  )  G_K^0  \rsp_\h \int \D y \,  e^{-\lambda \alpha^2 y^2 }  \notag \\
& \quad  - \frac{2}{\alpha^2} \re \lsp  G_K^0 | \phi( h_{K,\cdot}^1 + \varphi ) (u_\alpha - 1)  R  \phi( h_{K,\cdot}^1  )   G_K^0  \rsp_\h \int \D y \, e^{-\lambda \alpha^2 y^2 }  \notag \\
& \quad - \frac{2}{\alpha^2} \int \D y \, \re \lsp  G_K^0 | \phi( h_{K,\cdot}^1 +\varphi  ) ( T_{y}-1) u_\alpha  R  \phi( h_{K,\cdot}^1  )    G_K^0\rsp_\h e^{-\lambda \alpha^2 y^2 }  \notag \\
& \quad - \frac{2}{\alpha^2} \int \D y \, \re \lsp  G_K^0 | \phi( h_{\cdot}^1 - h_{K, \cdot}^1 )  T_y u_\alpha  R  \phi( h_{K,\cdot}^1 )   G_K^0 \rsp_\h e^{-\lambda \alpha^2 y^2 } \notag\\
& =   \sum_{n=1}^4 \mathcal G_{211}^{{\rm lo},n}.
\end{align}
Since $\PP \phi(\varphi) R = 0$, we have $\mathcal G_{211}^{{\rm lo},1} = \frac{2}{\alpha^2} \lsp \Upsilon_K | ( \mathbb H_K -\mathbb N_1 )  \Upsilon_K \rsp_\Fock (\frac{\pi}{\lambda \alpha^2})^{3/2}$, cf. \eqref{eq: Bogoliubov Hamiltonian maintext}, and hence we can use Proposition \ref{prop: norm bound} to conclude that
\begin{align}\label{eq: bound for G_211^lo 1}
\bigg|\, \mathcal G_{211}^{{\rm lo},1} - \mathcal N \frac{2}{\alpha^2}  \lsp \Upsilon_K |  (\mathbb H_K -\mathbb N_1) \Upsilon_K \rsp_\Fock \, \bigg| \, \le \,  C_\varepsilon \sqrt K \alpha^{-6 + \varepsilon} .
\end{align}
For the other terms, we shall show the combined error estimate
\begin{align}\label{eq: bound for G_211^lo rest}
|\mathcal G_{211}^{{\rm lo},2} | + |\mathcal G_{211}^{{\rm lo},3} | + |\mathcal G_{211}^{{\rm lo},4}  | \,  \le\,  C \big( \sqrt K \alpha^{-6} + K^{-1/2} \alpha^{-5} \big) .
\end{align}
In the last term, we recall $h_\cdot(y)= h_{K=\infty,\cdot}(y)$, and apply Lemma \ref{lem: LY CM} in combination with $\sno R^{1/2} u_\alpha T_{-y} \nabla\sno_\op \le C$. This gives
\begin{align}
|\mathcal G_{211}^{\rm lo,4}  | & \, \le\,  \frac{2}{\alpha^2} \int \D y \, e^{-\lambda \alpha^2 y^2} \, \sno R^{1/2} u_\alpha T_{-y} \phi( h_{\cdot}^1 - h_{K, \cdot}^1 )  \PP G_K^0 \sno_\h \, \sno R^{1/2}  \phi( h_{K,\cdot}^1 ) \PP G_K^0  \sno_\h \notag\\
& \, \le \,  C \alpha^{-5} K^{-1/2}.
\end{align}
Next we write $T_y - 1 = \int_0^1 \D s T_{sy} (y \nabla)$ in the third term to obtain an additional $|y|$,
\begin{align}
|\mathcal G_{211}^{\rm lo,3} | & \, \le\,   \frac{2}{\alpha^2} \bigg(\int \D y \, |y| e^{-\lambda \alpha^2 y^2} \bigg)  \sno \nabla u_\alpha R^{1/2} \sno_{\op} \, \sno  \phi( h_{K,\cdot}^1 +\varphi  ) G_K^0 \sno_\h   \sno    R^{1/2} \phi( h_{K,\cdot}^1  )   G_K^0 \sno_\h \notag\\
& \, \le \, C  \alpha^{-6} \sqrt K,
\end{align}
where the factor $\sqrt K$ comes from the $L^2$ norm of $h_{K,0}^1$ in the bound on the first field operator (since $\Delta R^{1/2}$ is unbounded, we can not apply the commutator method to this part). In the second term, we use $\psi(x)\le C e^{-|x|/C}$ for some $C>0$, and thus $\sno (u_\alpha -1) \psi\sno_\2 \le C e^{-\alpha/C}$, to estimate
\begin{align}
|\mathcal G_{211}^{\rm lo,2} | \, \le\,  \frac{C}{\alpha^5} \sno (u_\alpha-1) \psi \sno_\2 \, \sno \phi( h_{K,\cdot}^1 + \varphi )  R  \phi( h_{K,\cdot}^1  ) G_K^0 \sno_\h \, \le \, C \sqrt K  e^{-\alpha/C} .
\end{align}
This proves \eqref{eq: bound for G_211^lo rest}.

To bound the remaining contributions in $\mathcal G_{211}^{\rm err}$ and $\mathcal G_{212}$, we shall use
\begin{subequations}
\begin{align}
\I  \lsp  G_K^0 |  L_{2,y} G_K^0  \rsp \I & \, \le\,    C f_{2,\alpha}(y) \label{eq: bound for L_2} \\[1mm]
|\ell_2 (y) | & \, \le\, C f_{2,\alpha}(y)  ( y^2 + \alpha^{-2} ) ( | y |   + |y|^3 + \alpha^{-2} )  \label{eq: bound for l_2}
\end{align}
\end{subequations}
where 
\begin{align}\label{eq: def of f_2}
f_{2,\alpha}(y) =  \sno u_\alpha T_{-y} \PP  \sno_{\op} + \sno \nabla u_\alpha T_{-y} \PP  \sno_{\op}.
\end{align}
Using the exponential decay of $\psi$ and $|\nabla^ku_\alpha|(y) \le \mathbbm{1}(|y|\le 2\alpha)$, for $k=0,1$, it is easy to show that
\begin{align}\label{eq: bound for f2}
\sno f_{2,\alpha} \sno_{\su} \le C \quad \text{and} \quad \sno |\cdot |^n f_{2,\alpha} \sno_{\1} \le C_n \alpha^{3+n} \quad \text{for all}\ n\in \mathbb N_0.
\end{align}
To verify \eqref{eq: bound for L_2} and \eqref{eq: bound for l_2}, use $u_\alpha T_{-y}\phi(h_\cdot)=  \phi(h_{\cdot - y } ) u_\alpha T_{-y}$ and Cauchy--Schwarz to bound
\begin{align}
\I \lsp  G_K^0 |  L_{2,y}   G_K^0  \rsp_\h \I  & \, \le \,  \sno R^{1/2} \phi(h_{\cdot-y}^1 + \varphi) u_\alpha T_{-y} \PP G_K^0 \sno_\h \, \sno R^{1/2} \phi(h_{K,\cdot}^1) \PP G_K^0 \sno_\h .
\end{align}
Now we can use \eqref{eq: standard estimates for a and a*} and Lemma \ref{lem: LY CM} to obtain \eqref{eq: bound for L_2}. To estimate $\ell_2(y)$, defined in \eqref{eq: def of l2}, we proceed with
\begin{subequations}
\begin{align}
|\ell_2(y)| & \, \le\, 2  \I \lsp\psi |  T_y u_\alpha  \langle \underline{h_{\cdot-y}} | \widetilde w_{P,y} \rangle_\2  R  \langle \widetilde w_{P,y}^1 | \underline{h_{K,\cdot}^1} \rangle_\2 \psi \rsp_\2 \I   \label{eq: first line ell2}  \\[1mm]
&\quad \, +\, 2 \I \lsp \psi | T_y   u_\alpha   \langle \underline{\varphi_P} | \widetilde w_{P,y} \rangle_\2  R  \langle \widetilde w_{P,y}^1 | \underline{h_{K,\cdot}^1} \rangle_\2 \psi \rsp_\2 \I ,
\end{align}
\end{subequations}
and considering the first line, we use Cauchy--Schwarz, write out the two inner products (in the phonon variable) and then use Cauchy--Schwarz again, 
\begin{align}\label{eq: finr line ell2 2}
| \eqref{eq: first line ell2} | \, &   \le \, 2 \int \D u\, |\widetilde w_{P,y}(u)|  \, \sno  \PP T_y u_\alpha  \underline{h_{\cdot-y}}(u) R^{1/2} \sno_{\op} \int \D z \, |\widetilde w_{P,y}^1 (z) |\,   \sno R^{1/2}  \underline{ h_{K,\cdot}^1}(z) \psi \sno \notag\\
& \le  \, 2 \sno \widetilde w_{P,y}\sno_\2 \sno \widetilde w_{P,y}^1\sno_\2  \,\bigg( \int \D u \sno  \PP T_y u_\alpha  \underline{h_{\cdot-y}}(u)  R^{1/2} \sno_{\op}^2 \,    \int \D z \sno R^{1/2} \underline{h_{K,\cdot}^1}(z) \psi \sno_\2^2 \bigg)^{1/2} \notag\\[1.5mm]
& \le \, C f_{2,\alpha}(y) ( |y|   + y^3 + \alpha^{-2}) ( y^2 + \alpha^{-2}) ,
\end{align}
where the last step follows from Lemma \ref{lem: bound for w1 and w0} and Corollary \ref{cor: X h Y bounds} together with $\underline{ h_{K,\cdot} } = h_{K,\cdot}^0 + \Theta_K^{-1} h_{K,\cdot}^1$.
Since the second line is estimated similarly, we arrive at \eqref{eq: bound for l_2}. With \eqref{eq: bound for L_2} at hand we can apply Lemma \ref{lem: Gaussian lemma} and \eqref{eq: bound for f2} to get
\begin{align}\label{eq: bound for G_221^err}
| \mathcal G_{211}^{\rm err} | \,  \le \,  \frac{2}{\alpha^2}  \int \D y \, \I \lsp G_K^0 |  L_{2,y} G_K^0  \rsp_\h \I \, \I  n_{0,1}(y) - e^{-\lambda \alpha^2 y^2 } \I \,  \le \,  C  \alpha^{-6},
\end{align}
and further, using \eqref{eq: bound for l_2} and Corollary \ref{cor: Gaussian for errors}, we obtain
\begin{align}\label{eq: bound for G_212}
|\mathcal G_{212}  | \,  \le \,  C \int \D y\, |\ell_2 (y) |\,  n_{0,1}(y) \,  \le \,  C \alpha^{-6}.
\end{align}
This completes the analysis of $\mathcal G_{21}$. 

Next we introduce $ R_{4,y}  =  R_{4,y}^1  +  R_{4,y}^2 $ with
\begin{subequations}
\begin{align}
 R_{4,y}^1 & \, =\, \PP \phi( h_{K,\cdot}^1  ) R^{\frac{1}{2}}  ( e^{ - A_{P,y} } -1 )  R^{\frac{1}{2}}  \phi ( h_{\cdot-y}  +  \varphi_P )  u_\alpha   T_{-y} \PP  \\[0.5mm]
 R_{4,y}^2 & \, =\,  2  \alpha \PP \lsp h_{K,\cdot}  | \re ( w_{P,y}^1 ) \rsp_\2 R^{\frac{1}{2}}  ( e^{ - A_{P,y} } -1 ) R^{\frac{1}{2}}   \phi ( h_{\cdot-y} +  \varphi_P )u_\alpha   T_{-y} \PP .
\end{align}
\end{subequations}
Inserting \eqref{eq: decomposition of Upsilon} and \eqref{eq: identity for WG1} into \eqref{eq: def of G_22} it follows that 
\begin{align}
\mathcal G_{22}& \,  =\,  -  \frac{2}{\alpha^2} \int \D y\, \re \lsp R_{4,y} \psi \otimes \big( \Upsilon_K^< + \Upsilon_K^> \big)  | W(\alpha w_{P,y}) G_K^0 \rsp_\h \, = \, \mathcal G_{22}^< + \mathcal G_{22}^>.
\end{align}
With the aid of Lemma \ref{lem: LY CM} we obtain
\begin{align}
\sno  R_{4,y}^1 \Psi \sno_\h & \, \le\,   C  \sno  ( e^{ - A_{P,y} } -1 )  (\mathbb N+1)^{1/2} R^{1/2}  \phi ( h_{\cdot-y} +  \varphi_P ) u_\alpha T_{-y} \PP \Psi \sno_\h \label{eq: bound for R4 a},
\end{align}
and proceeding similarly as in \eqref{eq: R12 bound}, we find
\begin{align}
\sno  R_{4,y}^2 \Psi \sno_\h  &\,  \le\,  C  \alpha (y^2 + \alpha^{-2} )  \sno  ( e^{ - A_{P,y} } -1 )  R^{1/2} \phi ( h_{\cdot-y} +  \varphi_P ) u_\alpha T_{-y}  \PP \Psi  \sno_\h  \label{eq: bound for R4} .
\end{align}
For $\Psi = \psi \otimes \Upsilon_K^>$, a second application of Lemma \ref{lem: LY CM} (after using unitarity of $e^{-A_{P,y}}$) together with $\sno \varphi_P \sno^2_\2 \le C$ for $|P|/ \alpha \le c$ and Corollary \ref{Cor: Upsilon estimates} is sufficient to find
\begin{align} \label{eq: R4 Upsilon> bound}
\sno R_{4,y} \psi \otimes \Upsilon_K^> \sno_\h & \,  \le\,  C \big( \sno u_\alpha T_{- y} \PP \sno_{\op} +  \sno \nabla u_\alpha T_{ - y} \PP \sno_{\op} \big) (1+ \alpha y^2) \sno (\mathbb N+1) \Upsilon_K^>\sno_{\mathcal F} \notag\\[1mm]
& \,  \le\,  C_\delta \, \alpha^{-10} f_{2,\alpha}(y)(1+\alpha y^2) 
\end{align}
with $f_{2,\alpha}$ defined in \eqref{eq: def of f_2}. Using this bound in $G_{22}^>$ and recalling Corollary \ref{Cor: Upsilon estimates} and \eqref{eq: bound for f2} we thus obtain 
\begin{align}\label{eq: bound for G_22^>}
|\mathcal G_{22}^>| \le C_\delta \, \alpha^{- 6}.
\end{align}
In $\mathcal G_{22}^<$ we proceed by inserting \eqref{eq: number op identity} and use \eqref{eq: bound for epsilon N} and Lemma \ref{lem: exp N bound}. This gives
\begin{align}\label{eq: bound for G21}
| \mathcal G_{22}^< | \, \le\,  \frac{2\sqrt 2}{\alpha^2} \int \D y\, \sno  R_{4,y} \psi \otimes \Upsilon_K^< \sno_\h\, n_{\delta,\eta}(y).
\end{align}
The derivation of a suitable bound for the norm in the integrand is more cumbersome, so we go through it step by step. To shorten the notation let $G_K^{0<} = \psi \otimes \Upsilon_K^<$. We start from \eqref{eq: bound for R4 a} and \eqref{eq: bound for R4} where we insert $h_\cdot = h_{K,\cdot} + (  h_{\cdot } - h_{K, \cdot } ) $ and use the triangle inequality,
\begin{subequations}
\begin{align}
\sno  R_{4,y}^1 G_K^{0<}  \sno_\h & \, \le\,   C  \sno  ( e^{ - A_{P,y} } -1 )  (\mathbb N+1)^{1/2} R^{\frac{1}{2}}  \phi ( h_{K,\cdot-y} + \varphi_P ) u_\alpha T_{-y} G_K^{0<} \sno_\h \label{eq: R4 line 1} \\[1mm]  
& \quad   + \, C \sno  ( e^{ - A_{P,y} } -1 )  (\mathbb N+1)^{1/2} R^{\frac{1}{2}}  \phi ( h_{\cdot-y} -  h_{K, \cdot -y } ) u_\alpha T_{-y} G_K^{0<}   \sno_\h   ,  \label{eq: R4 line 2} \\[2mm]  
 \sno  R_{4,y}^2 G_K^{0<}  \sno_\h  & \, \le \,  C \alpha ( y^2 + \alpha^{-2} ) \sno  ( e^{ - A_{P,y} } -1 )  R^{\frac{1}{2}}  \phi ( h_{K,\cdot-y} +  \varphi_P )u_\alpha T_{-y} G_K^{0<}\sno_\h \label{eq: R4 line 3} \\[1mm]
&   \hspace{-0.3cm}  + \, C  \alpha ( y^2 + \alpha^{-2} )  \sno  ( e^{ - A_{P,y} } -1 ) R^{\frac{1}{2}}  u_\alpha \phi (  h_{\cdot-y} -  h_{K,\cdot-y} ) u_\alpha T_{-y} G_K^{0<} \sno_\h \label{eq: R4 line 4}  .
\end{align}
\end{subequations}
For the second and fourth line, we apply Lemma \ref{lem: LY CM} a second time (after bringing $(\mathbb N+1)^{1/2}$ to the right of $a$ and $a^\dagger$) to find
\begin{align}
\eqref{eq: R4 line 2} + \eqref{eq: R4 line 4} & \, \le\,  C K^{-1/2} (1 + \alpha y^2) ( \sno u_\alpha T_{-y} \PP \sno_{\op}  + \sno \nabla u_\alpha T_{-y} \PP \sno_{\op} ) \sno (\mathbb N+1) \Upsilon_K^< \sno_\Fock \notag \\[0.5mm]
& \,  \le\,  C K^{-1/2} (1 + \alpha y^2) f_{2,\alpha}(y).\label{eq: bounf for R4 line 2 and 4}
\end{align}
In the first and third line, we use the functional calculus and write out $A_{P,y}= i P_f y + i g_{P}(y)$,
\begin{subequations}
\begin{align}
\!\!\!\eqref{eq: R4 line 1} + \eqref{eq: R4 line 3} &  \le  C  \sno (P_f y) (\mathbb N+1)^{1/2} R^{\frac{1}{2}}   \phi ( h_{K,\cdot-y} +  \varphi_P ) u_\alpha T_{-y} G_K^{0<} \sno_\h \label{eq: R4 line 5}\\[1mm]  
&  \quad +  C \alpha (y^2 + \alpha^{-2} )  \sno (P_f y)  R^{\frac{1}{2}}  \phi (  h_{K,\cdot-y} +  \varphi_P ) u_\alpha T_{-y} G_K^{0<}  \sno_\h \label{eq: R4 line 6} \\[1mm]
&  \quad  +  C |g_{P}(y)|  \sno (\mathbb N+1)^{1/2} R^{1/2}  \phi ( h_{K,\cdot-y} +  \varphi_P ) u_\alpha T_{-y} G_K^{0<}  \sno_\h \label{eq: R4 line 7} \\[1mm]  
&  \quad +   C \alpha (y^2 + \alpha^{-2} ) |g_{P}(y)| \sno  R^{\frac{1}{2}}   \phi (  h_{K,\cdot-y} +  \varphi_P ) u_\alpha T_{-y} G_K^{0<} \sno_\h .\label{eq: R4 line 8}
\end{align}
\end{subequations}
Now we use $[iP_f y,\phi(f)]=\pi(y\nabla f)$ such that we can estimate the first line by
\begin{align}
\eqref{eq: R4 line 5}  &\, \le \,  C \big(  \sno  (\mathbb N+1)^{1/2} R^{1/2}  \phi ( h_{K,\cdot-y} +  \varphi_P ) (P_f y ) u_\alpha T_{-y} G_K^{0<}  \sno_\h \nonumber \\[1.5mm]
& \quad \quad \quad + \sno  (\mathbb N+1)^{1/2} R^{1/2} \pi ( y \nabla h_{K,\cdot-y} + y\nabla \varphi_P )u_\alpha T_{-y}  G_K^{0<}   \sno_\h  \big).
\end{align}
To bound the first line, we use again Lemma \ref{lem: LY CM}, while in the second line we use $(\nabla h_K)_\cdot = - \nabla (h_{K,\cdot}) = - [\nabla , h_{K,\cdot}]$ and \eqref{eq: standard estimates for a and a*} together with $\sno \nabla \varphi_P \sno_\2 \le C $ for $|P|/ \alpha \le c$. Together we obtain
\begin{align}
\eqref{eq: R4 line 5} & \, \le\,   C |y| \big(  \sno u_\alpha T_{-y} \PP \sno_{\op} +  \sno \nabla u_\alpha T_{-y}  \PP \sno_{\op} \big)  \big( \sno (\mathbb N+1) P_f \Upsilon_K^<  \sno_\Fock + \sqrt K \sno (\mathbb N+1) \Upsilon_K^< \sno_\Fock \big)\notag \\[1.5mm]
& \, \le\,   C \alpha^{\delta}|y| f_{2,\alpha}(y) \big( \sno  P_f \Upsilon_K^<  \sno_\Fock + \sqrt K \big) \notag\\[2mm]
& \, \le\,  C \alpha^{\delta} \sqrt K |y| f_{2,\alpha}(y),\label{eq: bounf for R4 line 5}
\end{align}
where the factor $\sqrt K$ in the first step comes from the $L^2$-norm of $h_{K,0}$, and the last step follows from Lemma \ref{lem: bounds for P_f}. In a similar fashion, one shows
\begin{align}
\eqref{eq: R4 line 6} \, \le\,   C   \alpha^{\delta} \sqrt K |y| (1 + \alpha y^2) f_{2,\alpha}(y),\label{eq: bounf for R4 line 6}
\end{align}
and, with \eqref{eq: bound for g_P}, one also verifies
\begin{align}
\eqref{eq: R4 line 7} +  \eqref{eq: R4 line 8} \,   \le \, C \alpha^\delta (\alpha^2 |y|^5 + \alpha |y|^3 ) f_{2,\alpha}(y).\label{eq: bounf for R4 line 7 and 8}
\end{align}
Collecting the estimates \eqref{eq: bounf for R4 line 2 and 4}, \eqref{eq: bounf for R4 line 5}, \eqref{eq: bounf for R4 line 6} and \eqref{eq: bounf for R4 line 7 and 8} we arrive at
\begin{align}\label{eq: bound for R4 Upsilon<}
\sno R_{4,y} \psi \otimes \Upsilon_K^< \sno_\h \, \le\,  C  f_{2,\alpha}(y)\alpha^{\delta} \Big(  K^{-\frac{1}{2}} ( 1 + \alpha y^2)  + \alpha^2 |y|^5 + \sqrt K ( |y| + \alpha |y|^3) \Big).
\end{align}
Now we can apply Corollary \ref{cor: Gaussian for errors} together with \eqref{eq: bound for f2} to bound the right side of \eqref{eq: bound for G21}. The result is
\begin{align}\label{eq: bound for G_22^<}
| \mathcal G_{22}^< | \, \le \, C \alpha^{-2+\delta } \big( K^{-1/2} \alpha^{-3} + \sqrt K  \alpha^{-4+4\delta}\big).
\end{align}

In view of the estimates \eqref{eq: bound for G_1^>}, \eqref{eq: bound for G_1^y}, \eqref{eq: bound for G_211^lo 1}, \eqref{eq: bound for G_211^lo rest}, \eqref{eq: bound for G_221^err}, \eqref{eq: bound for G_212}, \eqref{eq: bound for G_22^>} and \eqref{eq: bound for G_22^<}, the proof of Proposition \ref{prop: bound for G} is now complete.
\end{proof}

\subsection{Energy contribution $\mathcal K$}

Recall that $\mathcal K$ was defined in \eqref{eq: K}.

\begin{proposition}\label{prop: bound for K} Let $\mathbb H_K$ as in \eqref{eq: Bogoliubov Hamiltonian maintext}, $\mathbb N_1 = \D \Gamma(\Pi_1)$ and choose $c>0$. For every $\varepsilon> 0$ there exists a constant $C_\varepsilon >0$ \textnormal{(}we omit the dependence on $c$\textnormal{)} such that
\begin{align}
\bigg| \mathcal K + \mathcal N \frac{1}{\alpha^2}  \lsp \Upsilon_K| ( \mathbb H_K - \mathbb N_1 ) \Upsilon_K \rsp_\Fock \, \bigg| \le C_\varepsilon\, \alpha^\varepsilon \big( \sqrt K \alpha^{-6} + K^{-1/2} \alpha^{-5} \big)
\end{align}
for all $|P|/\alpha  \le c $ and all $K,\alpha$ large enough.
\end{proposition}

\begin{proof} We split this contribution into three terms
\begin{align}
\mathcal K & \, = \, \frac{1}{\alpha^2} \int \D y \,  \lsp G_K^1 | \big( h^{\rm Pek} + \alpha^{-2} \mathbb N + \alpha^{-1} \phi(h_\cdot + \varphi_P) \big)  T_y e^{A_{P,y}} W(\alpha w_{P,y})  G_K^1 \rsp_\h \notag\\
& \, = \,  \mathcal K_1 + \mathcal K_2 + \mathcal K_3,
\end{align}
and note that $\mathcal K_1$ provides the energy contribution of order $\alpha^{-2}$.\medskip

\noindent \underline{Term $\mathcal K_1$}. We start again by writing
\begin{align}
\mathcal K_1 &\,  =\,  \frac{1}{\alpha^2} \int \D y\, \lsp  G_K^1 | h^{\rm Pek} T_y W(\alpha w_{P,y})    G_K^1 \rsp_\h   \notag\\
& \quad +  \frac{1}{\alpha^2} \int \D y \, \lsp G_K^1 |h^{\rm Pek} T_y (e^{A_{P,y}}-1)   W(\alpha w_{P,y})  G_K^1 \rsp_\h  \,  =\,  \mathcal K_{11} + \mathcal K_{12}.
\end{align}
and proceed for the first term similarly as in the computation of $\mathcal G_2$, see \eqref{eq: def of G_22}. This leads to
\begin{align}
\mathcal K_{11} & \,  =\,  
\frac{1}{\alpha^2} \int \D y \, \lsp  G_K^0 | \phi( h_{K,\cdot}^1 ) R u_\alpha h^{\rm Pek} T_y   W(\alpha  w_{P,y}) u_\alpha R \phi( h_{K,\cdot}^1 ) G_K^0 \rsp_\h \notag\\
& \, = \, \frac{1}{\alpha^2} \int \D y \, \lsp  \psi \otimes \Omega| a (\underline{h_{K,\cdot}^1}) R u_\alpha h^{\rm Pek} T_y W(\alpha \widetilde w_{P,y}) u_\alpha R a^\dagger(\underline {h_{K,\cdot}^1 } )  \psi \otimes \Omega \rsp_\h \notag\\
& \, =\,  \frac{1}{\alpha^2} \int \D y\,  \lsp  G_K^0 | L_{3,y}  G_K^0 \rsp_\h n_{0,1}(y) - \int \D y \,  \ell_3(y) n_{0,1}(y) \,  =\,  \mathcal K_{111} + \mathcal K_{112} 
\end{align}
where
\begin{subequations}
\begin{align}
L_{3,y} &\,  =\,  \PP \phi({h_{K,\cdot}^1}) R u_\alpha h^{\rm Pek}  T_y u_\alpha R \phi({h_{K,\cdot}^1 } ) \PP \\[1mm]
\ell_3(y) & \, =  \, \lsp \psi|  \langle \underline{h_{K,\cdot}^1}|  \widetilde w_{P,y}^1 \rangle_\2 R u_\alpha h^{\rm Pek}  T_y u_\alpha R  \langle  \widetilde w_{ P,y }^1 | \underline{h_{K,\cdot}^1} \rangle_\2 \psi \rsp_\2 .
\end{align}
\end{subequations}
We go on with
\begin{align}
\mathcal K_{111} & \, =\,  \frac{1}{\alpha^2} \int \D y\, \lsp  G_K^0 | L_{3,y}  G_K^0 \rsp_\h \, e^{-\lambda \alpha^2 y^2 } \notag\\
&  \quad + \frac{1}{\alpha^2} \int \D y\, \lsp  G_K^0 | L_{3,y}  G_K^0 \rsp_\h  \big( n_{0,1}(y) - e^{-\lambda \alpha^2 y^2 } \big) =  \mathcal K_{111}^{\rm lo} + \mathcal K_{111}^{\rm err},
\end{align}
and in the leading-order term, we insert $T_y= 1 + (T_y-1)$ and $u_\alpha = 1 + (u_\alpha -1 )$,
\begin{align}
\mathcal K_{111}^{\rm lo}  & \, =\,  \frac{1}{\alpha^2}  \lsp    G_K^0 |  \phi( h_{K,\cdot}^1 ) R h^{\rm Pek}  R \phi ( h_{K,\cdot}^1 )   G_K^0 \rsp_\h \int \D y\,  e^{-\lambda \alpha^2 y^2 }  \notag\\
& \quad +  \frac{1}{\alpha^2} \lsp  G_K^0 |  \phi( h_{K,\cdot}^1 ) R (u_\alpha-1) h^{\rm Pek}  R \phi ( h_{K,\cdot}^1 )  G_K^0  \rsp_\h  \int \D y\, e^{-\lambda \alpha^2 y^2 }  \notag\\
& \quad +  \frac{1}{\alpha^2} \lsp   G_K^0 | \phi( h_{K,\cdot}^1 ) R u_\alpha h^{\rm Pek}  (u_\alpha-1) R \phi ( h_{K,\cdot}^1 )   G_K^0 \rsp_\h \int \D y \, e^{-\lambda \alpha^2 y^2 }  \notag\\
& \quad + \frac{1}{\alpha^2} \int \D y \, \lsp   G_K^0 | \phi( h_{K,\cdot}^1 ) R u_\alpha h^{\rm Pek}  (T_{y}-1) u_\alpha R \phi ( h_{K,\cdot}^1 )  G_K^0 \rsp_\h e^{-\lambda \alpha^2 y^2 }  \notag\\ 
& \, =\,  \sum_{n=1}^4 \mathcal K_{111}^{\rm lo,n}.
\end{align}
Since $R h^{\rm Pek} R = R$, one finds $ \mathcal K_{111}^{\rm lo,1} = -  \frac{1}{\alpha^2}  \lsp \Upsilon_K | ( \mathbb H_K - \mathbb N_1 )  \Upsilon_K  \rsp_\Fock (\frac{\pi}{\lambda \alpha^2})^3$, cf. \eqref{eq: Bogoliubov Hamiltonian maintext}, and with the aid of Proposition \ref{prop: norm bound}, this gives the leading-order contribution
\begin{align}\label{eq: bound for K111 lo1}
\bigg| \mathcal K_{111}^{\rm lo,1} + \mathcal N \frac{1}{\alpha^2} \lsp \Upsilon_K | ( \mathbb H_K - \mathbb N_1 ) \Upsilon_K \rsp_\Fock  \bigg| \, \le\, C_\varepsilon   \sqrt K \alpha^{ - 6 + \varepsilon} .
\end{align}
For the other terms, we shall show that
\begin{align}\label{eq: bound for K111 lo rest}
| \mathcal K_{111}^{\rm lo,2} | + | \mathcal K_{111}^{\rm lo,3} | + | \mathcal K_{111}^{\rm lo,4} | \le  C  \sqrt K \alpha^{-6} .
\end{align}
In the second term we use $h^{\rm Pek}R = \QQ = 1- \PP$ to write
\begin{align}
K_{111}^{\rm lo,2} & \, =\,   \alpha^{-2} \lsp  G_K^0 |  \phi( h_{K,\cdot}^1 ) R (u_\alpha-1) (1 -  \PP)  \phi ( h_{K,\cdot}^1 )  G_K^0  \rsp_\h \bigg( \frac{\pi}{\lambda \alpha^2}\bigg)^{3/2}
\end{align}
which is exponentially small in $\alpha$, since $\sno (u_\alpha-1) \psi \sno_\2 \le C  e^{-\alpha / C}$, and thus with Lemma \ref{lem: LY CM} one obtains $|\mathcal K_{111}^{\rm lo,2}| \le C \sqrt K  e^{-\alpha / C }$. In the next term we use $ [h^{\rm Pek}, u_\alpha-1]  = - [\Delta , u_\alpha] $ and again $h^{\rm Pek} R = 1-\PP$ to get
\begin{align}
\mathcal K_{111}^{\rm lo,3} &\,  =\,  \alpha^{-2}  \lsp G_K^0 |  \phi( h_{K,\cdot}^1 ) R u_\alpha  (u_\alpha-1) (1 - \PP )\phi ( h_{K,\cdot}^1 )  G_K^0 \rsp_\h \bigg( \frac{\pi}{\lambda \alpha^2}\bigg)^{3/2}   \notag\\
& \quad - \alpha^{-2}  \lsp G_K^0 | \phi( h_{K,\cdot}^1 ) R [\Delta , u_\alpha] R \phi ( h_{K,\cdot}^1 )  G_K^0 \rsp_\h \bigg( \frac{\pi}{\lambda \alpha^2}\bigg)^{3/2} .
\end{align}
Here the first line is bounded again exponentially in $\alpha$, whereas in the second line we use $[\Delta,u_\alpha] = 2 (\nabla u_\alpha) \nabla + (\Delta u_\alpha)$ and $\sno \nabla u_\alpha \sno_{\su} + \sno \Delta u_\alpha \sno_{\su} \le C\alpha^{-1}$, see \eqref{eq: properties of ualpha}. Together with Lemmas \ref{lem: bound for R} and \ref{lem: LY CM}, this implies $ |\mathcal K_{111}^{\rm lo,3}| \le C \alpha^{- 6 }$. In the last term we employ $T_y-1 = \int_0^1 \D s T_{sy} (y\nabla)$, $[h^{\rm Pek}, u_\alpha]  = - [\Delta , u_\alpha]$ and $h^{\rm Pek}R = \QQ$ to find
\begin{align}
\mathcal K_{111}^{\rm lo , 4} &  =   \alpha^{-2} \int \D y\, \int_0^1 \D s\, \lsp   G_K^0| \phi( h_{K,\cdot}^1 ) \QQ u_\alpha T_{sy} (y\nabla) u_\alpha R \phi ( h_{K,\cdot}^1 )   G_K^0 \rsp_\h e^{-\lambda \alpha^2 y^2 }  \notag\\
& \quad + \, \alpha^{-2} \int \D y\, \int_0^1 \D s\, \lsp  G_K^0 | \phi( h_{K,\cdot}^1 ) R [\Delta, u_\alpha] T_{sy} (y\nabla) u_\alpha R \phi ( h_{K,\cdot}^1 )  G_K^0 \rsp_\h e^{-\lambda \alpha^2 y^2 } .
\end{align}
In both lines there is an additional factor $y$, and together with \eqref{eq: properties of ualpha}, we thus obtain
\begin{align}
|\mathcal K_{111}^{ \rm lo,4}| &   \le  C\alpha^{-6} \sno \phi( h_{K,\cdot}^1 ) G_K^0 \sno_\h \, \sno \nabla u_\alpha R^{1/2}\sno_\op \sno R^{1/2} \phi( h_{K,\cdot}^1 ) G_K^0 \sno_\h  \notag \\[2mm]
& \quad +  C\alpha^{-6} \sno R^{1/2} \phi( h_{K,\cdot}^1 ) G_K^0 \sno_\h \sno R^{1/2} [\Delta,u_\alpha]\sno_\op  \sno \nabla u_\alpha R^{1/2}\sno_\op \sno R \phi( h_{K,\cdot}^1 ) G_K^0 \sno_\h  \notag \\[2mm]
& \le  C ( \alpha^{-6} \sqrt K + \alpha^{-7}).
\end{align}
This proves \eqref{eq: bound for K111 lo rest}.

To estimate $\mathcal K_{112}$ and $\mathcal K_{111}^{\rm err}$, we make use of
\begin{subequations}
\begin{align}
\I \lsp G_K^0 | L_{3,y} G_K^0 \rsp_\h \I & \, \le\,  C  f_{3,\alpha}(y) \label{eq: bound for L_3}\\
|\ell_3(y) | &\,  \le \,  C f_{3,\alpha}(y) (y^4 + \alpha^{-4}) \label{eq: bound for l_3}
\end{align}
\end{subequations}
where
\begin{align}\label{eq: def f_3}
f_{3,\alpha}(y) &  =  \sno u_\alpha T_{y} u_\alpha \sno_{\op} + \sno (\nabla u_\alpha ) T_{y} u_\alpha \sno_{\op} + \sno  u_\alpha T_{y} (\nabla u_\alpha) \sno_{\op}  + \sno (\nabla u_\alpha) T_{y} (\nabla u_\alpha) \sno_{\op}.
\end{align}
Recalling that by definition $| \nabla^k u_\alpha(y) | \le \mathbbm {1}(|y|\le 2\alpha )$ for $k=0,1$, it follows that $f_{3,\alpha}(y) \le 4 \mathbbm {1} (|y|\le 4\alpha)$ and thus 
\begin{align} \label{eq: bound for f3}
\sno f_{3,\alpha}\sno_{\su}\le 4 \quad \text{and} \quad \sno |\cdot|^n f_{3,\alpha} \sno_\1 \le C_n \alpha^{3+n}\quad \text{for all}\ n\in \mathbb N_0.
\end{align}
In order to verify \eqref{eq: bound for L_3}, use $h^{\rm Pek} = -\Delta + V^{\varphi} - \lambda^{\rm Pek}$ to write
\begin{align}
 R^{\frac{1}{2}} u_\alpha  T_{y} h^{\rm Pek} u_\alpha R^{\frac{1}{2}} \, & =\,   R^{\frac{1}{2}} u_\alpha \big( (-i\nabla) T_y (-i\nabla) + T_y  ( V^{\varphi} - \lambda^{\rm Pek})\big) u_\alpha  R^{\frac{1}{2}} \notag \\[1mm]
&  =  -  R^{\frac{1}{2}} (-i\nabla u_\alpha)  T_y (-i \nabla u_\alpha)   R^{\frac{1}{2}}  +  R^{\frac{1}{2}}  (-i\nabla) u_\alpha  T_y u_\alpha (-i\nabla)  R^{\frac{1}{2}} \notag \\[1mm]
& \quad  + R^{\frac{1}{2}} (-i\nabla) u_\alpha  T_y (-i \nabla u_\alpha)   R^{\frac{1}{2}} - R^{\frac{1}{2}} (-i \nabla u_\alpha)  T_y   u_\alpha   (-i\nabla) R^{\frac{1}{2}}    \notag \\[1mm]
& \quad   +    R^{\frac{1}{2}} u_\alpha T_y u_\alpha (V^\varphi - \lambda^{\rm Pek})   R^{\frac{1}{2}}.\label{eq: u h u term}
\end{align}
Since $\sno V^{\varphi} R^{1/2} \sno_{\op} \le C (\sno R\sno_{\op} + \sno \nabla R^{1/2} \sno_{\op}) \le C$, see Lemma \ref{lem: bound for R}, it thus follows that
\begin{align}
\sno  R^{\frac{1}{2}} u_\alpha  T_{y} h^{\rm Pek} u_\alpha R^{\frac{1}{2}} \sno_{\op} \, \le \, C f_{3,\alpha}(y).
\end{align}
With this at hand one applies Lemma \ref{lem: LY CM} to conclude the bound stated in \eqref{eq: bound for L_3}. For $\ell_3(y)$ we proceed similarly as in \eqref{eq: finr line ell2 2}, that is
\begin{align}
| \ell_3(y) | & \, \le \,  \sno R^{1/2} u_\alpha h^{\rm Pek}  T_y u_\alpha R^{1/2} \sno_{\op}  \sno R^{1/2} \langle  \widetilde w_{P,y}^1 |  \underline{h_{K,\cdot}^1} \rangle_\2 \psi \sno_{\2}^2 \notag\\[1mm]
&\, \le \, f_{3,\alpha}(y) \, \sno  \widetilde w_{P,y}^1  \sno^2_\2 \, \int \D z \sno \PP  \underline{h_{K,\cdot}^1}(z)  R^{1/2} \sno_{\op}^2 \, \le \, C f_{3,\alpha}(y) (y^4 + \alpha^{-4}).
\end{align}
Now we can apply Lemma \ref{lem: Gaussian lemma} and \eqref{eq: bound for f3} to estimate
\begin{align}\label{eq: bound for K111 err}
| \mathcal K_{111}^{\rm err} |\,  \le\,   \frac{C}{\alpha^2} \int \D y\,  f_{3,\alpha}(y)\, | n_{0,1}(y) - e^{-\lambda \alpha^2 y^2 } | \, \le \, C \alpha^{-6},
\end{align}
and further invoke Corollary \ref{cor: Gaussian for errors} to obtain
\begin{align}\label{eq: bound for K112}
| \mathcal K_{112} | \, \le\,  \int \D y\, f_{3,\alpha}(y) (|y|^4 + \alpha^{-4}) n_{0,1}(y)  \le C \alpha^{-7}.
\end{align}

Next we come to $\mathcal K_{12}$ which we rewrite with the aid of \eqref{eq: decomposition of Upsilon} and \eqref{eq: identity for WG1} as
\begin{align}
\mathcal K_{12} \, =\,  \frac{1}{\alpha^2} \int \D y\, \lsp R_{5,y} \psi  \otimes \big( \Upsilon_K^< + \Upsilon_K^> \big) | W(\alpha w_{P,y}) G_K^0 \rsp_\h = \mathcal K_{12}^< + \mathcal K_{12}^>
\end{align}
with the operator $ R_{5,y} = R_{5,y}^1  + R_{5,y}^2$ and
\begin{subequations}
\begin{align}
R_{5,y}^1 & \, =\,  \PP \phi( h_{K,\cdot}^1  ) R u_\alpha  ( e^{ - A_{P,y} } -1 ) T_{-y} h^{\rm Pek} u_\alpha R \phi ( h_{K,\cdot}^1 ) \PP \\[1mm]
R_{5,y}^2 & \, =\,  2 \alpha \PP \lsp h_{K,\cdot} |\re ( w_{P,y}^1) \rsp_\2 R u_\alpha  ( e^{ - A_{P,y} } -1 ) T_{-y} h^{\rm Pek} u_\alpha R \phi ( h_{K,\cdot}^1 ) \PP.
\end{align}
\end{subequations}
Utilizing Lemma \ref{lem: LY CM} and \eqref{eq: bound for w perp a}, we have
\begin{align}
\sno R_{5,y}^1 \Psi \sno_\h  & \, \le\,  C  \sno ( e^{ - A_{P,y} } -1 )  (\mathbb N+1)^{1/2} R^{\frac{1}{2}} u_\alpha  T_{-y} h^{\rm Pek} u_\alpha R \phi ( h_{K,\cdot}^1 ) \PP \Psi \sno_\h, \label{eq: bound for R51}
\end{align}
and following the same steps as in \eqref{eq: R12 bound},
\begin{align}
\sno R_{5,y}^2 \Psi \sno  & \, \le\,  C \alpha ( y^2 + \alpha^{-2}) \sno ( e^{ - A_{P,y} } -1 ) R^{\frac{1}{2}} u_\alpha  T_{-y} h^{\rm Pek} u_\alpha R \phi ( h_{K,\cdot}^1 ) \PP \Psi \sno_\h. \label{eq: bound for R52}
\end{align}
After using unitarity of $e^{-A_{P,y}}$ and \eqref{eq: bound for f3}, we can apply Lemma \ref{lem: LY CM} another time to obtain
\begin{align}
\sno R_{5,y} \psi\otimes \Upsilon_K^> \sno_\h  & \, \le\,   C f_{3,\alpha}(-y) (1+\alpha y^2) \sno (\mathbb N+1) \Upsilon_K^>  \sno_\Fock.
\end{align}
Thus we can estimate the tail with the aid of Corollary \ref{Cor: Upsilon estimates} and \eqref{eq: bound for f3},
\begin{align}\label{eq: bound for K_12^>}
| \mathcal K_{12}^>| \, \le\,  \frac{C}{\alpha^2} \sno (\mathbb N+1) \Upsilon_K^> \sno_\Fock \int \D y \, f_{3,\alpha}(-y) (1+\alpha y^2)\,  \le\, C_\delta\, \alpha^{- 6 }.
\end{align}
Then we use \eqref{eq: identity for tilde w transformation}, \eqref{eq: bound for epsilon N} and apply Lemma \ref{lem: exp N bound} to get 
\begin{align}
| \mathcal K_{12}^< |& \,  \le \, \frac{1}{\alpha^2} \int \D y\, \sno \mathbb U_K e^{\kappa \mathbb N} R_{5,y} \psi \otimes \Upsilon_K^< \sno_\h \sno e^{-\kappa \mathbb N}  W(\alpha \widetilde w_{P,y}) \Omega \sno_\Fock \notag \\
& \, \le\,  \frac{\sqrt 2}{\alpha^2} \int \D y\, \sno R_{5,y} \psi \otimes \Upsilon_K^<  \sno_\h \, n_{\delta,\eta}(y)\label{eq: bound for K12<}.
\end{align}
To bound the norm in the integral, we proceed in close analogy to the steps following \eqref{eq: bound for G21}. We abbreviate again $G_K^{0<}  = \psi \otimes \Upsilon_K^<$ and start from \eqref{eq: bound for R51} and \eqref{eq: bound for R52}. With \eqref{eq: bound for f3}, the functional calculus and $A_{P,y} = iP_fy + i g_{P}(y)$, one finds
\begin{subequations}
\begin{align}
\sno R_{5,y} G_K^{0<} \sno_\h &\,  \le\,  C \big( f_{3,\alpha}(-y) \sno ( e^{ - A_{P,y} } -1 )  (\mathbb N+1)^{1/2} R^{\frac{1}{2}}  \phi ( h_{K,\cdot}^1 )  G_K^{0<} \sno_\h  \notag \\[1mm]
& \quad + \alpha ( y^2 + \alpha^{-2})f_{3,\alpha}(-y)\sno ( e^{ - A_{P,y} } -1 )  R^{\frac{1}{2}} \phi ( h_{K,\cdot}^1 )   G_K^{0<} \sno_\h  \big)  \notag \\[1mm]
& \, \le \,  C \big( f_{3,\alpha}(-y) \big(  \sno (yP_f)  (\mathbb N+1)^{1/2} R^{\frac{1}{2}}  \phi ( h_{K,\cdot}^1 )   G_K^{0<} \sno_\h  \label{eq: bound for R52a} \\[1mm]
& \quad + f_{3,\alpha}(-y) |g_{P}(y)| (\mathbb N+1)^{1/2} R^{\frac{1}{2}}  \phi ( h_{K,\cdot}^1 )   G_K^{0<} \sno_\h   \label{eq: bound for R52b}\\[1mm]
& \quad +  f_{3,\alpha}(-y)  ( \alpha y^2 + \alpha^{-1})  \sno (P_f y)   R^{\frac{1}{2}} \phi ( h_{K,\cdot}^1 )  G_K^{0<} \sno_\h  \label{eq: bound for R52c}\\[1mm]
& \quad +  f_{3,\alpha}(-y)  ( \alpha y^2 + \alpha^{-1}) |g_{P}(y)| \sno R^{\frac{1}{2}} \phi ( h_{K,\cdot}^1 ) G_K^{0<} \sno_\h  \big) \label{eq: bound for R52d}.
\end{align}
\end{subequations}
In the second and fourth line, we use $|g_{P}(y)| \le C\alpha |y|^3$ and Lemma \ref{lem: LY CM},
\begin{align}
\eqref{eq: bound for R52b} + \eqref{eq: bound for R52d}&  \, \le\,  C  (\alpha^2 |y|^5 + \alpha |y|^3) f_{3,\alpha}(-y)  \sno (\mathbb N+1) \Upsilon_K^<\sno_\Fock \notag\\
& \, \le\,  C  (\alpha^2 |y|^5 + \alpha |y|^3) f_{3,\alpha}(-y)  . \label{eq: final bound for R52b}
\end{align}
In the first and third line, we employ the commutator $[iP_f y,\phi(f)]=\pi(y\nabla f)$ to get
\begin{subequations}
\begin{align}
\eqref{eq: bound for R52a} + \eqref{eq: bound for R52c} &\,  \le\,  C \big( f_{3,\alpha}(-y)  \sno   (\mathbb N+1)^{1/2} R^{\frac{1}{2}}  \phi ( h_{K,\cdot}^1 )  (yP_f) G_K^{0<} \sno_\h  \label{eq: bound for R52aa} \\[1mm]
& \quad + f_{3,\alpha}(-y)   \sno   (\mathbb N+1)^{1/2} R^{\frac{1}{2}}  \pi (y \nabla h_{K,\cdot}^1 )  G_K^{0<} \sno_\h  \label{eq: bound for R52ab} \\[1mm]
& \quad + f_{3,\alpha}(-y)  ( \alpha y^2 + \alpha^{-1})  \sno R^{\frac{1}{2}} \phi ( h_{K,\cdot}^1 )  (yP_f) G_K^{0<} \sno_\h    \label{eq: bound for R52ca} \\[1mm]
& \quad + f_{3,\alpha}(-y)  ( \alpha y^2 + \alpha^{-1})  \sno R^{\frac{1}{2}} \pi (y \nabla h_{K,\cdot}^1 )  G_K^{0<} \sno_\h  \big) \label{eq: bound for R52cb}.
\end{align}
\end{subequations}
After another application of Lemma \ref{lem: LY CM}, we can use \eqref{eq: decomposition of Upsilon} and then Lemma \ref{lem: bounds for P_f} for the terms involving $P_f$,
\begin{align}
\eqref{eq: bound for R52aa} + \eqref{eq: bound for R52ca} \, & \le\,  C f_{3,\alpha}(-y) (\alpha y^2 + 1)\, |y|\, \sno (\mathbb N+1) P_f \Upsilon_K^<\sno_\Fock \notag \\[1mm]
& \le \, C f_{3,\alpha}(-y) (\alpha |y|^3 + |y|) \alpha^\delta \sqrt K ,\label{eq: final bound for R52aa}.
\end{align}
while in the other two lines, we use $(\nabla h_{K})_\cdot = - [\nabla ,h_{K,\cdot}]$, to obtain
\begin{align}
\eqref{eq: bound for R52ab} + \eqref{eq: bound for R52cb} & \,  \le \, C f_{3,\alpha}(-y)\, |y|\,  (\alpha y^2 + 1)  \sno h_{K,0}\sno_\2 \sno  (\mathbb N+1) \Upsilon_K^{<} \sno_\Fock \notag \\[1mm]
& \, \le \,  C f_{3,\alpha}(-y)(\alpha |y|^3 + |y|)\sqrt K .\label{eq: final bound for R52ab}
\end{align}
Collecting all estimates we have thus shown that
\begin{align}
\sno R_{5,y} \psi \otimes \Upsilon_K^< \sno_\h \, \le\,  C f_{3,\alpha}(-y) \alpha^\delta \Big( \alpha^2 |y|^5 + \sqrt K ( \alpha |y|^3 + |y| ) \Big).
\end{align}
Using this bound in \eqref{eq: bound for K12<} we can invoke Corollary \ref{cor: Gaussian for errors} together with \eqref{eq: bound for f3} in order to obtain 
\begin{align}\label{eq: final bound for K_12^<}
|\mathcal K_{12}^<| \le C \sqrt K \alpha^{-6+5\delta}.
\end{align}
\noindent\underline{Term $\mathcal K_2$}. Using \eqref{eq: decomposition of Upsilon} and \eqref{eq: identity for WG1}, one finds
\begin{align}
\mathcal K_2 & \, =\,  \frac{1}{\alpha^{4}} \int \D y \,  \lsp R_{6,y} \psi \otimes \big( \Upsilon_K^< + \Upsilon_K^> \big) | W(\alpha w_{P,y})  G_K^0 \rsp_\h \, =\,  \mathcal K_{2}^< + \mathcal K_{2}^> 
\end{align}
with the operator $R_{6,y} = R_{6,y}^1 + R_{6,y}^2$ and
\begin{subequations}
\begin{align}
R_{6,y}^1 & \, =\,  \PP \phi(h_{K,\cdot}^1) R u_\alpha \mathbb N T_{-y} e^{-A_{P,y}} u_\alpha R  \phi(h_{K,\cdot}^1) \PP \\[0.5mm]
R_{6,y}^2 & \, =\,  2 \alpha \PP \phi(h_{K,\cdot}^1) R u_\alpha \mathbb N T_{-y} e^{-A_{P,y}} u_\alpha R  \langle \re (w_{P,y}^1) | h_{K,\cdot} \rangle_\2 \PP.
\end{align}
\end{subequations}
With Lemma \ref{lem: LY CM} and \eqref{eq: bound for w perp a} it is not difficult to verify
\begin{align}
\sno  R_{6,y} \Psi \sno_\h  & \, \le\,  C \sno u_\alpha T_{-y} u_\alpha \sno_{\op} (1+ \alpha y^2) \sno (\mathbb N+1)^2 \Psi \sno_\h  ,
\end{align}
and since $\sno u_\alpha T_{-y} u_\alpha \sno_{\op} \le \mathbbm{1} (|y|\le 4 \alpha)$, we can use Corollary \ref{Cor: Upsilon estimates} to estimate the part with the tail by
\begin{align}\label{eq: bound for K2^>}
|\mathcal K_{2}^>  |\,  \le\,  \frac{C}{\alpha^{4}} \sno  (\mathbb N+1)^2 \Upsilon_K^>\sno_\Fock \int \D y\, \mathbbm{1}(|y|\le 4\alpha ) (1+ \alpha y^2) \, \le\,  C_{\delta}\, \alpha^{-8 }.
\end{align}
To treat $\mathcal K_{2}^<$ we proceed as in \eqref{eq: bound for K12<}, that is
\begin{align}
| \mathcal K_{2}^<| & \le \frac{\sqrt 2}{\alpha^4} \int \D y \, \sno R_{6,y} \psi \otimes \Upsilon_K^< \sno_\h \, n_{\delta,\eta}(y) \le  \frac{C}{\alpha^4} \int \D y \, \mathbbm{1} (|y|\le \alpha) (1+\alpha y^2) \, n_{\delta,\eta}(y) .
\end{align}
It now follows from Corollary \ref{cor: Gaussian for errors} that 
\begin{align}\label{eq: bound for K2^<}
| \mathcal K_{2}^<| \le C \alpha^{-7}.
\end{align}

\noindent\underline{Term $\mathcal K_3$}. This term is similarly estimated as the previous one. With the aid of \eqref{eq: decomposition of Upsilon} and \eqref{eq: identity for WG1}, we have
\begin{align}
\mathcal K_3 & \, =\,  \frac{1}{\alpha^3} \int \D y \,  \lsp R_{7,y} \psi \otimes \big( \Upsilon_K^< + \Upsilon_K^> \big) |  W(\alpha w_{P,y}) G_K^0 \rsp_\h \, =\,  \mathcal K_{3}^< + \mathcal K_{3}^>
\end{align}
with the operator $R_{7,y} = R_{7,y}^1 + R_{7,y}^2$ and 
\begin{subequations}
\begin{align}
R_{7,y}^1 & \, =\,  \PP \phi(h_{K,\cdot}^1) R u_\alpha e^{-A_{P,y}} T_{-y} \phi(h_{\cdot} + \varphi_P) u_\alpha R \phi(h_{K,\cdot}^1) \PP \\[0.5mm]
R_{7,y}^2 & \, =\,   2 \alpha \PP \langle \re ( w_{P,y}^1 ) | h_{K,\cdot} \rangle_\2 R u_\alpha e^{-A_{P,y}} T_{-y} \phi(h_{\cdot} + \varphi_P) u_\alpha R \phi(h_{K,\cdot}^1) \PP. 
\end{align}
\end{subequations}
Utilizing again Lemma \ref{lem: LY CM} and \eqref{eq: bound for w perp a}, one shows that
\begin{align}
\sno R_{7,y} \Psi \sno_\h & \, \le\,  C f_{3,\alpha}(-y) (1 + \alpha y^2) \sno (\mathbb N+1)^{3/2} \Psi \sno_\h  
\end{align}
with $f_{3,\alpha}$ defined in \eqref{eq: def f_3}. Invoking Corollary \ref{Cor: Upsilon estimates} and \eqref{eq: bound for f3} we thus find
\begin{align} \label{eq: bound for K3^>}
| \mathcal K_{3}^>| \, \le\,  \frac{C}{\alpha^3} \sno (\mathbb N+1)^{3/2} \Upsilon_K^{>} \sno_\Fock  \int \D y\, f_{3,\alpha}(-y) (1 + \alpha y^2) \, \le\,  C_{\delta}\, \alpha^{- 7 }.
\end{align}
Similarly as in \eqref{eq: bound for K12<}, we also obtain
\begin{align}
| \mathcal K_{3}^<| \, \le\,  \frac{\sqrt 2}{\alpha^{3}} \int \D y \, \sno R_{7,y} \psi \otimes \Upsilon_K^{ < } \sno_\h\,  n_{\delta,\eta}(y)\,  \le\,  \frac{C}{\alpha^3} \int \D y\, f_{3,\alpha}(-y)(1 + \alpha y^2) n_{\delta,\eta}(y).
\end{align}
By Corollary \ref{cor: Gaussian for errors} and \eqref{eq: def f_3} it follows that 
\begin{align}\label{eq: bound for K3^<}
| \mathcal K_{3}^<|  \le C  \alpha^{-6 + 3 \delta}.
\end{align} 

This completes the analysis of $\mathcal K$. The proof of Proposition \ref{prop: bound for K} follows from combining \eqref{eq: bound for K111 lo1}, \eqref{eq: bound for K111 lo rest}, \eqref{eq: bound for K111 err}, \eqref{eq: bound for K112}, \eqref{eq: bound for K_12^>}, \eqref{eq: final bound for K_12^<}, \eqref{eq: bound for K2^>}, \eqref{eq: bound for K2^<}, \eqref{eq: bound for K3^>} and \eqref{eq: bound for K3^<}.
\end{proof}

\subsection{Concluding the proof of Proposition \ref{theorem: main estimate 2}\label{Sec: concluding the proof}}

Combining Propositions \ref{prop: bound for E}, \ref{prop: bound for G} and \ref{prop: bound for K}, we arrive at
\begin{align}\label{eq: conclusion bound}
\bigg| \frac{\mathcal E + \mathcal G + \mathcal K }{\mathcal N}  - \frac{\inf \sigma(\mathbb H_K)}{\alpha^2} + \frac{3}{2\alpha^2} \bigg| & \, \le \, C_\varepsilon\, \alpha^{\varepsilon} \bigg( \frac{ K^{-1/2}\alpha^{-5} + \sqrt K \alpha^{-6} }{\mathcal N}\bigg).
\end{align}
Now for $K\le \widetilde c \alpha $ we know from Proposition \ref{prop: norm bound} that $\mathcal N \ge C \alpha^{3}$ for some $C>0$, such that the right side is bounded by $C_\varepsilon \, \alpha^\varepsilon r(K,\alpha)$. It remains to show that one can replace $\alpha^{-2}\inf \sigma(\mathbb H_K)$ by $\alpha^{-2} \inf \sigma(\mathbb H_\infty)$ at the cost of an additional error. To this end, recall that $\inf \sigma (\mathbb H_K) = \langle \Upsilon_K |  \mathbb H_K  \Upsilon_K  \rangle_\Fock$ and use the variational principle to find
\begin{align}
 \lsp  \Upsilon_K  |  (\mathbb H_K - \mathbb H_\infty)  \Upsilon_K  \rsp_\Fock \, \le \, \inf \sigma(\mathbb H_K) - \inf \sigma(\mathbb H_\infty) \, & \le\,  \lsp  \Upsilon_\infty  | ( \mathbb H_K  -\mathbb H_\infty) \Upsilon_\infty \rsp_\Fock.
\end{align}
Writing
\begin{align}
\mathbb H_K - \mathbb H_\infty \, =\, \lsp \psi| \phi(h^1_{K,\cdot} - h^1_\cdot) R \phi(h^1_{K,\cdot}) \psi \rsp_\2 -\lsp \psi| \phi(h^1_{\cdot}) R \phi(h^1_{\cdot}-h^1_{K,\cdot})  \psi \rsp_\2 ,
\end{align}
and using Lemma \ref{lem: LY CM}, we can infer that for any $ \Psi \in \mathcal F$
\begin{align}
\I \lsp \Psi |  ( \mathbb H_K - \mathbb H_\infty ) \Psi \rsp_{\Fock} \I \le C K^{-1/2} \lsp \Psi |( \mathbb N_1 + 1 ) \Psi \rsp_\Fock.
\end{align}
By Corollary \ref{Cor: Upsilon estimates} we know that $\lsp \Upsilon_K |  (\mathbb N_1 + 1 ) \Upsilon_K \rsp_\Fock \le C$ for all $K\in (K_0,\infty]$ with $K_0$ large enough, and thus $|\inf \sigma(\mathbb H_K) - \inf \sigma(\mathbb H_\infty) | \le CK^{-1/2}$. In view of \eqref{eq: conclusion bound} and Lemma \ref{lem: energy identity} this completes the proof of Proposition \ref{theorem: main estimate 2}.

\section{Remaining Proofs \label{Sec: Remaining Proofs}}

\begin{proof}[Proof of Lemma \ref{lem: Hessian}] The form of the kernel is readily found using second order perturbation theory (we omit the details). (i) The lower bound $H^{\rm Pek} \ge 0 $ follows from \eqref{eq: definition of Hessian} whereas $H^{\rm Pek} \le 1 $ is a consequence of
\begin{align}
\lsp v | ( 1 - H^{\rm Pek} ) v \rsp_\2 & \, = 4 \bigg\| \int \D y \, v(y) R^{1/2} h_{\cdot }(y) \psi \bigg\|^2_\2 \label{eq 1-H}.
\end{align}
(ii) That $\text{Span} \{ \partial_i \varphi:i=1,2,3\} \subseteq \text{Ker}H^{\rm Pek}$ follows from translation invariance of the energy functional $\mathcal F$. To show equality we argue that there is a $\tau>0$ such that $\langle v | H^{\rm Pek} v \rangle_\2 \ge \tau \sno v \sno^2_\2$ for all $v\in L^2(\mathbb R^3)$ with $\langle v | \nabla \varphi\rangle_\2 =0$ (note that this also implies (iii)). For that purpose we quote \cite[Lemma 2.7]{FeliciangeliRS20} stating that there exists a constant $\tau >0$ such that
\begin{align}
\mathcal F(v ) - \mathcal F(\varphi ) \ge \tau \inf_{y \in \mathbb R^3} \sno v  - \varphi (\cdot -y) \sno^2_\2
\end{align}
for all $v  \in L^2(\mathbb R^3)$. (a key ingredient in the proof of this quadratic lower bound are the results about the Hessian of the Pekar energy functional \eqref{eq: electronic pekar functional} that were obtained in \cite{Lenzmann09}; see \cite{FeliciangeliRS20} for a detailed derivation). Combined with \eqref{eq: definition of Hessian} this implies
\begin{subequations}
\begin{align} \label{eq: gap of the hessian} 
 \langle v  | H^{\rm Pek}  v   \rangle_\2  &  \ge \tau  \lim_{\varepsilon \to 0 } \inf_{y \in \mathbb R^3} f_v (y,\varepsilon ) ,\\[0.5mm]
 f_v (y, \varepsilon ) & =   \sno v  \sno^2_\2 +  \varepsilon^{-2} \sno  \varphi -   \varphi (\cdot-y)  \sno^2_\2 +  2 \varepsilon^{-1}   \re  \langle v |  \varphi  - \varphi (\cdot-y)  \rangle_\2 .
\end{align}
\end{subequations}
Given any $v$ satisfying $\langle v | \nabla \varphi\rangle_\2 = 0$, we choose $y^* ( \varepsilon )$ such that $f_v  (y^*(\varepsilon), \varepsilon )$ is minimal. Furthermore, note that for every zero sequence $(\varepsilon_n)_{n\in \mathbb N}$ such that
\begin{align}
\lim_{n\to \infty} \sno  \varphi (\cdot-y^*(\varepsilon_n) ) - \varphi \sno_\2 > 0,
\end{align} 
it follows that $\lim_{n\to \infty}   f_v  (y^*(\varepsilon_n) ,  \varepsilon_n  ) = \infty$, and hence we can conclude that $|y^* ( \varepsilon )| \to 0$ as $\varepsilon \to 0$. To proceed, let $\eta(\varepsilon ) := \varphi - \varphi (\cdot  - y^*(\varepsilon) )$ and assume $\vert y^*(\varepsilon) \vert > 0$ (for if $ y^*(\varepsilon) = 0 $ it follows directly that $   f_v (y^*(\varepsilon) , \varepsilon ) = \sno v  \sno^2_\2 $). With this we can estimate
\begin{align}
f_v (y^*(\varepsilon) , \varepsilon )  &  \ge \sno v  \sno^2_\2  +  \varepsilon^{-2}   \sno \eta(\varepsilon) \sno^2_\2 -2  \varepsilon^{-1}  \vert \langle v  | \eta(\varepsilon)   \rangle_\2 \vert \notag\\
& \ge  \sno v \sno^2_\2  -  \vert \langle v | \frac{\eta(\varepsilon) }{\sno \eta(\varepsilon) \sno_\2} \rangle_\2 \vert^2 .
\end{align}
To bound the right side, write
\begin{align}
\eta(\varepsilon) (z)    =   \int_0^1 \D s\, (y^*(\varepsilon)   \nabla) \varphi  (z -  s  y^*(\varepsilon) ) 
\end{align}
and use, by dominated convergence, that
\begin{align}
\frac{\sno \int_0^1 \D s\, (y \nabla) \varphi  (\cdot - s y )  - ( y  \nabla ) \varphi  \sno_\2 }{\sno \int_0^1 \D s\, ( y \nabla ) \varphi  (\cdot - s y )  \sno_\2 } \to 0 \quad  \text{as} \quad \vert y \vert \to 0.
\end{align}
Combining the last statement with $\vert y^*(\varepsilon) \vert \to 0$ as $\varepsilon \to 0$ and $\langle v | \nabla \varphi \rangle_\2 = 0 $ we conclude that
\begin{align}
\lim_{\varepsilon \to 0 }  f_v (y^*(\varepsilon) , \varepsilon )  &  \ge \sno v \sno^2_\2   .
\end{align}
This completes the proof of items (ii) and (iii). Property (iv) follows from $H^{\rm Pek} \le (H^{\rm Pek})^{1/2}$ and $\text{Tr}_{L^2}(1-H^{\rm Pek}) <\infty$, see Lemma \ref{lem: regularized Hessian} for $K=\infty$.
\end{proof}

\begin{proof}[Proof of Lemma \ref{lem: regularized Hessian}] (i) The bound $H_K^{\rm Pek} \restriction \text{Ran} (\Pi_1) \le 1$ follows analogously to \eqref{eq 1-H} and $H_K^{\rm Pek} \restriction \text{Ran} (\Pi_0) = 0 $ holds by definition. The lower bound on $\text{Ran} (\Pi_1)$ is a consequence of $( H^{\rm Pek} - \tau) \restriction \text{Ran} (\Pi_1) \ge 0$ for some $\tau > 0$, see Lemma \ref{lem: Hessian}, in combination with
\begin{align}
\pm ( H^{\rm Pek} - H_K^{\rm Pek} )  \, \le \,  CK^{-1/2}.
\end{align}
To verify the latter, let $v\in \text{Ran}(\Pi_1)$, $\Pi_v =|v \rangle \langle v|$ and write
\begin{align}
\lsp v | ( H^{\rm Pek}_K - H^{\rm Pek} ) v \rsp_\2 & \, =\,   4 \int \D y\,   \re  \lsp \psi | \big( h_{K,\cdot}(y) -  h_{\cdot})(y) \big)  R  (\Pi_v h_{K,\cdot })(y)  \psi \rsp_\2 \notag \\
 & \quad \quad + 4  \int \D y\,  \re \lsp \psi | (\Pi_v h_{\cdot})(y)  R \big( h_{K,\cdot }(y) - h_{\cdot}(y) \big)  \psi \rsp_\2 .
\end{align} 
With Cauchy--Schwarz it follows that
\begin{align}
\I \lsp v | ( H^{\rm Pek}_K - H^{\rm Pek} ) v \rsp_\2 \I & \le 4 K^{1/2} \int \D y\, \sno R^{1/2} (h_{K,\cdot}(y) - h_{\cdot}(y) ) \PP \sno^2_{\op} \notag\\
& \hspace{-1.5cm} + 4 K^{-1/2} \int \D y\,  \big( \sno R^{1/2} (\Pi_v h_{K,\cdot})(y) P_\psi  \sno^2_{\op} + \sno R^{1/2}  (\Pi_v h_{\cdot}) (y) P_\psi \sno^2_{\op}  \big) ,
\end{align}
and from Corollary \ref{cor: X h Y bounds}, we obtain
\begin{align}
\I \lsp v | ( H^{\rm Pek}_K - H^{\rm Pek} ) v \rsp_\2 \I \le C K^{-1/2}.
\end{align}
(ii) On $\text{Ran}(\Pi_0)$ the inequality holds trivially, whereas on $\text{Ran}(\Pi_1)$, it follows from $\Theta_K\le 1$, $B_K^2   \le   \frac{1}{4}( \Theta_K^{-2} - 1 ) $, $\Theta_K^{-2} = ( 1 - (1-H_K^{\rm Pek}) )^{- 1/2}$ and the elementary inequality $( 1 - x )^{-1/2} \le 1 + \beta^{-3/2} x$ for all $x\in(0,1- \beta)$.\medskip

\noindent (iii) Here we use $\text{Tr}_{\textnormal{Ran}(\Pi_0)}(1-H_K^{\rm Pek}) = 3 $, write
\begin{align}\label{eq: bound for trace of TK1}
\text{Tr}_{\textnormal{Ran}(\Pi_1)}(1-H_K^{\rm Pek})  &  =   \int \D y\, \lsp \psi | h_{K,\cdot}^1 (y) R h^1_{K,\cdot}(y)  \psi\rsp_\2  =   \int \D y\, \sno R^{1/2} h^1_{K,\cdot}(y) \PP \sno^2_{\op} 
\end{align}
and apply Corollary \ref{cor: X h Y bounds}.\medskip

\noindent (iv) Since $1-H_K^{\rm Pek} = \Pi_0 + \Pi_1 (1-H_K^{\rm Pek}) \Pi_1 = \Pi_0 + 4 T_K$, cf. \eqref{eq: def of H:Pek:K:1} and \eqref{eq: def of H:Pek:K:0}, we can write
\begin{align}
& \textnormal{Tr}_{L^2}( (-i\nabla) ( 1-H_K^{\rm Pek}) (-i\nabla) ) =  \textnormal{Tr}_{L^2}\big( \nabla \Pi_0 \nabla \big) + 4 \textnormal{Tr}_{L^2}\big( \nabla T_K \nabla ).
\end{align}
Using the explicit form of $\Pi_0$, one shows that the first term is given by
\begin{align}
\textnormal{Tr}_{L^2}\big( \nabla \Pi_0 \nabla ) =  \frac{3}{\sno \nabla\varphi\sno^2_\2} \sum_{j=1}^3\textnormal{Tr}_{L^2}\big( \nabla |\nabla_j \varphi\rangle \langle \nabla_j \varphi| \nabla \big) \le 3 \frac{\sno \Delta \varphi\sno^2_\2}{\sno \nabla\varphi\sno^2_\2},
\end{align}
which is finite since $\Delta\varphi\in L^2$. For the second term it follows from a short computation that
\begin{align}
\textnormal{Tr}_{L^2}\big( \nabla T_K \nabla ) & \, = \, \int \D y\, \lsp \psi | [\nabla ,h_{K,\cdot}^1 (y)  ] R [ \nabla ,  h^1_{K,\cdot}  (y) ] | \psi\rsp_\2 .
\end{align}
Using the Cauchy--Schwarz inequality and $\sno \nabla \psi\sno_\2 + \sno R^{1/2} \sno_\op + \sno  R^{1/2} \nabla \sno_{\op}<\infty $, see Lemmas \ref{lemma: props_peks} and \ref{lem: bound for R}, we can estimate the last expression by
\begin{align}
\int \D y \, \sno R^{1/2} [\nabla, h_{K,\cdot}^1 (y)] \psi \sno^2_\2  \, & \le \, C \int \D y \, \big( \sno h_{K,\cdot}^1 (y) \psi \sno^2_\2 + \sno h_{K,\cdot}^1 (y) \nabla \psi \sno^2_\2 \big) \notag \\
& \le \, C \int \D y \, | h_{K,0}^1(y)|^2 \, \le \, C \sno h_{K,0}\sno^2_\2 \, = \, C K.
\end{align}
This completes the proof of the lemma.
\end{proof}

\begin{proof}[Proof of Lemma \ref{prop: diagonalization of HBog}] We recall that $H_K^{\rm Pek}\restriction \text{Ran}(\Pi_0) = 0$ and $T_K = \frac{1}{4} (H^{\rm Pek}_K - \Pi_1) $, and set $S_K= \frac{1}{2} (\Pi_1 +H^{\rm Pek}_K )$. For $(u_n)_{n\in \mathbb N }$ an orthonormal basis of $\text{Ran}(\Pi_1)$, we further set $a_n = a(u_n)$ and use this to write the Bogoliubov Hamiltonian as
\begin{align}\label{eq: Bogo Ham Appendix}
\mathbb H_K  & \, =\,   \sum_{n,m=1}^\infty \bigg(  \lsp u_n | S_K u_m  \rsp_\2 a_n^\dagger a_m + \big( \lsp u_n | T_K \overline{u_m} \rsp_\2  a_n^\dagger  a_m^\dagger + \text{h.c.} \big)  \bigg) + \text{Tr}_{L^2}(T_K).
\end{align}
Applying the transformation \eqref{eq: def of U}, a straightforward computation leads to 
\begin{align}
\mathbb U_K \mathbb H_K \mathbb U^\dagger_K & \,  =\,  \sum_{n,m=1}^\infty \Big(  \lsp u_n | ( A_K S_K A_K + B_KS_KB_K + 4 A_KT_KB_K ) u_m  \rsp_\2  a^\dagger_n a_m   \notag \\
& \quad \quad \quad  +  \big( \lsp u_n | ( A_K S_K B_K + A_K T_K A_K  + B_K T_K B_K  ) \overline{u_m} \rsp_\2  a^\dagger_n a^\dagger_m + \text{h.c.}\big) \Big) \notag \\[2.5mm]
& \quad \quad + \text{Tr}_{\text{Ran}(\Pi_1)}\big( T_K +  B_K S_K B_K  +  2  A_K T_K B_K ).
\end{align}
The statement of the lemma now follows from
\begin{subequations}
\begin{align}
\Pi_1 ( A_K S_K A_K + B_K S_K B_K + 4  A_K T_K B_K )\Pi_1 & \, =\,  \sqrt{H_K^{\rm Pek}}  \\[0.5mm]
\Pi_1 ( A_K S_K B_K +  A_K T_K A_K + B_K T_K B_K )\Pi_1 & \, =\,   0 \\[0mm]
\Pi_1 ( T_K + B_K S_K B_K + 2 A_K T_K B_K  )\Pi_1 & \, =\, \frac{1}{2} \big( \sqrt{H_K^{\rm Pek}} - \Pi_1 \big).
\end{align}
\end{subequations}
\end{proof}


\begin{proof}[Proof of Lemma \ref{lem: bound for w1 and w0}]
To bound $ \sno  w^1_{P,y}  \sno^2_\2$ we expand
\begin{align} 
w_{P,y}^1  \, =\,  \Pi_1 (1 - e^{ - y\nabla } ) (\varphi + i \xi_P )  & \, =\,  \int_0^1 \D s_1 \int_0^{s_1}\, \D s_2  \, \Pi_1 e^{- s_2 y\nabla } (y\nabla )^2 \varphi \notag \\
& \quad + \frac{i}{\alpha^2 M^{\rm LP}} \int_0^1 \D s \, \Pi_1 e^{-sy\nabla} (y\nabla)(P \nabla) \varphi ,\label{expansion for w perp}
\end{align}
where we used $\Pi_1 (y\nabla)\varphi = 0$. Thus, since $\Delta \varphi \in L^2$, we easily arrive at
\begin{align} 
 \sno w_{P,y}^1 \sno^2_\2  & \, \le\,  C \big( y^4 + \alpha^{-4} y^2 P^2 \big) \label{eq: bound for w perp} 
\end{align}
for some constant $C>0$, and with $|P|\le \alpha c$ we obtain the stated estimated. The bound for $\sno \widetilde w_{P,y}^1 \sno^2_\2$ follows from
\begin{align}
\sno \widetilde w_{P,y}^1 \sno^2_\2  = \sno \Theta_K \re( w_{P,y}^{1} ) \sno^2_\2  + \sno \Theta_K^{-1} \im (w_{P,y}^{1} ) \sno^2_\2 \le C \sno w_{P,y}^1 \sno^2_\2,
\end{align}
where we used that $\Theta_K$ is real-valued and satisfies
\begin{align}\label{eq: bound for Theta_K}
0 \, < \, \beta\, \le\, \Theta_K^2 \, \le \, 1 
\end{align}
when restricted to $\text{Ran}(\Pi_1)$; see Lemma \ref{lem: regularized Hessian}. To bound  $\sno  w^0_{P,y}  \sno^2_\2$ we use 
\begin{align}
   \sno  w^0_{P,y}  \sno^2_\2 = \sno  w^0_{0,y}  \sno^2_\2 + \sno \Pi_0(1-e^{-y\nabla})\xi_P\sno^2_\2,
\end{align} 
since $\varphi$, $\xi_P$ and $\Pi_0$ are all real-valued. Expanding $1-e^{-y\nabla}$ as in \eqref{expansion for w perp}, it is easy to conclude that $\sno \Pi_0(1-e^{-y\nabla})\xi_P\sno^2_\2 \leq CP^2y^2\alpha^{-4}$. Using the explicit form of $\Pi_0$ and $\langle \nabla\varphi |  \varphi \rangle_\2 = 0$, we can write 
\begin{align}
  \sno  w^0_{0,y}  \sno^2_\2 \, =\,  \frac{3}{\sno \nabla \varphi \sno^2_\2}\sum_{i=1}^3 \I \lsp \nabla_i \varphi|e^{-y\nabla} \varphi\rsp_\2 \I^2.
\end{align} 
Using the Fourier representation and rotation invariance, we have 
\begin{align}
  \I \lsp \nabla_i \varphi|e^{-y\nabla} \varphi\rsp_\2 \I \, =\, \bigg| \int p_i |\hat \varphi(p)|^2\sin (py)~ \D y \bigg|.
\end{align}
By the elementary inequality $|\sin z-z|\leq Cz^3$, the formula $\sno (y\nabla)\varphi\sno_\2^2= 2 \lambda y^2$ and the finiteness of $\|\Delta \varphi\|_\2,$ we conclude that
\begin{align}\label{eq: bound for w 0}
\big| \sno w_{P,y}^0 \sno^2_\2 - 2 \lambda y^2 \big|  \, \le \, C \big( y^4 + y^6 + \alpha^{-4} y^2 P^2 \big). 
\end{align}
To prove the last bound, we use
\begin{align}\label{eq: norm of tilde wP}
\sno \widetilde w_{P,y} \sno^2_\2 \, = \, \sno w^{0}_{P,y} \sno^2_\2 + \sno \Theta_K \re( w_{P,y}^{1} ) \sno^2_\2 + \sno \Theta_K^{-1} \im (w_{P,y}^{1} ) \sno^2_\2,
\end{align}
and hence with \eqref{eq: bound for Theta_K}, 
\begin{align} \label{eq: upper bound for w tilde}
 \beta \sno  w^1_{P,y}  \sno^2_\2  \le \sno \widetilde w_{P,y} \sno^2_\2 -\sno w^0_{P,y} \sno^2_\2  &\,   \le \,   \beta^{-1}  \sno  w^1_{P,y}  \sno^2_\2 .
\end{align}
The desired bound now follows from \eqref{eq: bound for w perp}  and \eqref{eq: bound for w 0}.
\end{proof}


\begin{proof}[Proof of Lemma \ref{lem: Gaussian lemma}] From Lemma \ref{lem: bound for w1 and w0}, we have
\begin{align}
\big| \sno \widetilde w_{P,y} \sno^2_\2 - 2 \lambda y^2 \big|  &\,   \le \,   C (\alpha^{-2} y^2 + y^4 + y^6) \, \le\,  C \frac{y^2}{\alpha} \quad \text{for all}\quad \frac{|P|}{\alpha} \le c , \ y^2 \le \alpha^{-1}.
\end{align}
Hence there is a constant $\mu>0$ such that for all $y^2 \le \alpha^{-1}$ the weight function \eqref{eq: definition of F} satisfies
\begin{subequations}
\begin{align}\label{eq: upper bound for n}
  n_{\delta,\eta  }(y) & \, \le\,   \exp( - (\lambda \eta \alpha^{2(1-\delta)}-\mu \alpha^{-2\delta+1}  )y^2 )  \\[1mm]
 n_{\delta,\eta  }(y) & \, \ge\,  \exp(- (\lambda \eta \alpha^{2(1-\delta)} + \mu \alpha^{-2\delta+1} ) y^2 ) .\label{eq: lower bound for n}
\end{align}
\end{subequations}
In the remainder let us abbreviate $f_n(y) = |y|^n g(y)$ and $Z(y) = | n_{\delta,\eta}(y) - e^{-\lambda \eta \alpha^{2(1-\delta)} y^2} | $. We then decompose the integral into
\begin{align}\label{eq: decomposition integral}
\int \D y\, f_n(y) Z(y)  \, =\,  \int_{B_\alpha} \D y \, f_n(y) Z(y) + \int_{B^c_\alpha } \D y\, f_n(y) Z(y) 
\end{align}
with $B_\alpha = \{ y \in \mathbb R^3 : y^2 \le \alpha^{-1} \} $. The bounds \eqref{eq: upper bound for n} and \eqref{eq: lower bound for n} imply that
\begin{align}
|Z(y)|\leq  e^{-\lambda \eta \alpha^{2(1-\delta)}}\left(e^{\mu \alpha^{-2\delta+1}y^2}-1\right) \quad \forall y\in B_\alpha
\end{align}
and thus by $|e^z-1| \le z e^z$ for $z >0$, we obtain
\begin{align}
\int_{B_\alpha} \D y \, f_n(y) Z(y)\, \le \,   \mu \alpha^{-2\delta+1} \int \D y\, f_n(y) y^{2} e^{ - ( \eta \lambda - \mu \alpha^{-1} ) \alpha^{2(1-\delta)} y^2}  .
\end{align}
The last expression is further bounded by
\begin{align}
\int \D y \, f_n (y) y^2 e^{ - ( \eta \lambda - \mu \alpha^{-1} ) \alpha^{2(1-\delta)} y^2} & \, \le\,  \sno g \sno_{\su} \int \D y \, |y|^{n+2} e^{ - ( \eta \lambda - \mu \alpha^{-1} )  \alpha^{2(1-\delta)} y^2}  \\ \notag
& \, =\, \frac{C_n\sno g \sno_{\su}}{\alpha^{(5+n)(1-\delta)}}\left(\eta\lambda-\mu\alpha^{-1}\right)^{-(n+5)/2}
\end{align}
and since the resulting expression is uniformly bounded in $\eta\ge \eta_0$ and $\alpha$ large, we get
\begin{align}
\int_{B_\alpha} \D y\, f_n (y) Z(y) \, \leq \,  C_n \frac{\sno g\sno_{\su} }{ \alpha^{(4+n)(1-\delta) + \delta}} .
\end{align}
To bound the second term in \eqref{eq: decomposition integral}, we estimate
\begin{align}
\int_{B^c_\alpha } \D y\,  f_n(y) Z(y) \, \le\,  \int_{B^c_\alpha } \D y\, f_n(y) n_{\delta,\eta}(y) + e^{- \lambda\eta \alpha^{-2\delta + 1 }} \int \D y\, f_n(y).
\end{align}
To see that the first summand is exponentially small as well, we use \eqref{eq: norm of tilde wP}, \eqref{eq: bound for Theta_K} and $\re(w_{P,y}^i) = \Pi_i\re( w_{P,y}) = \Pi_i\re( w_{0,y}) $ for $i=0,1$,
\begin{align}
\sno \widetilde w_{P,y} \sno^2_\2 &  \ge \sno\re( w_{P,y}^0) \sno^2_\2 + \beta \sno \re ( w^1_{P,y}) \sno^2_\2 \ge  \beta\sno \mathrm{Re} ( w_{0,y} ) \sno^2_\2 =  \beta \sno (1-e^{-y\nabla}) \varphi \sno^2_\2,
\end{align} 
and hence
\begin{align}
n_{ \delta,\eta  }(y) \, \le\,  \exp\Big( -  \eta  \beta  \alpha^{2(1-\delta)} q(y) \Big) \quad \text{with} \quad q(y) \, =\,  \frac{1}{2}  \sno (1-e^{-y\nabla} ) \varphi \sno^2_\2.
\end{align}
Since $\varphi$ is real-valued, we have $\langle \varphi |e^{-y\nabla}|\varphi\rangle_\2 =\langle \varphi |e^{y\nabla}|\varphi\rangle_\2=(\varphi\ast\varphi)(y)$ and thus
\begin{align}
  q(y)=\sno \varphi \sno^2_\2 -(\varphi\ast\varphi)(y).
\end{align}
Recall that, as shown in \cite{Lieb1977}, the electronic Pekar minimizer $\psi$ is radial and non-increasing and hence $\varphi$,  cf. \eqref{eq: optimal phonon mode}, is radial and non-increasing as well, as convolutions of radial non-increasing functions are themselves radial non-increasing functions. Consequently, $q(y)$ is radial and monotone non-decreasing, and thus $q(y)\ge q(y')$ for all $y\in B_{\alpha}^c$, $y' \in B_\alpha$. On the other hand, by a simple computation, using the regularity of $\varphi$, one finds that $q(y)\geq  C_0y^2$ for some $C_0>0$ and all $|y|$ small enough, and thus $q(y)\geq C_0\alpha^{-1}$ for all $y\in B_\alpha^c$ and $\alpha$ large. Therefore
\begin{align}
\int_{B^c_\alpha } \D y\, f_n (y) n_{ \delta,\eta  }(y) & \, \le\,   \int_{B^c_\alpha } \D y\, f_n (y) e^{ - \eta \beta \alpha^{2(1-\delta)} q(y)} \notag\\
& \, \le \,  e^{ - C_0   \eta \beta  \alpha^{2(1-\delta)-1} }  \int \D y\, f_n (y) \, \le\,  e^{-d\alpha^{-2 \delta + 1  } }  \int \D y\, f_n (y)
\end{align}
for some $d>0$, which completes the proof of the lemma.
\end{proof}


\begin{proof}[Proof of Lemma \ref{lem: bounds for P_f}] 
Let $p = -i\nabla$. By a straightforward computation using the transformation property \eqref{eq: def of U}, we arrive at 
\begin{align}\label{eq: U Pf Omega identity}
  \mathbb{U}_K P_f \mathbb{U}_K^{\dagger}  \Omega = \sum_{n}  a^{\dagger}(A_K u_n)a^{\dagger}(B_K p \overline{u_n} )  \Omega  +\textnormal{Tr}_{L^2}(B_KpB_K)  \Omega 
\end{align}
for some orthonormal basis $(u_n)_{n\in \mathbb N}$ of $L^2(\mathbb{R}^3)$. That $B_KpB_K$ is trace-class can be seen via
\begin{align}
\text{Tr}_{L^2}{| B_K p B_K | }\, \le\,  \sno B_K\sno_{\HS} \, \sno p B_K \sno_{\HS} \, \le \, C K,
\end{align}
where the second step follows from Lemma \ref{lem: regularized Hessian}, implying $\sno B_{K}\sno_{\HS}\le C$, and
\begin{align}\label{eq: proof of pBK norm}
 \sno p B_K \sno_{\HS}^2 \, =  \, \textnormal{Tr}_{L^2}(p  B_K B_K  p) \, \le \, \textnormal{Tr}_{L^2}(p  (1 - H_K^{\rm Pek}) p) \, \le \, C K .
\end{align}
By rotation invariance $\textnormal{Tr}_{L^2}(B_KpB_K) = 0$. The first term in \eqref{eq: U Pf Omega identity}, on the other hand, is seen to be a two-particle wave function $\Phi_K$ given by 
\begin{align}
   \Phi_K(x,y)=\frac{1}{\sqrt{2}}\left(A_KpB_K+B_KpA_K\right)(x,y).
\end{align}
Thus 
\begin{align}
\lsp \Upsilon_K | (P_f)^2  \Upsilon_K \rsp_\Fock \, & = \, \frac{1}{2}  \|A_KpB_K+B_KpA_K\|_{\HS}^2 \, \leq  \,  2 \sno A _K  \sno_{\op}^2 \sno p B_K \sno_{\HS}^2 \, \le \, C K ,
\end{align} 
where we invoked again \eqref{eq: proof of pBK norm}.
\end{proof}

\hspace{3mm}

\noindent\textbf{Acknowledgments}. Financial support through the European Research Council (ERC) under the European Union’s Horizon 2020 research and innovation programme grant agreement No.
694227 (R.S.) and the Maria Skłodowska-Curie grant agreement No. 665386 (K.M.) is gratefully
acknowledged.

\end{spacing}


\begin{thebibliography}{D}

\makeatletter
\renewcommand{\@biblabel}[1]{[#1]\hfill}
\makeatother

\footnotesize{

\bibitem{DevreeseA2010} 
A.S.\ Alexandrov and J.T.\ Devreese. \textit{Advances in Polaron Physics}. Springer. (2010)

\bibitem{Allcock65} G.R. Allcock. On the polaron rest energy and effective mass. \emph{Adv. Phys. 5:20, 412--451}. (1956)

\bibitem{Betz2021} V. Betz and S. Polzer. A functional central limit theorem for polaron path measures. \textit{arXiv:\href{https://arxiv.org/abs/2106.06447}{2106.06447}}. Preprint. (2021)

\bibitem{Betz2022} V. Betz and S. Polzer. Effective mass of the Polaron: a lower bound. \textit{arXiv:\href{https://arXiv.org/abs/2201.06445}{2201.06445}}. Preprint. (2022)

\bibitem{Born1954}
M. Born and K. Huang. Dynamical theory of crystal lattices. \emph{Oxford University Press}. (1954) 

\bibitem{Born27} M. Born and R. Oppenheimer. Zur Quantentheorie der Molekeln. \emph{Ann. Phys. (Leipzig) 84, 457--484}. (1927)

\bibitem{Bossmann2019} L. Boßmann, S. Petrat, P. Pickl and A. Soffer. Beyond Bogoliubov Dynamics. \emph{Pure Appl. Anal. 3(4), 677--726}. (2021)

\bibitem{JonasPHD17} J. Dahlb\ae k. Spectral analysis of large particle systems. PhD Thesis. Aarhus University. (2017)

\bibitem{Donsker1983}
M.D.\ Donsker and S.R.S.\ Varadhan. Asymptotics for the polaron. \textit{Comm. Pure Appl. Math.\ 36, 505--528}. (1983)

\bibitem{DybalskyS2020}
W. Dybalski and H. Spohn. Effective mass of the polaron -- revisited. \textit{Ann. Henri Poincar\'{e} 21, 1573--1594}. (2020)

\bibitem{FeliciangeliRS20} 
D.\ Feliciangeli,\ S.\ Rademacher and R.\ Seiringer. Persistence of the spectral gap for the Landau--Pekar equations. \textit{Lett. Math. Phys 111, 19}. (2021)

\bibitem{FeliciangeliRS21} D. Feliciangeli, S. Rademacher and R. Seiringer. 
The effective mass problem for the Landau--Pekar equations. \textit{J. Phys. A Math. Theor. 55 015201}. (2022)

\bibitem{FeliciangeliS21} 
D.\ Feliciangeli and R.\ Seiringer. The strongly coupled polaron on the torus: quantum corrections to the Pekar asymptotics. \textit{Arch. Rat. Mech. Anal. 242, 1835--1906.} (2021)

\bibitem{FrankG2017} R.L.\ Frank and Z.\ Gang. Derivation of an effective evolution equation for a strongly coupled polaron. \textit{Anal. PDE 10, 379--422}. (2017)

\bibitem{FrankS2014} R.L.\ Frank\ and\ B.\ Schlein. Dynamics of a strongly coupled polaron. \textit{Lett.\ Math.\ Phys. 104, 911--929}. (2014) 

\bibitem{FrankS2021} R.L.\ Frank and R.\ Seiringer. Quantum corrections to the Pekar asymptotics of a strongly coupled polaron. \textit{Comm. Pure Appl. Math. 74(3), 544–588}. (2021)

\bibitem{Froehlich1933} H. Fr\"ohlich. Theory of electrical breakdown in ionic crystals. \emph{ Proc. R. Soc. Lond. A 160, 230--241}. (1937)

\bibitem{Froehlich1954} H.\ Fr\"ohlich. Electrons in lattice fields. \textit{Adv. in Phys. 3, 325--362}. (1954)

\bibitem{Froehlich1974} J.\ Fr\"ohlich. Existence of dressed one-electron states in a class of persistent models. \textit{Fort. Phys. 22, 159--198}. (1974) 

\bibitem{Gerlach03} B. Gerlach, F. Kalina, and M. Smondyrev. On the LO-polaron dispersion in D dimensions. \emph{phys. stat. sol. (b) 237, 204} (2003).

\bibitem{Gerlach91} B. Gerlach and H. L\"owen. Analytical properties of polaron systems or: Do polaronic phase transitions exist or not? \emph{Rev. Mod. Phys. 63, 63}. (1991)

\bibitem{Gerlach08} B. Gerlach and M. Smondyrev. Upper and lower bounds for the large polaron dispersion in 1, 2, or 3 dimensions. \emph{Phys. Rev. B 77(17)}. (2008)

\bibitem{Griesemer2017} M.\ Griesemer. On the dynamics of polarons in the strong-coupling limit. \textit{Rev. Math.\ Phys. 29, 10}. (2017) 

\bibitem{Gross76} E.P. Gross. Strong coupling polaron theory and translational invariance. \emph{Ann. Phys. 99, 1--29}. (1976)

\bibitem{Hoehler1961} G. H\"ohler. The polaron model. In \textit{Lectures on field theory and the many-body problem, edited by E.R. Caianiello}. Academic Press (1961)

\bibitem{Landau1933} L.D. Landau. \"Uber die Bewegung der Elektronen in Kristallgitter. \emph{Phys. Z. Sowj. 3, 644--645.} (1933)

\bibitem{Landau1948} L.D. Landau and S.I. Pekar. Effective mass of a polaron. \textit{J. Exp. Theor. Phys, 18, 419--423}. (1948)

\bibitem{LeeLowPines} T.D. Lee, F. Low and D. Pines. The motion of slow electrons in a polar crystal. \emph{Phys. Rev. 90, 297}. (1953)

\bibitem{Lenzmann09}
E. Lenzmann. Uniqueness of ground states for pseudorelativistic Hartree equations. \textit{Anal. PDE 2, 1--27}. (2009)

\bibitem{LeopoldMRSS2020} 
N.\ Leopold,\ D. Mitrouskas, S.\ Rademacher,\ B.\ Schlein and R.\ Seiringer. Landau--Pekar equations and quantum fluctuations for the dynamics of a strongly coupled polaron. \emph{Pure Appl. Anal. 3(4), 653--676}. (2021)

\bibitem{LeopoldRSS2019} 
N.\ Leopold,\ S.\ Rademacher,\ B.\ Schlein and R.\ Seiringer. The Landau--Pekar equations: Adiabatic theorem and accuracy. \emph{Anal. PDE 14, 2079–2100} (2021)

\bibitem{Lieb1977} E.H. Lieb. Existence and uniqueness of the minimizing solution of Choquard's nonlinear equation. \textit{St. Appl. Math. 57, 93-105}. (1977)

\bibitem{LiebSeiringer2014} E.H. Lieb and R. Seiringer. Equivalence of two definitions of the effective mass of a polaron. \emph{J. Stat. Phys. 154, 51--57}. (2014)

\bibitem{Lieb2020} E.H. Lieb and R. Seiringer. Divergence of the effective mass of a polaron in the strong coupling limit. \textit{J. Stat. Phys. 180, 23--33}. (2020)

\bibitem{Lieb1997} E.H. Lieb and L.E. Thomas. Exact ground state energy of the strong-coupling polaron. \textit{Comm. Math. Phys. 183(3), 511--519}. (1997)

\bibitem{Lieb1958} E.H. Lieb and K. Yamazaki. Ground-state energy and effective mass of the polaron. \emph{Phys. Rev. 111, 728}. (1958)

\bibitem{Mitra1987}
T.K. Mitra, A. Chatterjee and S. Mukhopadhyay. Polarons. \textit{Physic Reports 153, 91--207}. (1987)

\bibitem{Mitrouskas21} D. Mitrouskas. A note on the Fr\"ohlich dynamics in the strong coupling limit. \emph{Lett. Math. Phys. 111, 45}. (2021)

\bibitem{Miyake1976} S.J. Miyake. The ground state of the optical polaron in the strong-coupling case. \emph{J. Phys. Soc. Jpn. 41, 747--752}. (1976)

\bibitem{Moeller2006}
J.S. M\o ller. The polaron revisited. \textit{Rev. Math. Phys. 18, 485--517}. (2006)

\bibitem{Moroz2017}  V. Moroz and J.V. Schaftingen. A guide to the Choquard equation.
\emph{J. Fixed Point Theory Appl 19, 773--813}. (2017)

\bibitem{Mukherjee19} C. Mukherjee and S.R.S. Varadhan. Identification of the polaron measure i: Fixed coupling regime and the central limit theorem for large times. \emph{Comm. Pure Appl. Math. 73(2):350--383.} (2019)

\bibitem{Mysliwy2021}
K. My\' sliwy and R. Seiringer. Polaron models with regular interactions at strong coupling. \textit{J. Stat. Phys. 186:5.} (2022)

\bibitem{Nagy1989} P. Nagy. A note to the translationally-invariant strong coupling theory of the polaron. \textit{Czech. J. Phys. B 39, 353–356}. (1989)

\bibitem{Pekar46} S.I. Pekar. \emph{Zhurnal Eksperimentalnoi I Teoreticheskoi
Fiziki 16, 341} (1946). 

\bibitem{Pekar54} S.I. Pekar. Untersuchung \"uber die Elektronentheorie der Kristalle. \emph{Berlin, Akad. Verlag.} (1954)

\bibitem{JPS2007}
J.\,P.\ Solovej.
Many Body Quantum Mechanics. Lecture notes. \url{http://web.math.ku.dk/~solovej/MANYBODY/mbnotes-ptn-5-3-14.pdf}. (2014)

\bibitem{Spohn1987}
H. Spohn. Effective mass of the polaron: a functional integral approach. \textit{Ann. Phys. 175, 278--318}. (1987)

\bibitem{Spohn1988}
H. Spohn. The polaron at large momentum. \textit{J. Phys. A: Math. Gen. 21 1199--1211}. (1988)

\bibitem{Tjablikow54} S.W. Tjablikow. Adiabatische Form der St\"orungstheorie im Problem der Wechselwirkung eines Teilchens mit einem gequantelten Feld. \emph{Abhandl. Sowj. Phys. 4, 54--68}. (1954)

\bibitem{Whitfield65} G. Whitfield and R.D. Puff. Weak-coupling theory of the Polaron energy-momentum relation. \emph{Phys. Rev. 139, A338} (1965)

}

\end{thebibliography}
\end{document}